\documentclass[a4paper,twocolumn,notitlepage,superscriptaddress,floatfix,aps,pra,10pt]{revtex4-2}

\usepackage{array} 
\usepackage[english]{babel}
\usepackage[utf8]{inputenc}
\usepackage{graphicx}
\usepackage{amsmath}
\usepackage{amsthm}
\usepackage{amssymb}
\usepackage{dsfont}
\usepackage{physics}
\usepackage{mathrsfs} 
\usepackage{mathtools} 
\usepackage{tikz}
\usetikzlibrary{arrows}
\usepackage{tcolorbox}
\usepackage{enumitem}
\newlist{steps}{enumerate}{1}
\setlist[steps, 1]{label = Step \arabic*:}
\usepackage[colorinlistoftodos,prependcaption]{todonotes}
\usepackage{booktabs}
\usepackage{thmtools,thm-restate}
\usepackage{accents}
\usepackage{subcaption}
\usepackage[percent]{overpic}

\usepackage[pass]{geometry} 


\usepackage{setspace}       

\usepackage[framemethod=TikZ]{mdframed}

\AtBeginDocument{
	\heavyrulewidth=.08em
	\lightrulewidth=.05em
	\cmidrulewidth=.03em
	\belowrulesep=.65ex
	\belowbottomsep=0pt
	\aboverulesep=.4ex
	\abovetopsep=0pt
	\cmidrulesep=\doublerulesep
	\cmidrulekern=.5em
	\defaultaddspace=.5em
}

\makeatletter
\def\label#1{%
  \@bsphack
  \begingroup
  \UseHookWithArguments{label}{1}{#1}%
  \protected@write\@auxout{}{%
    \string\newlabel{#1}{{\@currentlabel}{\thepage}{\@currentlabelname}{\@currentHref}{\@kernel@reserved@label@data}}%
  }%
  \endgroup
  \@esphack
}
\makeatother

\usepackage[hypertexnames=false]{hyperref}
\usepackage{cleveref}
\usepackage{xurl} 
\hypersetup{
colorlinks = true,
citecolor= blue,
urlcolor= blue,
linkcolor = blue
}

\makeatletter
\def\maketitle{
\@author@finish
\title@column\titleblock@produce
\suppressfloats[t]}
\makeatother

\makeatletter
\newcommand{\makesupplementtitle}[1]{%
  \begingroup
    \ltx@footnote@pop
    \def\@mpfn{mpfootnote}%
    \def\thempfn{\thempfootnote}%
    \c@mpfootnote\z@
    \let\@makefnmark\frontmatter@makefnmark
    \frontmatter@setup
    \thispagestyle{titlepage}
    {\parindent\z@\centering
      \frontmatter@title@above
      \frontmatter@title@format
      #1\par
      \frontmatter@title@below
    }%
    \groupauthors@sw{\frontmatter@author@produce@group}{\frontmatter@author@produce@script}%
    \frontmatter@RRAPformat{\expandafter\produce@RRAP\expandafter{\@date}}%
    \par
    \frontmatter@finalspace
  \endgroup
}
\makeatother

\usepackage{ragged2e}

\makeatletter
\long\def\@makecaption#1#2{%
  \par\vskip\abovecaptionskip
  \begingroup
    \small\rmfamily
    \justifying
    #1. #2\par
  \endgroup
  \vskip\belowcaptionskip
}
\makeatother

\makeatletter
\newcommand{\pagefootnote}[1]{%
  \begingroup
  \stepcounter{footnote}%
  \protected@xdef\@thefnmark{\thefootnote}%
  \@footnotemark
  \@footnotetext{#1}%
  \endgroup
}
\makeatother

\makeatletter
\renewcommand{\p@subfigure}{\thefigure.}
\makeatother
\crefname{subfigure}{figure}{figures}
\Crefname{subfigure}{Figure}{Figures}

\usepackage[normalem]{ulem} 
\usepackage[dvipsnames]{xcolor}

\newtheorem{theorem}{Theorem}[section]
\newtheorem{thm}[theorem]{Theorem}

\newtheorem{remark}[theorem]{Remark}
\newtheorem{cor}[theorem]{Corollary}
\newtheorem{lem}[theorem]{Lemma}

\newtheorem{ass}[theorem]{Assumption}

\theoremstyle{definition} 
\newtheorem{definition}[theorem]{Definition}
\newtheorem{example}[theorem]{Example}

\numberwithin{equation}{section}

\newcounter{algsubstate}


\newcommand{\HH}{\mathcal{H}}
\newcommand{\Id}{\mathrm{Id}}
\newcommand{\id}{\mathds{1}}
\newcommand{\CC}{\mathbb{C}}
\newcommand{\ZZ}{\mathbb{Z}}

\newcommand{\NN}{\mathbb{N}}

\newcommand{\RR}{\mathbb{R}}

\newcommand{\BB}{\mathcal{B}}
\newcommand{\MM}{\mathcal{M}}
\newcommand{\UU}{\mathcal{U}}
\newcommand{\EE}{\mathcal{E}}

\newcommand{\Ss}{\mathcal{S}}

\newcommand{\VV}{\mathcal{V}}

\newcommand{\cC}{\mathcal{C}}
\newcommand{\cM}{\mathcal{M}}

\newcommand{\Eeff}{\mathbf{E}^{\mathrm{eff}}}

\newcommand{\poly}{\mathsf{poly}}

\newcommand{\negl}{\mathsf{negl}}

\newcommand{\cone}{\mathsf{cone}}
\newcommand{\conv}{\mathsf{conv}}
\newcommand{\supp}{\mathsf{supp}}

\newcommand{\Span}{\mathsf{span}}
\newcommand{\Pos}{\mathrm{Pos}}
\newcommand{\Herm}{\mathrm{Herm}}

\newcommand{\ptr}[2]{\operatorname{tr}_{#1}\left[#2\right]}
\newcommand{\term}[1]{\textup{\textit{#1}}} 

\newcommand{\CompMaxDiv}{\overset{\smash{\text{\tiny\hspace{0.3em}\raisebox{-0.3ex}{c}}}}{D}_{\mathrm{max}}}

\newcommand{\CompDiv}{ \accentset{\text{\tiny\hspace{0.1em}\raisebox{0.05ex}{c}}}{D}}

\newcommand{\CompMin}{\overset{\smash{\text{\tiny\hspace{0.3em}\raisebox{-0.3ex}{c}}}}{H}_{\mathrm{min}}}
\newcommand{\CompDownArrowMin}{\overset{\smash{\text{\tiny\hspace{0.3em}\raisebox{-0.3ex}{c}}}}{H}_{\infty}^{\downarrow}}

\newcommand{\CompPGuess}{\accentset{\text{\tiny\hspace{-0.1em}\raisebox{0.05ex}{c}}}{p}_{\mathrm{guess}}}
\newcommand{\CompQCorr}{\accentset{\text{\tiny\hspace{-0.1em}\raisebox{0.05ex}{c}}}{q}_{\mathrm{corr}}}

\usepackage{cancel}
\usepackage{stackengine}
\usepackage{scalerel}

\newcommand{\sorbonne}{Sorbonne Université, CNRS, LIP6, 4 Place Jussieu, 75005 Paris, France.}
\newcommand{\weizman}{The Center for Quantum Science and Technology, Department of Physics of Complex Systems, Weizmann Institute of Science, Rehovot, Israel.}
\newcommand{\weizmanMath}{The Center for Quantum Science and Technology,
Faculty of Mathematics and Computer Science, Weizmann Institute of
Science, Rehovot, Israel.}
\newcommand{\inria}{Inria, Télécom Paris-LTCI, Institut Polytechnique de Paris, 91120 Palaiseau, France.}
\newcommand{\EPFL}{Institute of Computer and Communication Sciences, École Polytechnique Fédérale de Lausanne (EPFL), Lausanne CH-1015, Switzerland.}
\newcommand{\LMU}{LMU Munich $\&$ Munich Center for Quantum Science and Technology (MCQST), Munich, Germany}

\newcommand{\nocontentsline}[3]{}

\begin{document}	
\let\oldaddcontentsline=\addcontentsline
\let\addcontentsline=\nocontentsline

	\title{Accessible Quantum Correlations Under Complexity Constraints}
    
	\date{\today}

        \author{Álvaro Yángüez}
        \affiliation{\sorbonne}\affiliation{\LMU}
        
        \author{Noam Avidan}
        \affiliation{\weizmanMath}
        
        \author{Jan Kochanowski}
        \affiliation{\inria}
           
        \author{Thomas A. Hahn}
        \affiliation{\LMU}\affiliation{\weizman}\affiliation{\EPFL}

\begin{abstract}
Quantum systems may contain underlying correlations which are inaccessible to computationally bounded observers. We capture this distinction through a framework that analyses bipartite states only using efficiently implementable quantum channels. This leads to a complexity-constrained max-divergence and a corresponding computational min-entropy. The latter quantity recovers the standard operational meaning of the conditional min-entropy: in the fully quantum case, it quantifies the largest overlap with a maximally entangled state attainable via efficient operations on the conditional subsystem. For classical-quantum states, it further reduces to the optimal guessing probability of a computationally bounded observer with access to side information. Lastly, in the absence of side information, the computational min-entropy simplifies to a computational notion of the operator norm. We then establish strong separations between the information-theoretic and complexity-constrained notions of min-entropy. For pure states, there exist highly entangled families of states with extremal min-entropy whose efficiently accessible entanglement in terms of computational min-entropy is exponentially suppressed. For mixed states, the separation is even sharper: the information-theoretic conditional min-entropy can be highly negative while the complexity-constrained quantity remains nearly maximal. Under the cryptographic Learning with Errors hardness assumption, we construct efficiently preparable states exhibiting both separations, even when their preparation circuits are fully known. Overall, our results demonstrate that computational constraints can fundamentally limit the quantum correlations that are observable in practice.

\end{abstract}
\maketitle

\section{Efficiently accessible quantum correlations}\label{sec:introduction}
Quantum theory provides very strong correlations, but their presence in a state does not guarantee that we can observe them. Doing so requires measurements and transformations that can be implemented in a concrete experimental setting. Every such implementation faces the same constraint: measurements and transformations must be constructed from a limited set of physical operations, using resources that scale reasonably with the system size. As the system grows, the distinction between what is mathematically allowed and what can be implemented can determine which correlations are accessible. The theoretical description should therefore incorporate this restriction from the outset.

This distinction is especially consequential in large many-body systems, where the quantum operations needed to reveal a state's correlations may themselves be difficult to implement. Pseudoentanglement provides a concrete example: under cryptographic assumptions, states with near-area-law and volume-law entanglement across spatial partitions can remain indistinguishable to efficient quantum observers, even when their preparation circuits are known~\cite{BFG+23}. Related work on many-body dynamics shows that states with much less entanglement can nevertheless reproduce the efficiently observable predictions of thermal equilibrium~\cite{FI25}.

To study the consequences of complexity constraints independently of any particular physical setting, we turn to computational complexity. Inspired by physical and complexity-theoretic formulations of the Church--Turing thesis~\cite{AD12,BV97,Deu85}, we model scalable implementability through the existence of polynomial-size quantum circuit descriptions. This restriction has fundamental consequences for how we describe quantum states: it gives rise to computational counterparts of familiar information-theoretic quantities that can take sharply different values on the same state.

To characterize these computational quantities systematically, we develop a channel-based framework for quantum information under computational restrictions and use it to identify the part of quantum correlations that remains operationally accessible to efficient observers. Comparing bipartite states only through efficiently implementable quantum channels induces a restricted order relation on quantum states, from which natural complexity-constrained versions of the max-divergence and the min-entropy emerge. Importantly, the latter retains the operational interpretation of the conditional min-entropy \cite{KRS09}: in the fully quantum setting, it quantifies the largest overlap with a maximally entangled state attainable by an efficient channel acting on the conditioning system. For classical-quantum states, it reduces to the optimal probability of inferring a classical variable from quantum side information using efficient measurements. When the conditioning system is trivial, it reduces to a single-system entropy determined by a computational analogue of the operator norm. In this way, our framework unifies recently defined computational conditional entropies \cite{avidan2025quantum,AHRA25} with the efficient-measurement formalism of \cite{YHK25}. It extends the latter from measurements to general quantum channels, recovers the former in both the fully quantum and classical-quantum settings, and includes the binary efficient-effect model of \cite{YHK25} as a special case. Within this same framework, we then prove strong separations between information-theoretic and computationally accessible correlations.

These separations already appear if one simply considers pure states. Even with access to polynomially many copies, we find pure states with nearly maximal entanglement whose computational conditional min-entropy differs substantially from its information-theoretic counterpart. When considering mixed states, the same techniques yield an even stronger and nearly maximal separation. There exist families of states whose information-theoretic conditional min-entropy remains highly negative, while the complexity-constrained conditional min-entropy is, up to negligible corrections, as large as the information-theoretic min-entropy of an uncorrelated, maximally mixed state. Moreover, using tools from quantum cryptography, namely a Learning With Errors hardness assumption, we provide explicit constructions of efficiently preparable states witnessing both separations under standard hardness assumptions. The bounds also hold when the observer is given the full classical description of the preparation circuit.
Computational complexity constraints can therefore induce an almost maximal separation between information-theoretic and computational notions of conditional min-entropy: entanglement and side information may be extensively present in principle, yet become nearly operationally inaccessible to any efficient observer.

Recent works have shown that computational restrictions can substantially reshape familiar notions from quantum information theory. In entanglement theory, one-shot computational entanglement cost and distillable entanglement have been introduced under efficient LOCC constraints \cite{ABV23}, and recent asymptotic results for pure states have shown that these constraints can lead to large separations between information-theoretic and computational notions of entanglement manipulation \cite{LREJ25}. 

In parallel, efficient-measurement frameworks have shown how distinguishability, information measures, and resource quantification change when the admissible tests themselves are computationally restricted \cite{YHK25, MRRLJE25, MRRLJE26}. The present work is close in spirit to recent proposals for quantum computational entropies, incorporating computational constraints to the operational meaning of entropy \cite{avidan2025quantum,AHRA25,HLR07}, as opposed to computational smoothing based on computational indistinguishability from high entropy states \cite{HILL99,CCL+17}.
Our approach is the following: by extending the efficient-measurement perspective of \cite{YHK25} to general quantum channels, we obtain a channel-based framework from which a complexity-constrained conditional min-entropy arises naturally through a computational divergence. This recovers the previously proposed computational conditional min-entropy from \cite{avidan2025quantum,AHRA25}, thus cementing our computational divergence as a foundational and operationally meaningful quantity. Related ideas also appear in complexity-constrained thermodynamics, where one-shot entropic quantities govern feasible implementations of thermodynamic tasks \cite{YKH+22,MKN+25,MRRLJE25}.

We now turn to the channel-based formulation that underpins our main results. 

\subsection{Complexity-Constrained Quantum Operations}
Quantum information theory typically treats all theoretically feasible physical transformations equally, independent of the underlying computational complexity~\cite{T16,KW24}. 
This standard assumption leads to a clean operational interpretation of quantum correlations and entropies (see, e.g.,~\cite{BD11, KRS09, D09}), but it ignores basic physical constraints: not every admissible transformation is efficiently implementable. Once computational complexity is treated as a physical constraint, the relevant observer is no longer characterized by the full set of quantum channels, but by a restricted class of efficiently realizable operations.

To make the notion of a feasible observer precise, we fix a finite (typically universal) gate set~\cite{NC10} and a polynomial circuit-size bound $p(n)$. We then consider the set of quantum channels (CPTP maps) which can be implemented by circuits of size at most $p(n)$. Similar to, e.g.,~\cite{avidan2025quantum,AHRA25,YHK25, MRRLJE25}, we consider the following operations to be free: initializing ancillary registers in a fixed basis state $\ket{0}$,  discarding subsystems, and measuring subsystems in the computational basis. The corresponding implementation cost is given by the gate complexity $C(T)$ of the quantum channel $T$.

We describe this set in the Heisenberg picture and, from now on,  restrict ourselves to bipartite systems for the remainder of the paper. Thus, given a bipartite state $\rho_{AB}$, the allowed transformation acts only on the subsystem $B$, while $A$ is kept as a reference system. We highlight that, for typical practical settings, the size of both subsystems implicitly scale as a function of $n$. For an isomorphic copy of the reference system $A \cong A^\prime$, the corresponding Choi representation \cite{CHOI1975285, JAMIOLKOWSKI1972275} represents the dual map $T^*$ acting from $A^\prime$ to $B$ by its Choi operator
\begin{align*}
    J(T^*)
    \coloneq
    d_A (\Id_A \otimes T^*)
    \bigl(\ketbra{\Omega_{AA'}}{\Omega_{AA'}}\bigr)\,,
\end{align*}
where $|\Omega_{AA'}\rangle=\frac{1}{\sqrt{d_A}}\sum_{x \in [d_A]}\ket{x_{A}x_{A'}}$ denotes a maximally entangled state on $\HH^A_n \otimes \HH^{A^\prime}_n$ and $d_A$ the dimension of the system $A$, as represented in \Cref{fig:overview-b}. The dual channels $T^*$ (CPU channels) and thus the admissible set of Choi operators are implicitly constrained by the gate complexity underpinning the set of efficient quantum channels $T$.

\begin{figure*}[t!]
    \centering
    \begin{overpic}[width=\textwidth]{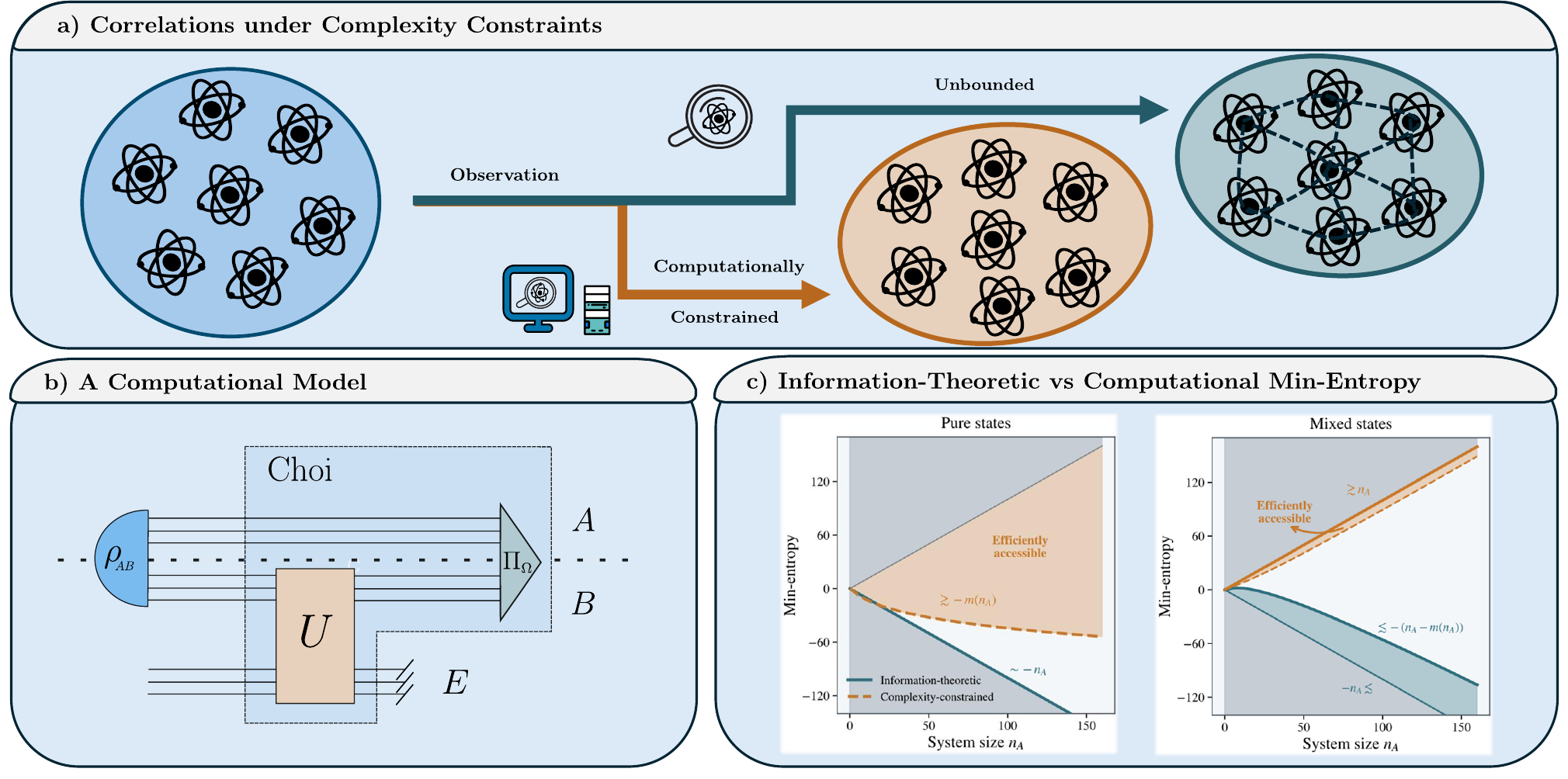}

        \put(3,92){\phantomsubcaption\label{fig:overview-a}}
        \put(3,45){\phantomsubcaption\label{fig:overview-b}}
        \put(52,45){\phantomsubcaption\label{fig:overview-c}}
    \end{overpic}
     \caption{\textbf{a)} A bipartite state \(\rho_{AB}\) may carry entanglement or quantum side information that is present information-theoretically but only partially available to a computationally bounded observer. \textbf{b)} To model this restriction, we allow only efficiently implementable quantum channels and represent them by their Choi operators. Efficiency is quantified by the gate complexity of a unitary circuit \(U\) acting on subsystem \(B\) and ancillary registers \(E\). The associated Choi operator is obtained from the projector onto a maximally entangled state, \(\Pi_{\Omega}=|\Omega\rangle\langle\Omega|\), under the channel induced by \(U\). These feasible channels define a restricted order on bipartite operators. \textbf{c)} The induced order yields a computational conditional min-entropy that captures the gap between the resource present in principle and those that remain operationally accessible under computational constraints. We illustrate this gap for the pure-state and mixed-state families considered in the paper, with \(m(n)=\omega(\log n)\).}
    \label{fig:overview}
\end{figure*}

\begin{definition}
Given a finite, discrete gate set $\mathcal{G}$, the set of efficiently implementable Choi operators is defined as
\begin{align*}
    \mathbf{J}^{p(n)}_{n,A|B}\coloneq\left\{J(T^*)\;:\;C(T) \le p(n)\right\}\,.
\end{align*}
Moreover, we denote the smallest cone containing $\mathbf{J}^{p(n)}_{n,A|B}$ by
\begin{align*}
\cC^{\mathbf{E}_\mathrm{eff}}_{n,A|B}\coloneq\operatorname{cone}\,\left(\mathbf{J}^{p(n)}_{n,A|B}\right) \, .
\end{align*}
\end{definition}

This construction effectively extends the efficient-measurement framework introduced in \cite{YHK25}  from 2-outcome quantum-classical channels to general quantum channels, where the latter are restricted to acting on the subsystem $B$.

Under the informational-completeness condition stated in \Cref{lem:proper.cone}, $\cC^{\mathbf{E}_\mathrm{eff}}_{n,A|B}$ is a proper cone with respect to the subspace of Choi operators \cite{RKW11, Jen13, Jen14}. Of particular interest is the corresponding dual cone $\cC^{\mathbf{S}_\mathrm{eff}}_{n,A|B}$, which contains all bipartite states in $AB$. Moreover, this dual cone induces a (complexity-constrained) partial order $\le_{\cC^{\mathbf{S}_\mathrm{eff}}_{n,A|B}}$, which not only allows  one to relate two quantum states but also induces a computational max-divergence~\cite{RKW11, GC24,YHK25}.
Structural properties of this construction are given in \Cref{sec:choi-cone}.

\begin{definition}
For positive semi-definite operators $\rho_{AB}$ and $\sigma_{AB}$, the \emph{computational max-divergence} is defined as
\begin{align*}
   \CompMaxDiv^{A|B}(\rho_{AB}\Vert \sigma_{AB})
    \coloneq
    \inf
    \left\{
        \lambda \in \mathbb{R}
        \;:\;
        \rho_{AB}
        \le_{\cC^{\mathbf{S}_\mathrm{eff}}_{n, A|B}}
        2^\lambda \sigma_{AB}
    \right\} \,.
\end{align*}
\end{definition}
This divergence can, e.g., be related to hypothesis testing; see \Cref{lem:DHr.Dmax.choi}. Moreover, if $\sigma_{AB}=\id_A\otimes \sigma_B$, then $\CompMaxDiv^{A|B}(\rho_{AB}\Vert \sigma_{AB})$ is independent of the choice of $\sigma_B$. As we show in the following section, this simple property allows one to derive a closed-form expression for the corresponding computational min-entropy, which recovers the results of~\cite{avidan2025quantum,AHRA25}.

\section{Computational min-entropy}
Motivated by the information-theoretic definition of the min-entropy, we define a computational analogue to it in \Cref{def:inf.H.min} that is expressed in terms of $\CompMaxDiv^{A|B}$. 
 We show below that it retains the standard operational interpretation of the conditional min-entropy from~\cite{KRS09}, except that admissible channels are now required to be efficiently implementable.

\begin{definition}\label{def:inf.H.min}
For a bipartite quantum state $\rho_{AB}$, the \emph{computational conditional min-entropy} of $A$ given $B$ is defined by
\begin{align*}
    \CompMin(A|B)_{\rho}
    \coloneq
    -
    \inf_{\substack{\sigma_B \ge 0 \\ \Tr[\sigma_B]=1}}
    \CompMaxDiv^{A|B}\!\left(\rho_{AB}\,\middle\|\, \id_A\otimes \sigma_B\right) \,.
\end{align*}
\end{definition}
Since the computational max-divergence is independent of the choice of $\sigma_B$ (\Cref{cor:indeo.max.div}), the computational min-entropy is determined solely by the cone structure underlying the set of feasible Choi operators. Apart from the properties derived below, the resulting entropy is also sub-additive under tensor products (\Cref{lem:sub-add.Hmin}).

\subsection{The Fully Quantum Setting}
This subsection is dedicated to our first main result, namely that our proposed computational min-entropy admits the same operational interpretation as its information-theoretic analogue~\cite{KRS09}. In particular, we recover the computational min-entropy defined in~\cite{AHRA25}, which quantifies the largest fidelity between the maximally entangled state $\ket{\Omega_{AA'}}$ and any state achievable by applying an efficiently implementable channel to the subsystem $B$, given the initial state $\rho_{AB}$.

\begin{theorem}\label{th:inf.op.mean.fully}
For every bipartite quantum state $\rho_{AB}$,
\begin{align*}
    \CompMin(A|B)_{\rho}    &=-\log\!\left(    d_A \max_T F_T(\rho_{AB}) \right),
\end{align*}
where
\begin{align*}
    F_T(\rho_{AB})\coloneq\bra{\Omega_{AA'}}(\Id_A\otimes T)(\rho_{AB})\ket{\Omega_{AA'}}
\end{align*}
and the maximum is taken over efficiently implementable quantum channels $T:B\to A^\prime$.
\end{theorem}

A proof of this result can be found in \Cref{sec:comp-min-entropy}, where this result is restated as \Cref{thm:op.mean.fullyquantum}. In particular, \Cref{th:inf.op.mean.fully} implies that the single quantity $\CompMaxDiv^{A|B}$ underpins the following two, a priori, separate approaches: the operational approach of defining the computational min-entropy in terms of restricted entanglement manipulation~\cite{AHRA25} and the cone-theoretic computational max-divergence construction from~\cite{YHK25} which primarily considered hypothesis testing.

\subsection{The Classical-Quantum Setting}

Let us now consider classical-quantum (cq-) states, i.e., states of the form 
\begin{align*}
    \rho_{XB}
    =
    \sum_x p_x \ketbra{x}{x}_X \otimes \rho_B^x\, .
\end{align*}
For this class of states, the closed-form expression for the computational conditional min-entropy can be further simplified. In particular, its resulting operational meaning can be expressed in terms of
 efficient guessing: the observer holding $B$ seeks to (optimally) infer the classical value stored in $X$ using an efficiently implementable measurement. In this case, the computational conditional min-entropy retains the familiar interpretation of the conditional min-entropy~\cite{KRS09}, but now relative to a computationally bounded observer~\cite{AHRA25}.

\begin{thm}\label{cor:reductiontounpent}
For every classical--quantum state $ \rho_{XB} = \sum_x p_x \ketbra{x}{x}_X \otimes \rho_B^x$, one has
\begin{align*}
    \CompMin(X|B)_{\rho}
    =-\log\CompPGuess(X|B)_{\rho}\,,
\end{align*}
where
\begin{align*}
    \CompPGuess(X|B)_{\rho}
    \coloneq
    \max_{\{E_B^x\}_x \in \mathcal M^{p(n)}_{n,|X|}}
    \sum_x p_x \Tr[E_B^x \rho_B^x]
\end{align*}
is the optimal guessing probability achievable by an efficiently implementable $|X|$-outcome POVM on subsystem $B$.
\end{thm}

The full proof can be found in \Cref{sec:comp-min-entropy}, following \Cref{th:cq.oper.int}.
Analogous to~\cite[Lemma~32]{AHRA25}, a consequence of \Cref{cor:reductiontounpent} is that our construction recovers the computational unpredictability entropy from~\cite{avidan2025quantum} for cq-states; see also the guessing pseudoentropy from~\cite{CCL+17}. 
Since the computational unpredictability entropy can be related to the extraction of pseudorandomness~\cite[Theorem II.5]{avidan2025quantum}, \Cref{cor:reductiontounpent} immediately implies that our constructed computational divergence and entropy can be related to operationally meaningful tasks.

\subsection{The Non-Conditional Setting}
The present framework also admits a natural non-conditional reduction. Here, one simply appends a trivial, one-dimensional conditioning register $B$, i.e., one defines the state
\begin{align*}
    \hat{\rho}_{AB}
    \coloneq
    \rho_A \otimes \ketbra{1}{1}_B\, 
\end{align*}
and sets
\begin{align*}
    \CompMin(A)_{\rho}
    \coloneq
    \CompMin(A|B)_{\hat{\rho}}\,.
\end{align*}
In this way, the computational conditional min-entropy reduces to a genuine single-system entropy. Moreover, as we describe below, when the conditioning system is trivial, the set of admissible quantum channels reduces to an optimization over efficiently preparable quantum states, and the entropy reduces to a complexity-constrained analogue of the operator norm.

\begin{lem}
For every quantum state $\rho_A$, its computational min-entropy can be expressed as 
\begin{align*}
    \CompMin(A)_{\rho}
 &=- \log \sup_{\tau_{A} \in \mathcal{S}^{p(n)}_{n,A}}\Tr\!\left[\rho_A \tau_{A}^{\,\mathsf T}\right]\\
 &\eqcolon-\log\|\rho_A\|_{\infty, \text{comp}}\,,
\end{align*}
where $\mathcal{S}^{p(n)}_{n,A}$ denotes the set of efficiently preparable states on $A$.
\end{lem}

A proof of this result can be found in \Cref{sec:non.cond.case}. Thus, even without side information, computational power constrains this entropy at a foundational level: it is determined not by the largest overlap with an arbitrary state, but by the largest overlap with a state that can be prepared efficiently. Lastly, in the absence of computational constraints, the derived expression reduces to the ordinary non-conditional min-entropy $H_{\min}(A)_{\rho} = -\log \|\rho_A\|_\infty$. This can easily be seen following, e.g. 
~\cite[Sec.~IV.2]{BhatiaMatrixAnalysis}.

\section{Computationally Hidden Quantum Correlations}

As discussed in the previous section, the operational interpretation of $\CompMin(A|B)$ is directly related to the entanglement that can be accessed from system $B$ under computational constraints. $H_{\min}(A|B)$ captures the same phenomenon in the information-theoretic setting. However, the gap between the computational and information-theoretic min-entropy can be significant; see also \Cref{fig:overview-c}.

We first show that the gap between computational and information-theoretic conditional min-entropy can be extensive even for pure states. For mixed states, the effect can be much stronger, as there exist states whose information-theoretic conditional min-entropy is highly negative while their computational conditional min-entropy is nearly maximal. These results also strongly improve and extend on the previous computational vs. informational min-entropy gap from~\cite[Lemma~33]{AHRA25}, which only considers fixed circuit depth (i.e., $p(n)=c$) and analyzes specific one-shot single-copy states rather than the poly-many copy regime.

Lastly, under standard cryptographic assumptions, we provide explicit constructions of efficiently preparable states witnessing both separations. Thus, efficient preparation alone does not guarantee efficient access to the correlations contained in a state.

\subsection{Pure States}
We begin by considering pure states. Building on the Haar-subsystem construction of~\cite{LREJ25}, we show that nearly maximal entanglement need not translate into efficient one-sided recovery. The gap holds even when the observer acts jointly on polynomially many copies: there exist highly entangled pure states whose information-theoretic conditional min-entropy nearly saturates the quantum bound, while their computational conditional min-entropy stays much closer to zero.

\begin{theorem}\label{th:pure.state.sep}
Let $m:\mathbb{N}\to\mathbb{N}$ satisfy
\begin{align*}
    m(n_A)=\Theta(\log^2 (n_A))\,,
\end{align*}
and let $k\in\poly(n)$. Then there exists a family of bipartite $n$-qubit pure states
$\{|\psi_{AB,n}\rangle\}_n$, on a bipartition with $n_A=n_B=\frac{n}{2}$, such that
\begin{align*}
    H_{\min}(A^k|B^k)_{|\psi_{n}\rangle^{\otimes k}}
    &=
    -k\bigl(n_A-O(1)\bigr)\,, \\
    \CompMin(A^k|B^k)_{|\psi_{n}\rangle^{\otimes k}}
    &\ge
    -k\,m(n_A)\,.
\end{align*}
\end{theorem}

A complete proof appears in \Cref{th:lower-bound.ineff}. The separation grows extensively with the system size, despite allowing the recovery channel to process all $k$ copies jointly. Thus, pure states that are nearly maximally entangled from the information-theoretic viewpoint can nevertheless appear only weakly correlated to efficient operations.

\subsection{Mixed States}

We now turn to mixed states, where mixing can hide correlations even more effectively than in the pure-state setting. The construction is based on the generalized Hilbert--Schmidt ensemble, obtained by tracing out $m$ qubits from a Haar-random $(n+m)$-qubit pure state \cite{BMB+25, Col15, Ne07, ZPN11, Ha98}. For these states, the informational conditional min-entropy remains highly negative, while the computational conditional min-entropy becomes almost maximal.

\begin{theorem}
Fix an $A|B$ bipartition with $n_A=n_B= \frac{n}{2}$, let $m:\mathbb{N}\to\mathbb{N}$ satisfy
\begin{align*}
    m(n_A)\Theta(\log^2 (n_A))\,,
\end{align*}
and let $k\in\poly(n)$. Then there exists a family of $n$-qubit mixed states
$\{\rho_{n}\}_n$ such that
\begin{align*}
    H_{\min}(A^k|B^k)_{\rho_{n}^{\otimes k}}
    &\le -k\bigl(n_A-m(n_A)\bigr)+\negl(n)\,, \\
    \CompMin(A^k|B^k)_{\rho_{n}^{\otimes k}}
    &\ge k\,n_A-\negl(n)\,.
\end{align*}
\end{theorem}

This mixed-state separation is proven in \Cref{sec:sep.mixed.state}. It is qualitatively stronger than the gap achieved in the pure-state setting. In the mixed-state construction, the computational conditional min-entropy is nearly maximal, rendering the underlying correlations essentially inaccessible through efficient operations on the conditioning subsystem. In this sense, a state that is highly correlated can nevertheless appear, to an efficient observer, as if it were almost maximally uncorrelated.

\subsection{Efficiently Preparable States}

While the previous separations establish a fundamental gap between the computational conditional min-entropy and its information-theoretic counterpart, the states witnessing them are not known to be efficiently preparable. We now extend both separations to efficiently preparable states by providing an explicit construction. This construction relies on the hardness of Learning With Errors (LWE), a standard assumption in post-quantum cryptography~\cite{Pei16}.

\begin{theorem}\label{th:inf.eff.mix}
Fix an $A|B$ bipartition with $n_A=n_B=\frac{n}{2}$, let $m:\mathbb{N}\to\mathbb{N}$ satisfy
\begin{align*}
    m(n_A)=\Theta(\log^2(n_A))\,,
\end{align*}
and let $k\in\poly(n)$. Under the LWE hardness assumption stated in \Cref{aqc:ass:subexp.lwe}, for every constant $c>0$ there exists an efficiently generated keyed ensemble of $n$-qubit mixed states $\{\rho_{s,n}\}_{s,n}$, such that, with probability $1-\negl(n)$ over the key $s$,
\begin{align*}
    H_{\min}(A^k|B^k)_{\rho_{s,n}^{\otimes k}}
    &\le -k\bigl(n_A-m(n_A)\bigr)+\negl(n)\,, \\
    \CompMin(A^k|B^k)_{\rho_{s,n}^{\otimes k}}
    &\ge k\,n_A-O(n^{-c})\,.
\end{align*}
\end{theorem}

A complete proof appears in \Cref{aqc:th:lower-bound.eff.mix}. Moreover, the same approach reproduces the pure-state separation above; see \Cref{aqc:th:lower-bound.eff}. The proof replaces the random ensembles used above with efficiently generated pseudorandom counterparts ~\cite{BS20,BCJZ25,BMB+25}. More strongly, the same bounds hold even when the recovery channel is tailored to the full classical description of the preparation circuit, as in the state-aware setting of \cite{LREJ25}. Thus, efficient preparation, even together with full knowledge of how the state is generated, does not necessarily provide efficient access to its correlations.

\section{Discussion and open problems}

Computational complexity can qualitatively reshape the scaling of operationally accessible quantum correlations. We capture this phenomenon by showing large gaps between $\CompMin(A|B)$ and $H_{\min}(A|B)$. This distinction already appears when considering pure states, where one-shot processing of polynomially many copies leaves only a logarithmic amount of entanglement within efficient reach in the worst case. For mixed states, this effect is sharper: we construct an ensemble with highly negative information-theoretic conditional min-entropy but a complexity-constrained conditional min-entropy that is, up to negligible corrections, as large as that of an uncorrelated maximally mixed state.
Under standard cryptographic assumptions, these gaps also arise for efficiently preparable states. Thus, quantum correlations can be efficiently generated yet remain largely inaccessible to computationally bounded observers, even with full knowledge of how the state was prepared.

An important open problem is to determine whether comparable separations arise for physically motivated families, such as low- or intermediate-temperature Gibbs states of interacting local Hamiltonians \cite{Alhambra.2023, Kuwahara.2021} and states generated by shallow local dynamics that are pseudorandom to efficient observers~\cite{BHHP25}. Resolving these questions would show whether almost-maximal separations between information-theoretic and operationally accessible side information are exceptional consequences of randomness, or a robust feature of complexity-constrained quantum many-body systems. Furthermore, it would be relevant to understand its implications in fields such as quantum thermodynamics \cite{MKN+25,Gour.2018}.

Moreover, in this work, we unify two previous, separately established frameworks in computational quantum information theory~\cite{YHK25,avidan2025quantum,AHRA25} by constructing a computational max-divergence which recovers the operationally meaningful computational min-entropy from~\cite{avidan2025quantum,AHRA25}. As we show in \Cref{subesc:compmax.mesured.divergenceYHK25equive}, our construction allows us to also recover the single-system computational max-divergence from~\cite{YHK25} via a suitable embedding. Interestingly, for bipartite states, the computational min-entropy convergences to the information-theoretic min-entropy in the absence of computational constraints. The same is not true for the underlying computational max-divergence which does not converge to its information-theoretic counterpart. In fact, in \Cref{rem:non.recovery} we discuss how our work suggests that this recovery is incompatible with the divergence leading to an operationally meaningful computational min-entropy.

Our work suggests viewing complexity-constrained quantum information not only as a collection of restricted tasks, but as a modification of the order structure and operational primitives from which entropic quantities arise. From this perspective, several subsequent research directions appear natural. An immediate question that arises is how our approach can be extended to more general Rényi-divergences~\cite{T16} to retain an operational meaning in complexity retained settings \cite{Audenaert.2008}.  Moreover, it remains unclear how computational entropies should be smoothed. Whereas computational notions of smoothing allow to more naturally study standard cryptographic objects, such as pseudo-random generators~\cite{CCL+17}, information-theoretic smoothing yields desirable chain rules~\cite{avidan2025quantum,AHRA25}. Finding a middle-ground would thus be of practical interest. Lastly, it is not yet fully understood how these computational divergences and entropies can be adapted to settings where one considers other notions of the computational complexity underpinning quantum states and channels; see, e.g.,~\cite{bostanci2023unitary}.

\section{Acknowledgments}

We thank Zoe García del Toro for insightful discussions and assistance with figures. We would also like to acknowledge Johannes Jakob Meyer and Verena Yacoub for helpful feedback on a previous version of this draft. We thank Pouria Fallahpour and Alex Grilo for their valuable guidance on the cryptographic constructions. 

ÁY is supported by the European Union's Horizon Europe Framework Programme under the Marie Sklodowska Curie Grant No. 101072637, Project Quantum-Safe Internet (QSI). NA and TAH are supported by the Air Force Office of Scientific Research under award number FA9550-22-1-0391. JK acknowledges support from the Program QuanTEdu-France n° ANR-22-CMAS-0001 France 2030.

\bibliographystyle{unsrtnat}

\bibliography{cited_references_only}

@inproceedings{BFG+23,
  author        = {Adam Bouland and Bill Fefferman and Soumik Ghosh and Tony Metger and Umesh Vazirani and Chenyi Zhang and Zixin Zhou},
  title         = {Public-Key Pseudoentanglement and the Hardness of Learning Ground State Entanglement Structure},
  booktitle     = {39th Computational Complexity Conference (CCC 2024)},
  series        = {Leibniz International Proceedings in Informatics (LIPIcs)},
  volume        = {300},
  pages         = {21:1--21:23},
  year          = {2024},
  publisher     = {Schloss Dagstuhl -- Leibniz-Zentrum f{\"u}r Informatik},
  address       = {Dagstuhl, Germany},
  doi           = {10.4230/LIPIcs.CCC.2024.21},
  url           = {https://doi.org/10.4230/LIPIcs.CCC.2024.21},
  eprint        = {2311.12017},
  archivePrefix = {arXiv},
  primaryClass  = {quant-ph}
}

@article{FI25,
  author  = {Feng, Xiaozhou and Ippoliti, Matteo},
  title   = {Dynamics of pseudoentanglement},
  journal = {Journal of High Energy Physics},
  volume  = {2025},
  number  = {2},
  year    = {2025},
  month   = feb,
  doi     = {10.1007/JHEP02(2025)128},
  url     = {https://doi.org/10.1007/JHEP02(2025)128}
}

@article{BMB+25,
  title = {Pseudorandom Density Matrices},
  author = {Bansal, Nikhil and Mok, Wai-Keong and Bharti, Kishor and Koh, Dax Enshan and Haug, Tobias},
  journal = {PRX Quantum},
  volume = {6},
  issue = {2},
  pages = {020322},
  numpages = {34},
  year = {2025},
  month = {May},
  publisher = {American Physical Society},
  doi = {10.1103/PRXQuantum.6.020322},
  url = {https://link.aps.org/doi/10.1103/PRXQuantum.6.020322}
}

@inproceedings{JLS18,
  title = {Pseudorandom Quantum States},
  booktitle = {Advances in Cryptology -- {{CRYPTO}} 2018},
  author = {Ji, Zhengfeng and Liu, Yi-Kai and Song, Fang},
  editor = {Shacham, Hovav and Boldyreva, Alexandra},
  year = 2018,
  pages = {126--152},
  publisher = {Springer International Publishing},
  address = {Cham},
  isbn = {978-3-319-96878-0}
}

@article{dwarakanathSamplingdiscreteGaussianslatticebasedcryptographyconstraineddevice2014a,
  title = {Sampling from Discrete {{Gaussians}} for Lattice-Based Cryptography on a Constrained Device},
  author = {Dwarakanath, Nagarjun C. and Galbraith, Steven D.},
  year = 2014,
  month = jun,
  journal = {Applicable Algebra in Engineering, Communication and Computing},
  volume = {25},
  number = {3},
  pages = {159--180},
  issn = {0938-1279, 1432-0622},
  doi = {10.1007/s00200-014-0218-3},
  urldate = {2026-09-28},
  copyright = {http://www.springer.com/tdm},
  langid = {english},
  note = {http://link.springer.com/10.1007/s00200-014-0218-3}
}

@misc{ABV23,
  title={{Computational Entanglement Theory}}, 
  author={{Arnon-Friedman}, Rotem and Brakerski, Zvika and Vidick, Thomas},
  year={2023},
  eprint={2310.02783},
  archivePrefix={arXiv},
  primaryClass={quant-ph},
  url={https://arxiv.org/abs/2310.02783}, 
}

@book{NC10,
	Author = {Nielsen, Michael A. and Chuang, Isaac L.},
	Doi = {10.1017/CBO9780511976667},
	Place = {Cambridge},
	Publisher = {Cambridge University Press},
	Title = {{Quantum Computation and Quantum Information: 10th Anniversary Edition}},
	Year = {2010}}

@article{KRS09,
   title={The Operational Meaning of Min- and Max-Entropy},
   volume={55},
   ISSN={0018-9448},
   url={http://dx.doi.org/10.1109/TIT.2009.2025545},
   DOI={10.1109/tit.2009.2025545},
   number={9},
   journal={IEEE Transactions on Information Theory},
   publisher={Institute of Electrical and Electronics Engineers (IEEE)},
   author={Konig, Robert and Renner, Renato and Schaffner, Christian},
   year={2009},
   month=sep, pages={4337–4347} }

@ARTICLE{avidan2025quantum,
  author={Avidan, Noam and Arnon, Rotem},
  journal={IEEE Transactions on Information Theory}, 
  title={Quantum Computational Unpredictability Entropy and Quantum Leakage Resilience}, 
  year={2026},
  volume={},
  number={},
  pages={1-1},
  doi={10.1109/TIT.2026.3658830}}

@misc{AHRA25,
      title={Fully Quantum Computational Entropies}, 
      author={Noam Avidan and Thomas A. Hahn and Joseph M. Renes and Rotem Arnon},
      year={2025},
      eprint={2506.14068},
      archivePrefix={arXiv},
      primaryClass={quant-ph},
      url={https://arxiv.org/abs/2506.14068}, 
}

@misc{MRRLJE26,
      title={The computational two-way quantum capacity}, 
      author={Johannes Jakob Meyer and Jacopo Rizzo and Asad Raza and Lorenzo Leone and Sofiene Jerbi and Jens Eisert},
      year={2026},
      eprint={2601.15393},
      archivePrefix={arXiv},
      primaryClass={quant-ph},
      url={https://arxiv.org/abs/2601.15393}, 
}

@book{BhatiaMatrixAnalysis,
  year = {1997},
  publisher = {Springer New York, NY},
  author = {Rajendra Bhatia},
  title = {{Matrix Analysis}}
}

@misc{CCL+17,
      title={Computational Notions of Quantum Min-Entropy}, 
      author={Yi-Hsiu Chen and Kai-Min Chung and Ching-Yi Lai and Salil P. Vadhan and Xiaodi Wu},
      year={2017},
      eprint={1704.07309},
      archivePrefix={arXiv},
      primaryClass={cs.CR},
      url={https://arxiv.org/abs/1704.07309}, 
}

@article{YHK25,
author={Yángüez, Álvaro and Hahn, Thomas A. and Kochanowski, Jan},
  journal={IEEE Transactions on Information Theory}, 
  title={Efficient Quantum Measurements: Computational Max- and Measured Rényi Divergences and Applications}, 
  year={2026},
  volume={},
  number={},
  pages={1-1},
  doi={10.1109/TIT.2026.3680247}}

@article{LREJ25,
  author        = {Lorenzo Leone and Jacopo Rizzo and Jens Eisert and Sofiene Jerbi},
  title         = {Entanglement theory with limited computational resources},
  journal       = {Nature Physics},
  volume        = {21},
  pages         = {1847--1854},
  year          = {2025},
  doi           = {10.1038/s41567-025-03048-8},
  url           = {https://doi.org/10.1038/s41567-025-03048-8},
  eprint        = {2502.12284},
  archivePrefix = {arXiv},
  primaryClass  = {quant-ph}
}

@article{MKN+25,
  author  = {Munson, Anthony and Kothakonda, Naga Bhavya Teja and Haferkamp, Jonas and Yunger Halpern, Nicole and Eisert, Jens and Faist, Philippe},
  title   = {Complexity-Constrained Quantum Thermodynamics},
  journal = {PRX Quantum},
  volume  = {6},
  number  = {1},
  year    = {2025},
  month   = mar,
  doi     = {10.1103/PRXQuantum.6.010346},
  url     = {https://doi.org/10.1103/PRXQuantum.6.010346}
}

@article{CHOI1975285,
title = {Completely positive linear maps on complex matrices},
journal = {Linear Algebra and its Applications},
volume = {10},
number = {3},
pages = {285-290},
year = {1975},
issn = {0024-3795},
doi = {10.1016/0024-3795(75)90075-0},
url = {https://www.sciencedirect.com/science/article/pii/0024379575900750},
author = {Man-Duen Choi}
}

@article{JAMIOLKOWSKI1972275,
title = {Linear transformations which preserve trace and positive semidefiniteness of operators},
journal = {Reports on Mathematical Physics},
volume = {3},
number = {4},
pages = {275-278},
year = {1972},
issn = {0034-4877},
doi = {10.1016/0034-4877(72)90011-0},
url = {https://www.sciencedirect.com/science/article/pii/0034487772900110},
author = {A. Jamiołkowski},
}

@ARTICLE{D09,
  author={Datta, Nilanjana},
  journal={IEEE Transactions on Information Theory}, 
  title={Min- and Max-Relative Entropies and a New Entanglement Monotone}, 
  year={2009},
  volume={55},
  number={6},
  pages={2816-2826},
  doi={10.1109/TIT.2009.2018325}}

@misc{KW24,
      title={Principles of Quantum Communication Theory: A Modern Approach}, 
      author={Sumeet Khatri and Mark M. Wilde},
      year={2024},
      eprint={2011.04672},
      archivePrefix={arXiv},
      primaryClass={quant-ph},
      url={https://arxiv.org/abs/2011.04672}, 
}

@article{Bus73a,
  author       = {P. J. Bushell},
  title        = {Hilbert's metric and positive contraction mappings in a Banach space},
  journal      = {Archive for Rational Mechanics and Analysis},
  volume       = {52},
  pages        = {330--338},
  year         = {1973},
  doi          = {10.1007/BF00247467},
  url          = {https://doi.org/10.1007/BF00247467}
}

@article{Bus73b,
  author       = {P. J. Bushell},
  title        = {On the Projective Contraction Ratio for Positive Linear Mappings},
  journal      = {Journal of the London Mathematical Society},
  volume       = {s2-6},
  number       = {2},
  pages        = {256--258},
  year         = {1973},
  month        = {February},
  doi          = {10.1112/jlms/s2-6.2.256},
  url          = {https://doi.org/10.1112/jlms/s2-6.2.256}
}

@article{Eve95,
  author       = {S. P. Eveson},
  title        = {Hilbert’s projective metric and the spectral properties of positive linear operators},
  journal      = {Proceedings of the London Mathematical Society},
  volume       = {70},
  pages        = {411--440},
  year         = {1995}
}

@article{RKW11,
   title={Hilbert’s projective metric in quantum information theory},
   volume={52},
   ISSN={1089-7658},
   url={http://dx.doi.org/10.1063/1.3615729},
   DOI={10.1063/1.3615729},
   number={8},
   journal={Journal of Mathematical Physics},
   publisher={AIP Publishing},
   author={Reeb, David and Kastoryano, Michael J. and Wolf, Michael M.},
   year={2011},
   month=aug }

@article{GC24,
   title={Cone-restricted information theory},
   volume={57},
   ISSN={1751-8121},
   url={http://dx.doi.org/10.1088/1751-8121/ad52d5},
   DOI={10.1088/1751-8121/ad52d5},
   number={26},
   journal={Journal of Physics A: Mathematical and Theoretical},
   publisher={IOP Publishing},
   author={George, Ian and Chitambar, Eric},
   year={2024},
   month=jun, pages={265302} }

@misc{bostanci2023unitary,
      title={Unitary Complexity and the Uhlmann Transformation Problem}, 
      author={John Bostanci and Yuval Efron and Tony Metger and Alexander Poremba and Luowen Qian and Henry Yuen},
      year={2025},
      eprint={2306.13073},
      archivePrefix={arXiv},
      primaryClass={quant-ph},
      url={https://arxiv.org/abs/2306.13073}, 
}

@article{HILL99,
  author    = {H{\aa}stad, Johan and Impagliazzo, Russell and Levin, Leonid A and Luby, Michael},
  journal   = {SIAM Journal on Computing},
  number    = {4},
  pages     = {1364--1396},
  publisher = {SIAM},
  title     = {A pseudorandom generator from any one-way function},
  volume    = {28},
  year      = {1999}
}

@inproceedings{HLR07,
author = {Hsiao, Chun-Yuan and Lu, Chi-Jen and Reyzin, Leonid},
title = {Conditional Computational Entropy, or Toward Separating Pseudoentropy from Compressibility},
year = {2007},
isbn = {9783540725398},
publisher = {Springer-Verlag},
address = {Berlin, Heidelberg},
url = {https://doi.org/10.1007/978-3-540-72540-4_10},
doi = {10.1007/978-3-540-72540-4_10},
booktitle = {Proceedings of the 26th Annual International Conference on Advances in Cryptology},
pages = {169–186},
numpages = {18},
location = {Barcelona, Spain},
series = {EUROCRYPT '07}
}

@book{BV04,
  author    = {Stephen Boyd and Lieven Vandenberghe},
  title     = {Convex Optimization},
  year      = {2004},
  publisher = {Cambridge University Press},
  address   = {Cambridge},
  isbn      = {978-0521833783},
  url       = {https://web.stanford.edu/~boyd/cvxbook/},
  note      = {First edition}
}

@book{We39,
  author    = {Weyl, Hermann},
  title     = {The Classical Groups: Their Invariants and Representations},
  year      = {1939},
  publisher = {Princeton University Press},
  address   = {Princeton, NJ},
  series    = {Princeton Mathematical Series},
  volume    = {1},
  note      = {London distribution: Oxford University Press (Humphrey Milford).}
}

@book{FH91,
  author    = {Fulton, William and Harris, Joe},
  title     = {Representation Theory: A First Course},
  series    = {Graduate Texts in Mathematics},
  volume    = {129},
  publisher = {Springer-Verlag},
  address   = {New York},
  year      = {1991},
  isbn      = {978-0-387-97527-6}
}

@book{Fu97,
  author    = {Fulton, William},
  title     = {Young Tableaux: With Applications to Representation Theory and Geometry},
  series    = {London Mathematical Society Student Texts},
  volume    = {35},
  publisher = {Cambridge University Press},
  year      = {1997},
  isbn      = {9780521567244}
}

@article{YKH+22,
   title={Resource theory of quantum uncomplexity},
   volume={106},
   ISSN={2469-9934},
   url={http://dx.doi.org/10.1103/PhysRevA.106.062417},
   DOI={10.1103/physreva.106.062417},
   number={6},
   journal={Physical Review A},
   publisher={American Physical Society (APS)},
   author={Yunger Halpern, Nicole and Kothakonda, Naga B. T. and Haferkamp, Jonas and Munson, Anthony and Eisert, Jens and Faist, Philippe},
   year={2022},
   month=dec }

@book{H11,
  title     = {Probabilistic and Statistical Aspects of Quantum Theory},
  author    = {Alexander Holevo},
  year      = {2011},
  publisher = {Edizioni della Normale Pisa},
  series    = {Publications of the Scuola Normale Superiore},
  doi       = {10.1007/978-88-7642-378-9},
  url       = {https://link.springer.com/book/10.1007/978-88-7642-378-9}
}

@misc{MRRLJE25,
      title={Computational Relative Entropy}, 
      author={Johannes Jakob Meyer and Asad Raza and Jacopo Rizzo and Lorenzo Leone and Sofiene Jerbi and Jens Eisert},
      year={2025},
      eprint={2509.20472},
      archivePrefix={arXiv},
      primaryClass={quant-ph},
      url={https://arxiv.org/abs/2509.20472}, 
}

@article{Jen13,
   title={Extremal generalized quantum measurements},
   volume={439},
   ISSN={0024-3795},
   url={http://dx.doi.org/10.1016/j.laa.2013.10.006},
   DOI={10.1016/j.laa.2013.10.006},
   number={12},
   journal={Linear Algebra and its Applications},
   publisher={Elsevier BV},
   author={Jenčová, Anna},
   year={2013},
   month=dec, pages={4070–4079} }

@article{Jen14,
   title={Base norms and discrimination of generalized quantum channels},
   volume={55},
   ISSN={1089-7658},
   url={http://dx.doi.org/10.1063/1.4863715},
   DOI={10.1063/1.4863715},
   number={2},
   journal={Journal of Mathematical Physics},
   publisher={AIP Publishing},
   author={Jenčová, A.},
   year={2014},
   month=feb }

@book{T16,
   title={Quantum Information Processing with Finite Resources},
   ISBN={9783319218915},
   ISSN={2197-1765},
   url={http://dx.doi.org/10.1007/978-3-319-21891-5},
   DOI={10.1007/978-3-319-21891-5},
   journal={SpringerBriefs in Mathematical Physics},
   publisher={Springer International Publishing},
   author={Tomamichel, Marco},
   year={2016} }

@article{RSB24,
  author        = {Tobias Rippchen and Sreejith Sreekumar and Mario Berta},
  title         = {Locally-Measured R\'enyi Divergences},
  journal       = {IEEE Transactions on Information Theory},
  volume        = {71},
  number        = {8},
  pages         = {6105--6133},
  year          = {2025},
  doi           = {10.1109/TIT.2025.3571527},
  url           = {https://doi.org/10.1109/TIT.2025.3571527},
  eprint        = {2405.05037},
  archivePrefix = {arXiv},
  primaryClass  = {quant-ph}
}

@article{ZPN11,
   title={Generating random density matrices},
   volume={52},
   ISSN={1089-7658},
   url={http://dx.doi.org/10.1063/1.3595693},
   DOI={10.1063/1.3595693},
   number={6},
   journal={Journal of Mathematical Physics},
   publisher={AIP Publishing},
   author={Zyczkowski, Karol and Penson, Karol A. and Nechita, Ion and Collins, Benoît},
   year={2011},
   month=jun }

@article{Ne07,
   title={Asymptotics of Random Density Matrices},
   volume={8},
   ISSN={1424-0661},
   url={http://dx.doi.org/10.1007/s00023-007-0345-5},
   DOI={10.1007/s00023-007-0345-5},
   number={8},
   journal={Annales Henri Poincaré},
   publisher={Springer Science and Business Media LLC},
   author={Nechita, Ion},
   year={2007},
   month=nov, pages={1521–1538} }

@article{Me24,
   title={Introduction to Haar Measure Tools in Quantum Information: A Beginner Tutorial},
   volume={8},
   ISSN={2521-327X},
   url={http://dx.doi.org/10.22331/q-2024-05-08-1340},
   DOI={10.22331/q-2024-05-08-1340},
   journal={Quantum},
   publisher={Verein zur Forderung des Open Access Publizierens in den Quantenwissenschaften},
   author={Mele, Antonio Anna},
   year={2024},
   month=may, pages={1340} }

@article{HLW06,
   title={Aspects of Generic Entanglement},
   volume={265},
   ISSN={1432-0916},
   url={http://dx.doi.org/10.1007/s00220-006-1535-6},
   DOI={10.1007/s00220-006-1535-6},
   number={1},
   journal={Communications in Mathematical Physics},
   publisher={Springer Science and Business Media LLC},
   author={Hayden, Patrick and Leung, Debbie W. and Winter, Andreas},
   year={2006},
   month=mar, pages={95–117} }

@article{Col15,
   title={Random matrix techniques in quantum information theory},
   volume={57},
   ISSN={1089-7658},
   url={http://dx.doi.org/10.1063/1.4936880},
   DOI={10.1063/1.4936880},
   number={1},
   journal={Journal of Mathematical Physics},
   publisher={AIP Publishing},
   author={Collins, Benoît and Nechita, Ion},
   year={2015},
   month=dec }

@article{Ha98,
   title={Random quantum correlations and density operator distributions},
   volume={242},
   ISSN={0375-9601},
   url={http://dx.doi.org/10.1016/S0375-9601(98)00190-X},
   DOI={10.1016/s0375-9601(98)00190-x},
   number={3},
   journal={Physics Letters A},
   publisher={Elsevier BV},
   author={Hall, Michael J.W.},
   year={1998},
   month=may, pages={123–129} }

@article{BD11,
  author    = {Fernando G. S. L. Brand{\~a}o and Nilanjana Datta},
  title     = {One-shot rates for entanglement manipulation under non-entangling maps},
  journal   = {IEEE Transactions on Information Theory},
  volume    = {57},
  number    = {3},
  pages     = {1754--1760},
  year      = {2011},
  doi       = {10.1109/TIT.2011.2104531},
  archivePrefix = {arXiv},
  eprint    = {0905.2673},
  primaryClass = {quant-ph}
}

@book{Billingsley1999,
  author    = {Patrick Billingsley},
  title     = {Convergence of Probability Measures},
  edition   = {2},
  publisher = {Wiley},
  year      = {1999}
}

@article{Gi65,
  author  = {Ginibre, Jean},
  title   = {Statistical Ensembles of Complex, Quaternion, and Real Matrices},
  journal = {Journal of Mathematical Physics},
  volume  = {6},
  number  = {3},
  pages   = {440--449},
  year    = {1965},
  month   = mar,
  doi     = {10.1063/1.1704292},
  url     = {https://doi.org/10.1063/1.1704292}
}

@misc{BS20,
      title={Scalable Pseudorandom Quantum States}, 
      author={Zvika Brakerski and Omri Shmueli},
      year={2020},
      eprint={2004.01976},
      archivePrefix={arXiv},
      primaryClass={quant-ph},
      url={https://arxiv.org/abs/2004.01976}, 
}

@misc{BHHP25,
      title={Efficient Quantum Pseudorandomness from Hamiltonian Phase States}, 
      author={John Bostanci and Jonas Haferkamp and Dominik Hangleiter and Alexander Poremba},
      year={2025},
      eprint={2410.08073},
      archivePrefix={arXiv},
      primaryClass={quant-ph},
      url={https://arxiv.org/abs/2410.08073}, 
}

@article{Gour.2018,
  author       = {Gour, Gilad and Jennings, David and Buscemi, Francesco and Duan, Runyao and Marvian, Iman},
  title        = {Quantum majorization and a complete set of entropic conditions for quantum thermodynamics},
  journaltitle = {Nature Communications},
  year         = {2018},
  volume       = {9},
  number       = {1},
  pages        = {5352},
  doi          = {10.1038/s41467-018-06261-7},
  issn         = {2041-1723},
  date         = {2018-12-17}
}

@article{Alhambra.2023,
  title = {Quantum Many-Body Systems in Thermal Equilibrium},
  author = {Alhambra, \'Alvaro M.},
  journal = {PRX Quantum},
  volume = {4},
  issue = {4},
  pages = {040201},
  numpages = {35},
  year = {2023},
  month = {Nov},
  publisher = {American Physical Society},
  doi = {10.1103/PRXQuantum.4.040201},
  url = {https://link.aps.org/doi/10.1103/PRXQuantum.4.040201}
}

@article{Kuwahara.2021,
  title = {Improved Thermal Area Law and Quasilinear Time Algorithm for Quantum Gibbs States},
  author = {Kuwahara, Tomotaka and Alhambra, \'Alvaro M. and Anshu, Anurag},
  journal = {Phys. Rev. X},
  volume = {11},
  issue = {1},
  pages = {011047},
  numpages = {42},
  year = {2021},
  month = {Mar},
  publisher = {American Physical Society},
  doi = {10.1103/PhysRevX.11.011047},
  url = {https://link.aps.org/doi/10.1103/PhysRevX.11.011047}
}

@article{Audenaert.2008,
  author       = {Audenaert, K. M. R. and Nussbaum, M. and Szkoła, A. and Verstraete, F.},
  title        = {Asymptotic Error Rates in Quantum Hypothesis Testing},
  journaltitle = {Communications in Mathematical Physics},
  year         = {2008},
  volume       = {279},
  number       = {1},
  pages        = {251--283},
  doi          = {10.1007/s00220-008-0417-5},
  issn         = {1432-0916},
  date         = {2008-04-01}
}

@misc{BCJZ25,
      title={Scalable, quantum-accessible, and adaptive pseudorandom quantum state and pseudorandom function-like quantum state generators}, 
      author={Rishabh Batra and Zhili Chen and Rahul Jain and YaoNan Zhang},
      year={2026},
      eprint={2507.22535},
      archivePrefix={arXiv},
      primaryClass={quant-ph},
      url={https://arxiv.org/abs/2507.22535}, 
}

@misc{BS19,
      title={(Pseudo) Random Quantum States with Binary Phase}, 
      author={Zvika Brakerski and Omri Shmueli},
      year={2019},
      eprint={1906.10611},
      archivePrefix={arXiv},
      primaryClass={quant-ph},
      url={https://arxiv.org/abs/1906.10611}, 
}

@article{Zha12,
author = {Zhandry, Mark},
title = {How to Construct Quantum Random Functions},
year = {2021},
issue_date = {October 2021},
publisher = {Association for Computing Machinery},
address = {New York, NY, USA},
volume = {68},
number = {5},
issn = {0004-5411},
url = {https://doi.org/10.1145/3450745},
doi = {10.1145/3450745},
journal = {J. ACM},
month = aug,
articleno = {33},
numpages = {43}
}

@article{Pei16,
author = {Peikert, Chris},
title = {A Decade of Lattice Cryptography},
year = {2016},
issue_date = {Mar 2016},
publisher = {Now Publishers Inc.},
address = {Hanover, MA, USA},
volume = {10},
number = {4},
issn = {1551-305X},
url = {https://doi.org/10.1561/0400000074},
doi = {10.1561/0400000074},
journal = {Found. Trends Theor. Comput. Sci.},
month = mar,
pages = {283–424},
numpages = {145}
}

@article{GGM86,
author = {Goldreich, Oded and Goldwasser, Shafi and Micali, Silvio},
title = {How to construct random functions},
year = {1986},
issue_date = {Oct. 1986},
publisher = {Association for Computing Machinery},
address = {New York, NY, USA},
volume = {33},
number = {4},
issn = {0004-5411},
url = {https://doi.org/10.1145/6490.6503},
doi = {10.1145/6490.6503},
journal = {J. ACM},
month = aug,
pages = {792–807},
numpages = {16}
}

@misc{BLP+13,
      title={Classical Hardness of Learning with Errors}, 
      author={Zvika Brakerski and Adeline Langlois and Chris Peikert and Oded Regev and Damien Stehlé},
      year={2013},
      eprint={1306.0281},
      archivePrefix={arXiv},
      primaryClass={cs.CC},
      url={https://arxiv.org/abs/1306.0281}, 
}

@article{EAO+02,
  author = {Ekert, Artur K. and Alves, Carolina Moura and Oi, Daniel K. L. and Horodecki, Micha{\l} and Horodecki, Pawe{\l} and Kwek, Leong Chuan},
  title = {Direct Estimations of Linear and Nonlinear Functionals of a Quantum State},
  journal = {Physical Review Letters},
  volume = {88},
  number = {21},
  pages = {217901},
  year = {2002},
  doi = {10.1103/PhysRevLett.88.217901}
}

@article{Gir07,
  author  = {Giraud, Olivier},
  title   = {Distribution of bipartite entanglement for random pure states},
  journal = {Journal of Physics A: Mathematical and Theoretical},
  volume  = {40},
  number  = {11},
  pages   = {2793--2801},
  year    = {2007},
  doi     = {10.1088/1751-8113/40/11/014}
}

@article{AD12,
  author  = {Pablo Arrighi and Gilles Dowek},
  title   = {The physical {Church--Turing} thesis and the principles of quantum theory},
  journal = {International Journal of Foundations of Computer Science},
  volume  = {23},
  number  = {5},
  pages   = {1131--1145},
  year    = {2012},
  doi     = {10.1142/S0129054112500153}
}

@article{BV97,
  author  = {Ethan Bernstein and Umesh Vazirani},
  title   = {Quantum complexity theory},
  journal = {SIAM Journal on Computing},
  volume  = {26},
  number  = {5},
  pages   = {1411--1473},
  year    = {1997},
  doi     = {10.1137/S0097539796300921}
}

@article{Deu85,
  author  = {David Deutsch},
  title   = {Quantum theory, the {Church--Turing} principle and the universal quantum computer},
  journal = {Proceedings of the Royal Society of London. A. Mathematical and Physical Sciences},
  volume  = {400},
  number  = {1818},
  pages   = {97--117},
  year    = {1985},
  doi     = {10.1098/rspa.1985.0070}
}

\let\addcontentsline=\oldaddcontentsline
\clearpage

\appendix
\makeatletter
\renewcommand{\thesection}{\Alph{section}}
\renewcommand{\thesubsection}{\thesection.\arabic{subsection}}
\renewcommand{\thesubsubsection}{\thesubsection.\Alph{subsubsection}}

\renewcommand{\p@section}{}
\renewcommand{\p@subsection}{}
\renewcommand{\p@subsubsection}{}
\makeatother

\onecolumngrid
\newgeometry{left=1.05in,right=1.05in,top=1.0in,bottom=1.1in}
\setstretch{1.08} 
\setlength{\parskip}{0.15em}
\setlength{\parindent}{0pt}

\title{Supplemental Material }

\makesupplementtitle{Supplemental Material\\[0.3em]Accessible Quantum Correlations under Complexity Constraints}

\onecolumngrid

\tableofcontents

\newgeometry{left=1.05in,right=1.05in,top=1.0in,bottom=1.1in}
\setstretch{1.08} 
\setlength{\parskip}{0.35em}
\setlength{\parindent}{0pt}

\section{Preliminaries}

\subsection{Notation}\label{sec:notation}
In this work we assume all Hilbert spaces to be finite dimensional. Hilbert spaces will be denoted by $\HH$. When we consider $n$ qudits of local dimension $d$, we write $\HH_n \simeq (\CC^{d})^{\otimes n}$. Given a Hilbert space $\HH$, we denote the set of bounded operators acting on it as $\BB(\HH)$ and its subset of positive semi-definite operators by $\Pos(\HH)$. An operator $\rho \in \Pos(\HH)$ is called a \emph{quantum state} if we have $\Tr \left[\rho \right] = 1$. The set of quantum states on $\HH$ is denoted by $\Ss(\HH)$. The Hilbert-Schmidt inner product on $\BB(\HH)$ is given by $\langle a,b\rangle\coloneqq\Tr[a^*b]$, where $^*$ denotes the Hilbert-Schmidt adjoint. The partial trace  is the unique map $\tr_B: \BB(\HH_A \otimes \HH_B) \rightarrow \BB(\HH_A)$ such that 
$$\Tr[\tr_B[X]Y] = \Tr[X(Y_A \otimes \id_B) ]\,.$$ 
We denote the identity operator on $\HH$ as $\id \in \BB(\HH)$ while the identity channel is the map 
$$\Id:\BB(\HH)\rightarrow \BB(\HH) \;,$$ 
such that $\Id(X)=X$  for all $X \in \BB(\HH)$. For two operators $\rho, \sigma \in \Pos(\HH)$ we write $\rho \ll \sigma$ if $\ker{\sigma} \subseteq \ker{\rho}$, where $\ker{\tau} \coloneqq  \{\ket{v} \,:\, \tau \ket{v} = 0\}$.  

In Landau notation, given two functions $f(n)$ and $g(n)$, we write $f(n)=o(g(n))$ if \newline$\lim_{n\rightarrow \infty} f(n)/g(n) = 0$. In the same way, $f(n)=\omega(g(n))$ if $\lim_{n\rightarrow \infty} f(n)/g(n) = \infty$. $f(n)=O(g(n))$ if there exist two constant $ C,N>0$ such that $|f(n)| \leq C|g(n)|$ for all $n \geq N$. Similarly, $f(n)=\Omega(g(n))$ if there exist two constant $ C,N>0$ such that $|f(n)| \geq C|g(n)|$ for all $n \geq N$. Lastly, $f(n)=\Theta(g(n))$ if both $f(n)=O(g(n))$ and $f(n)=\Omega(g(n))$. A function $f(n)$ is $\negl(n)$, i.e., negligible, if, for every fixed $c$, $f(n) = o(1/n^c)$.

With $\log$ we denote the logarithm to base 2 and with $\ln$ the logarithm to base $e$. 
For any $n\in \mathbb{N}$, we denote $[n]=\{1,\dots,n\}$. 
For a set \(S\), its affine hull is given by 
$$\operatorname{aff}(S)\coloneq \Bigl\{\sum_{i=1}^m \lambda_i x_i:\ x_i\in S,\ \lambda_i\in\mathbb{R},\ \sum_{i=1}^m \lambda_i=1\Bigr\}\;.$$
Lastly, its relative interior is $\operatorname{ri}(S)\coloneq \{x\in S:\exists \varepsilon>0 \text{ such that } B(x,\varepsilon)\cap \operatorname{aff}(S)\subseteq S\}$, where $ B(x,\varepsilon)$ is an open ball of radius $\varepsilon$ around $x$.

\subsection{Information Theoretic Concepts}

We begin by introducing the class of linear maps acting on quantum states that will be used throughout this paper and their adjoints.

\begin{definition}[Dual map]\label{def:dual.map}
Given a bounded linear map $T: \BB(\HH_A) \rightarrow \BB(\HH_B)$, there exists a unique linear map $T^{*} : \BB(\HH_B) \rightarrow \BB (\HH_A)$ defined as the dual map, such that $\forall X \in \BB(\HH_B), \; \forall \rho \in \BB(\HH_A)$:

\begin{equation}
    \Tr[T(\rho)X]= \Tr[\rho T^*(X)]\,.
    \label{eq:dual}
\end{equation}
    
\end{definition}

Moreover, the linear maps that represent physically admissible transformations are \emph{quantum channels}, i.e., completely positive and trace-preserving (CPTP) maps in the Schrödinger picture. Their adjoints are completely positive and \emph{unital}  (CPU) maps, i.e.\ they satisfy
\begin{align*}
    T^{*}(\id_{\HH_B})=\id_{\HH_A}.
\end{align*}
Throughout this work we will frequently use the Choi-Jamiołkowski representation of quantum channels.

\begin{definition}[Choi-Jamiołkowski representation \cite{CHOI1975285,JAMIOLKOWSKI1972275}] \label{def:Choi}
Let $\HH^A_n, \HH^{A'}_n$ be finite-dimensional Hilbert spaces with $\dim \HH^A_n = d_A$, and let $\HH_A \cong \HH_{A'}$. Let $T: \BB(\HH_{A^\prime}) \to\BB(\HH_{B})$ be a CPTP map. For such a $T$, its Choi-Jamiołkowski operator representation on $AB$ is given by
	\begin{align*}
		J(T) \coloneq  d_A (\Id_A \otimes T) (\ketbra{\Omega_{AA'}}{\Omega_{AA'}})\,,
	\end{align*}
	where $\ket{\Omega_{AA'}} = \frac{1}{\sqrt{d_A}}\sum_{x \in [d_A]} \ket{x}_A \otimes \ket{x}_{A'}$.
\end{definition}

Having fixed our notation for quantum channels, we next specialize to channels arising from quantum measurements, which are conveniently described in terms of POVMs.

\begin{definition} [POVM] A \term{positive operator-valued measure} (POVM) is given by a tuple of positive semi-definite operators $M =( E_{1},\dots,E_{m} )$ (for some $m \in \mathbb{N}$) such that ${\sum_{k=1}^m E_{k} = \id}$. Moreover, if ${E_{i} E_{j} = \delta_{ij} E_{i}}$ for all $i,j \in \{1,\dots, m \}$, we call this a \term{projection-valued measure} (PVM).
The operators $E_i$ of a POVM are called \textit{effect operators} or \textit{POVM elements}.
\end{definition}

\begin{theorem}[Naimark’s Dilation Theorem in Finite Dimensions; see, e.g., Theorem 2.5.1 of~\cite{H11}] \label{Thm: Naimark}
	Any POVM $M=(E_i)_{i=1}^K$ on $\mathcal{H}$ can be implemented via an isometry and a projective measurement on ancilla systems. In particular, there exists an isometry $V:\mathcal{H}\to \mathcal{H}\otimes \mathcal{H}_{K}$, where $\mathcal{H}_{K}$ denotes a Hilbert space of dimension $K$, such that
	\begin{align}
		E_i =  V^* \left(\id_\mathcal{H} \otimes \ketbra{i}{i}_{\mathcal{H}_K} \right) V =T^*
		(\ketbra{i}{i})\quad \forall i\in [K]\,.
	\end{align}
\end{theorem}

Here $T^*(X)\coloneq V^* \left(\id_\mathcal{H} \otimes X\right) V $ for $X \in \BB(\HH_K)$ is a completely positive unital (CPU) map. Note that one may choose $V=U\ket{0}$, where $U: \mathcal{H}\otimes \mathcal{H}_{K} \to \mathcal{H}\otimes \mathcal{H}_{K}$ is a unitary and $\ket{0}\in\mathcal{H}_K$.

Moreover, one can associate to any $m$-outcome POVM $M=(E_1,...,E_m)$ a channel that maps quantum states to classical distributions over $\{1,\dots,m\}$, via
\begin{align}
    \MM_M\, : \; \rho \mapsto \sum_{i=1}^{m}\Tr[\rho E_{i}]\, \ketbra{i}{i} \,.
\end{align} 
We will sometimes drop the subscript $M$, when the POVM to which this measurement map corresponds is clear from context. 

To quantify how well states can be distinguished we will use a number of standard distance and divergence measures. We begin with fidelity and relative entropy, and then introduce the conditional min-entropy, which will play a central role in our paper.

\begin{definition} [Fidelity]
    Given  two quantum states $\rho , \sigma\in \Ss(\HH)$, the  \term{fidelity} between $\rho$, $\sigma$ is given by:
		\begin{align*}
			F(\rho,\sigma) &\coloneqq \norm{\sqrt{\rho }\sqrt{\sigma }}_{1}^{2}
				 \; .
		\end{align*}
\end{definition}

\begin{definition} [Max-divergence~\cite{D09}]
    Given any two positive semi-definite operators $\rho , \sigma\in \text{Pos}(\HH)$ with $\Tr\left[\rho\right] > 0$ and $\rho \ll \sigma$, the \term{max-divergence} between $\rho$, $\sigma$  is given by:
		\begin{align}
                D_\text{max}(\rho\|\sigma) &\coloneqq 
				\log\inf\{\lambda \geq 0\,|\,\rho\leq\lambda\sigma\} \; .
		\end{align}
\end{definition}

\begin{definition}[Conditional Min-Entropy \cite{KRS09}]
Let $\rho_{AB}\in \Pos(\HH_A \otimes \HH_B)$ be a bipartite quantum state. The \emph{min-entropy} of $A$ conditioned on $B$ is given by,
\begin{align*}
    H_\text{min}(A|B)_{\rho}\coloneq -\inf_{\substack{\sigma_B \geq 0 \\ \Tr[\sigma_B]=1}} D_\text{max}(\rho_{AB}\|\id_A \otimes \sigma_B)\,.
\end{align*}
\end{definition}

Moreover, as shown in \cite{KRS09}, the conditional min-entropy admits the following operational interpretation:
\begin{align*}
    H_\text{min}(A|B)_{\rho} = -\log d_A \max_{T} F\left((\id_A \otimes T)(\rho_{AB}),|\Omega_{AA'}\rangle \langle \Omega_{AA'}|\right)\,,
\end{align*}
where the maximization is over channels $T:B\to A'$, and $|\Omega_{AA'}\rangle= \frac{1}{\sqrt{d_A}}\sum_{x \in [d_A]} \ket{x}_A \otimes \ket{x}_{A'}$. Thus, the conditional min-entropy is a one-shot measure that quantifies how well one can recover a maximally entangled state on $AA'$ from $\rho_{AB}$ by acting locally on $B$. For states with quantum correlations, $H_\text{min}(A|B)_{\rho}$ can be negative, reaching $-\log d_A$ for a maximally entangled state.

In the case of the input state being a classical-quantum state, i.e., $\rho_{XB}=\sum_x p_x |x\rangle \langle x| \otimes \rho_B^x$, the operational meaning of the conditional min-entropy is the following:
\begin{align*}
     H_\text{min}(X|B)_{\rho}=-\log \max_{\{E_B^x\}_x}\sum_x p_x \Tr[E^x_B \rho^x_B]\,,
\end{align*}
where $\{E_B^x\}_x$ is a POVM on $B$. Then, in this case, the conditional min-entropy quantifies the predictability of the classical register $X$ given access to $B$.

One can define other relative entropies or divergence from an operational perspective. In this case, we will also work with the hypothesis testing relative entropy.

\begin{definition}[$\varepsilon$-hypothesis testing relative entropy \cite{KW24}]
 Given any two positive semidefinite operators $\rho , \sigma\in \text{Pos}(\HH)$ with $\Tr\left[\rho\right] > 0$ and $\rho \ll \sigma$, and $\varepsilon\in [0,1]$, its corresponding \emph{$\varepsilon$-hypothesis testing relative entropy } is given by
    \begin{align}\label{eq:info_hypothesis_testing_D_1minuseps}
        D_{H}^{\varepsilon}(\rho \| \sigma) \coloneq  -\log \inf_{E} \{\Tr[E\sigma] : 0\le E\le \id \;, \Tr[(\id-E)\rho ] \le \varepsilon  \} \,.
    \end{align}
\end{definition}

This entropy can equivalently be expressed as
\begin{equation}
	D_{H}^{1-\varepsilon}(\rho || \sigma) \coloneq  -\log \inf_{E} \{\Tr[E\sigma] : 0\le E\le \id \;, \Tr[E\rho ] \ge \varepsilon  \} \;.
\end{equation}

\subsection{Cone Theory}\label{sec:cone.theory}
In this section we summarize some basic notions of cone theory. For a deeper understanding, the reader may refer to \cite{Bus73a,Bus73b,Eve95,RKW11,GC24}.
\begin{definition} [Convex Cone]
    Let $\mathcal{V}$ be a finite-dimensional real vector space with inner product $\langle\cdot,\cdot\rangle$. 
A \term{convex cone} $\mathcal{C}\subset \mathcal{V}$ is a subset that satisfies
\begin{align*}
    a,b\in \mathcal{C} \implies \lambda a+\mu b\in\mathcal{C} \quad \forall \lambda,\mu\geq 0\, .
\end{align*}
Moreover, for any $K\subset \mathcal{V}$, we denote by $\cone(K)\coloneqq\conv(\bigcup_{\lambda\geq 0}\lambda K)$ the smallest (convex) cone that contains $K$.
\end{definition}

A cone is \term{closed} if it is closed in the norm topology induced by the inner product on $\mathcal{V}$ and, hence, it contains its boundary. Following~\cite{RKW11}, $\cC$ is \term{pointed} if $\cC\cap(-\cC)=\{0\}$, where
\begin{align*}
    a\in -\cC \iff -a\in \cC \, .
\end{align*}
Moreover, a cone is \term{solid} if $\Span(\cC)=\mathcal{V}$. In finite dimensions this is equivalent to $\mathsf{ri}(\cC)\neq\emptyset$.

\begin{definition} [Proper Cone]
A convex, closed, pointed, and solid cone is called a \textit{proper cone}.
\end{definition}

One can also assign to each cone a dual cone.
\begin{definition} [Dual Cone]
Given a cone $\mathcal{C}$, its \term{dual cone} $\cC^*$ is given by
\begin{align*}
    \cC^*\coloneqq\{x\in\VV^*| \langle x,c\rangle\geq 0, \ \forall c\in\cC\} \, .
\end{align*} 
\end{definition}

Therefore, if $\cC_1\subset\cC_2$, then $\cC^*_2\subset\cC_1^*$. We additionally note that the dual cone is also closed and convex. Moreover, if $\cC$ is a proper cone, then so is $\cC^*$. Lastly, in finite-dimensional vector spaces, closed convex cones satisfy that the corresponding dual cone is also its pre-dual cone, i.e. $(\cC^*)^*=\cC$.
Given a real vector space $\mathcal{V}$ and a proper cone $\mathcal{C} \subset \mathcal{V}$, one can define a natural partial order on $\mathcal{V}$.
\begin{definition}[Partial Order]
Let $\mathcal{V}$ be a finite-dimensional real vector space and $\mathcal{C} \subset \mathcal{V}$ be a proper cone. We say that $a\geq_{\cC}b$ if and only if $a-b\in\cC$.
\end{definition}

In the following, we take $\VV = \Herm(\HH)$. This vector space is isomorphic to $\RR^{d^2}$, where $d=\dim(\HH)$.
A canonical example of a proper cone in quantum information theory is the cone of \term{positive semidefinite} operators in $\Herm(\HH)$, 
\begin{align}
    \Pos(\mathcal{H})\coloneqq\{a\in\Herm(\HH) \,| \, a\geq 0\}\,.
\end{align}
This cone is self-dual in $\VV$ and we denote its induced partial order by $\geq$. This particular partial order is relevant for the definition of the max-divergence~\cite{D09,RKW11}, which can be written as
\begin{align*}
    D_\text{max}(a\|b) &=  \log\inf\{\lambda\in\RR| a\leq\lambda b\} \\
    &= \log\sup_{v\in\Pos(\mathcal{H})}\frac{\langle v,a\rangle}{\langle v,b\rangle} \, .
\end{align*}
In the optimization above, $v$  should be understood as an effect operator in the pre-dual of the cone of $\Pos(\HH)$, in which $a,b$ live, rather than that cone itself. Since $\Pos(\HH)$ is self-dual with respect to the Hilbert-Schmidt inner product, these cones coincide here. This distinction will matter for general cones in the coming sections.

If one considers other proper cones inside of $\Herm(\HH)$, then one can use 
this relation between the max-divergence and the conic structure of $\Pos(\HH)$ to derive a more general notion of max-divergences.
\begin{definition} [$\cC$-Max-divergence~\cite{RKW11, GC24}] \label{Def: CMaxDiv}
    Let $\cC$ be a proper cone. For any $a,b\in\cC\setminus\{0\}$, the \term{$\cC$-max-divergence} between $a,b$ is given by
    \begin{align}
        D^{\cC}_\text{max}(a\|b)\coloneqq  \log\inf\{\lambda\geq 0 \,| \, a\leq_{\cC}\lambda b\} \, .
    \end{align}
    If no $\lambda\in\RR$ exists such that $\lambda b-a\in\cC$, i.e. the infimum is over an empty set, then we set $ D^{\cC}_\text{max}(a\|b)=\infty$.
\end{definition}
Moreover, there exists a dual formulation of the $\cC$-max-divergence.
\begin{lem} [Dual Expression~\cite{RKW11,GC24}] \label{lem: dualexp}
    Let $\cC$ be a proper cone. Then 
    \begin{align}
    D^{\cC}_\textup{max}(a\|b)\coloneqq \log\inf\{\lambda\in\RR| a\leq_{\cC}\lambda b\}=\log\sup_{v\in\cC^*}\frac{\langle v,a\rangle}{\langle v,b\rangle}  \, .
    \end{align}
\end{lem}

In order to exploit the duality between cones, it is convenient to use tools from \emph{conic} (and in particular \emph{semidefinite}) programming \cite{BV04}. Let $\VV_1$ and $\VV_2$ be finite-dimensional Euclidean spaces equipped with inner products $\langle\cdot,\cdot\rangle_1$ and $\langle\cdot,\cdot\rangle_2$. Generalizing \Cref{def:dual.map}, given a linear map $\Phi:\VV_1\to\VV_2$, its (Hilbert-space) adjoint $\Phi^*:\VV_2\to\VV_1$ is the unique linear map satisfying
\begin{align*}
    \langle \Phi(v_1), v_2\rangle_2=\langle v_1,\Phi^*(v_2)\rangle_1
    \qquad \forall\; v_1\in\VV_1,\ \forall\; v_2\in\VV_2\,.
\end{align*}
Let $\mathcal{C}_1\subseteq\VV_1$ and $\mathcal{C}_2\subseteq\VV_2$ be proper cones, and let $c\in\VV_1$ and $b\in\VV_2$. A standard conic program (the \emph{primal} problem) and its associated \emph{dual} problem are given by
\begin{align}
\label{eq:conic.primal}
    \gamma^{\mathrm{primal}}
    \;\coloneq \;
    \inf_{x\in\VV_1}\ \langle c,x\rangle_1
    \quad \text{such that}\quad
    \Phi(x)-b \in \mathcal{C}_2,\quad x\in \mathcal{C}_1\,,
\end{align}
and
\begin{align}
\label{eq:conic.dual}
    \gamma^{\mathrm{dual}}
    \;\coloneq \;
    \sup_{y\in\VV_2}\ \langle b,y\rangle_2
    \quad \text{such that}\quad
    c-\Phi^*(y)\in \mathcal{C}_1^*,\quad y\in \mathcal{C}_2^*\,,
\end{align}
where $\mathcal{C}_1^*$ and $\mathcal{C}_2^*$ denote the dual cones with respect to $\langle\cdot,\cdot\rangle_1$ and $\langle\cdot,\cdot\rangle_2$, respectively. When $\mathcal{C}_1$ and $\mathcal{C}_2$ are positive semidefinite cones, these become \emph{semidefinite programs}.

If the primal and dual problems are feasible, then \emph{weak duality} holds: for every primal-feasible
$x\in\VV_1$ and dual-feasible $y\in\VV_2$ one has
\begin{align*}
    \langle c,x\rangle_1 \;\ge\; \langle b,y\rangle_2,
\end{align*}
and in particular $\gamma^{\mathrm{primal}}\ge \gamma^{\mathrm{dual}}$.

Moreover, under a standard constraint qualification (Slater's interiority condition) there is no duality
gap. One convenient conic formulation is the following.
\begin{theorem}[Slater's interiority condition]\label{thm:Slater}
Assume that \eqref{eq:conic.primal} is feasible and that there exists a \emph{strictly feasible} point
$\bar x\in \mathsf{ri}\,(\mathcal{C}_1)$ such that $\Phi(\bar x)-b\in \mathsf{ri}\,(\mathcal{C}_2)$.
If $\gamma^{\mathrm{primal}}$ is finite, then strong duality holds, i.e.,
\begin{align*}
    \gamma^{\mathrm{primal}}=\gamma^{\mathrm{dual}},
\end{align*}
and the dual optimum is attained.
\end{theorem}

\subsection{Representation Theory}
\label{sec:rep.theory}
Representation theory provides a natural framework for exploiting symmetries in composite quantum systems, in particular when dealing with many identical copies. A key instance is Schur-Weyl duality, which describes the commuting actions of the unitary group and the symmetric group on tensor powers $\HH^{\otimes k}$, yielding a canonical decomposition into irreducible sectors and the associated Schur basis. These tools underpin several standard primitives in quantum information theory—such as spectrum estimation, permutation-invariant tomography, and entanglement concentration—and they will be used throughout this work. Let us briefly recall the notions and notation we need. For a more detailed treatment, we refer the reader to \cite{We39,FH91}.

Given a group $G$, a (unitary) representation of $G$ on a Hilbert space $\HH$ is a homomorphism $R: G \to \mathrm{U}(\HH)$, i.e., a family of unitaries $\{R_g\}_{g\in G}\subset \BB(\HH)$ satisfying $R_e=\id$ and $R_{gh}=R_gR_h$ for all $g,h\in G$. A subspace $K\subseteq \HH$ is called $G$-invariant (or simply \term{invariant}) if $R_g K\subseteq K$ for all $g\in G$. In this case, the restricted maps $\{R_g|_{K}\}_{g\in G}$ form a \term{subrepresentation}. The representation is \term{irreducible} if its only invariant subspaces are $\{0\}$ and $\HH$. These irreducible representations of a group $G$ play a fundamental role in representation theory.

In particular, every finite-dimensional unitary representation of a compact group is \term{completely reducible}, i.e., it decomposes as a direct sum of irreducible representations (or \term{irreps}), possibly with multiplicities. To relate two representations $R:G\to \mathrm{U}(\HH_1)$ and $S:G\to \mathrm{U}(\HH_2)$, we introduce \term{intertwiners}: a linear map $T:\HH_1\to \HH_2$ is an intertwiner if $T R_g = S_g T$ for all $g\in G$. Schur's lemma characterizes intertwiners between irreducible representations.

\begin{lem}[Schur's Lemma~\cite{FH91}]
Let $R:G\to \mathrm{U}(\HH_1)$ and $S:G\to \mathrm{U}(\HH_2)$ be irreducible representations, and let $T:\HH_1\to \HH_2$ be an intertwiner, i.e.\ $TR_g=S_gT$ for all $g\in G$. Then:
\begin{itemize}
    \item If $R$ and $S$ are not isomorphic, then $T=0$.
    \item If $R$ and $S$ are isomorphic, then $T$ is an isomorphism; in particular, if $\HH_1=\HH_2$ and $R=S$, then $T=\alpha\,\Id_{\HH_1}$ for some $\alpha\in\CC$.
\end{itemize}
\end{lem}

Let us now focus on two the groups that we will use in this work: the unitary group $\UU(d)$ acting \emph{collectively} on all copies, and the symmetric group $S_k$ permuting the $k$ copies \cite{We39,FH91}. Schur--Weyl duality describes the joint representation of these commuting actions on the tensor-power space $(\CC^d)^{\otimes k}$. Let $d=\dim(\CC^d)$.

The \term{unitary group} on $\CC^d$ is given by
\begin{align*}
    \UU(d)\coloneqq \{U\in \BB(\CC^d)\,:\,U^* U=\id\}\,,
\end{align*}
and it acts on $k$ copies via the tensor-power representation $U\mapsto U^{\otimes k}$ on $(\CC^d)^{\otimes k}$.

The \term{symmetric group} $S_k$ is the group of permutations of $[k]\coloneqq\{1,\dots,k\}$ under composition. It acts on $(\CC^d)^{\otimes k}$ through the permutation unitaries $\{R_\pi\}_{\pi\in S_k}$ defined by
\begin{align*}
    R_\pi\!\left(\ket{\psi_1}\otimes\cdots\otimes\ket{\psi_k}\right)
    =\ket{\psi_{\pi^{-1}(1)}}\otimes\cdots\otimes\ket{\psi_{\pi^{-1}(k)}} \, .
\end{align*}

These two actions commute, i.e.\ $(U^{\otimes k})R_\pi=R_\pi (U^{\otimes k})$ for all $U\in \UU(d)$ and $\pi\in S_k$. Since $\UU(d)$ is compact and $S_k$ is finite, their representations are completely reducible. Consequently, Schur-Weyl duality yields the decomposition
\begin{align*}
    (\CC^d)^{\otimes k} \cong \bigoplus_{\lambda \in \mathcal{I}_{(k,d)}}\UU^{(d)}_{\lambda} \otimes \VV_\lambda \, ,
\end{align*}
where $\UU^{(d)}_{\lambda}$ is an irreducible subspace for the action of $\UU(d)$ and $\VV_\lambda$ is the corresponding irreducible subspace for the action of $S_k$. In this decomposition, $\UU(d)$ acts non-trivially only on $\UU^{(d)}_{\lambda}$ and $S_k$ acts non-trivially only on $\VV_\lambda$. The index set $\mathcal{I}_{(k,d)}$ consists of partitions, commonly known as Young coefficients, $\lambda\vdash k$ with at most $d$ parts, i.e.,
\begin{align*}
    \mathcal{I}_{(k,d)}\coloneq \left\{\lambda=(\lambda_1,\ldots,\lambda_d)\in\NN^d:\ 
    \lambda_1\ge \cdots \ge \lambda_d\ge 0,\ \sum_{i=1}^d \lambda_i = k\right\}\,.
\end{align*}
We identify such a partition $\lambda\vdash k$ with its associated Young diagram, consisting of $k$ boxes arranged in left-justified rows of lengths $\lambda_1,\lambda_2,\ldots$. The constraint ``at most $d$ parts'' is equivalently $\ell(\lambda)\le d$, i.e.\ the diagram has at most $d$ rows. We refer to $\lambda$ as the \term{Young coefficient} indexing the corresponding Schur-Weyl block. For further background on Young diagrams and their properties, see \cite{Fu97}.

Therefore, by Schur-Weyl duality there exists an orthonormal basis of $(\CC^d)^{\otimes k}$, the \emph{Schur basis}, which we label as $\ket{\lambda,u,v}$ with $\lambda\in\mathcal{I}_{(k,d)}$, $u\in[\dim \UU^{(d)}_\lambda]$, and $v\in[\dim \VV_\lambda]$. With respect to this basis, the commuting actions of $S_k$ and $\UU(d)$ are block-diagonal:
\begin{align*}
    R_\pi &\cong \bigoplus_{\lambda \in \mathcal{I}_{(k,d)}} \Id_{\dim \UU^{(d)}_\lambda} \otimes R_\pi^{\lambda},\\
    R_U &\cong \bigoplus_{\lambda \in \mathcal{I}_{(k,d)}} R_U^{\lambda} \otimes \Id_{\dim \VV_\lambda} \;,
\end{align*}
where $R_\pi^{\lambda}$ denotes the irreducible representation of $S_k$ on $\VV_\lambda$ and $R_U^{\lambda}$ denotes the irreducible representation of $\UU(d)$ on $\UU^{(d)}_\lambda$. The unitary transformation that maps the computational basis to the Schur basis is known as \term{Schur transform}:
\begin{align*}
    U_{Sch}: \{|k\rangle \otimes |l\rangle\} \mapsto \{|\lambda, u_i\rangle \otimes |\lambda, v_j\rangle\}\,.
\end{align*}

For each $\lambda\in\mathcal{I}_{(k,d)}$, let $\Pi_\lambda$ denote the orthogonal projector onto the Schur--Weyl block $\UU^{(d)}_\lambda\otimes \VV_\lambda$. This projector can be written as
\begin{align*}
    \Pi_\lambda \;=\; \frac{\dim \VV_\lambda}{k!}\sum_{\pi\in S_k} \chi^\lambda(\pi)\, R_\pi \,,
\end{align*}
where $\chi^\lambda(\pi)\coloneqq \Tr\!\left(R_\pi^{(\lambda)}\right)$ is the character of the irrep $\lambda$ of $S_k$. The family $\{\Pi_\lambda\}_{\lambda\in\mathcal{I}_{(k,d)}}$ forms a projective measurement, i.e.\ $\Pi_\lambda\Pi_\mu=\delta_{\lambda\mu}\Pi_\lambda$ and $\sum_{\lambda}\Pi_\lambda=\Id_{(\CC^d)^{\otimes k}}$.

The dimensions of the irreps of $S_k$ and $\UU(d)$ can be expressed as \cite{We39,LREJ25}
\begin{align*}
    \dim \UU^{(d)}_\lambda &= \prod_{1\le i<j\le d}\frac{ \lambda_i -\lambda_j+j-i}{j-i}\,,\\
    \dim \VV_\lambda &= \frac{k!\prod_{1\le i<j\le d} (\lambda_i -\lambda_j+j-i)}{\prod_{i=1}^d(\lambda_i+d-i)!}\,.\\
\end{align*}

We lastly focus on the \emph{symmetric subspace} of $(\CC^d)^{\otimes k}$, i.e.\ the subspace invariant under all permutations of the $k$ tensor factors. In the Schur--Weyl decomposition, it corresponds to the one-row partition $\lambda=(k,0,0,\ldots)$. This subspace is particularly important because Haar averages of $k$-fold tensor powers are supported entirely on it. Then, the \term{symmetric subspace} of $(\CC^d)^{\otimes k}$ is defined as
\begin{align*}
    \mathsf{Sym}^k(\CC^d)\coloneqq \{\ket{\varphi}\in(\CC^d)^{\otimes k}:\ R_\pi\ket{\varphi}=\ket{\varphi}\ \ \forall\,\pi\in S_k\}\,.
\end{align*}
Its orthogonal projector is
\begin{align*}
    \Pi^{(d,k)}_{\mathsf{Sym}}\coloneqq \frac{1}{k!}\sum_{\pi\in S_k} R_\pi \, ,
\end{align*}
and it satisfies $\Tr\left[\Pi^{(d,k)}_{\mathsf{Sym}}\right]=\dim \mathsf{Sym}^k(\CC^d)=\binom{d+k-1}{k}$.

\section{Complexity-constrained Quantum Information}

Quantum information processing implementations are inherently limited to operations that admit efficient circuit descriptions. This has recently motivated complexity-aware variants of standard information measures, defined by restricting the class of physically realizable procedures~\cite{ABV23,CCL+17,avidan2025quantum,AHRA25,MRRLJE26,YHK25,MKN+25,YKH+22,MRRLJE25}. While works such as~\cite{YHK25,MRRLJE25} impose computational constraints by restricting the set of available \emph{measurements}, leading to computational divergences with an operational meaning in efficient state discrimination tasks and resource estimation, we will follow a more general approach. More specifically, throughout this work we place the computational restriction on the \emph{transformations} themselves, that is on the admissible quantum channels. In \Cref{sec:comp-bounded-channels}, we formally introduce this complexity restriction and define the sets of polynomially implementable quantum channels. We then show in \Cref{sec:choi-cone} how one can associate these sets with proper cones, extending the formalism introduced by \cite{YHK25} to the case of quantum channels.

\subsection{Computationally Bounded Quantum Channels}\label{sec:comp-bounded-channels}

In this section, we formalize the notion of complexity-constrained quantum channels. We measure complexity in terms of \term{gate complexity}: informally, given a fixed gate set $\mathcal{G}$, the cost of a transformation is the number of gates from $\mathcal{G}$ required to implement it. Since complexity-theoretic notions are asymptotic, we work with a scaling parameter $n$, denoting the number of registers in the circuit. We consider $n$-qudit systems, with Hilbert space $\HH_n\simeq (\CC^{d})^{\otimes n}$ and $\dim(\HH_n)=d^n$, where $d$ is the local dimension.

Our complexity measure counts only gates. Accordingly, one has unrestricted access to ancillary registers and  measurements in the computational basis. Moreover, for simplicity,  polynomial-time classical (post-)processing is treated as free. When we refer to efficient channels or measurements, we mean those implementable with at most $p(n)$ gates for some polynomial $p$. 

Let us first recall the notion of complexity-constrained measurements introduced in \cite{YHK25,MRRLJE25} via Naimark's dilation (\Cref{Thm: Naimark}). It allows us to represent the generation of POVM elements ---from now on \emph{effect operators}--- as the action of a unitary gate acting on the state space and an auxiliary system followed by a projective measurement in the ancilla.

\begin{definition}[Gate complexity of a POVM]\label{def:POVMcomplexity}
Given a finite gate set $\mathcal{G}$, the \emph{gate complexity} $C(M)$ of a POVM $M\coloneq (E_i)_{i=1}^K$ is the minimal number of gates necessary to implement the unitary operator $U \in \mathcal{B}(\mathcal{H}\otimes \mathcal{H}_{K})$ such that
	\begin{align}
		\Tr \left[ E_i \rho\right] = \Tr \left[ \left(\id_\mathcal{H} \otimes \ketbra{i}{i}_{\mathcal{H}_K} \right) \left(U \left(\rho \otimes  \ketbra{0}{0}_{\mathcal{H}_K} \right) U^* \right)\right]  \qquad  \forall i \in [K] \, ,
	\end{align}
	with $\sum_{i= 1}^K E_i = \id$.
\end{definition}

Similarly to \cite{MRRLJE25,YHK25}, one can define the set of quantum measurements restricted to some $G$ gates and with $K$ outcomes.
\begin{definition}[Set of complexity and output constrained POVMs]\label{def:eff.POVM}
	Given a finite gate set $\mathcal{G}$, for $G \in \NN$ and $K \geq 2$, the set of $G$-complexity $K$-outcome POVMs is defined as,
	\begin{align}
		\mathbf{M}^G_{n,K}\coloneq  \left\{M=(E_i)_{i=1}^K \; \Big| \; C(M) \leq G  , \, \text{dim}(\mathcal{H}_K) =K, \, E_i \in \Pos(\mathcal{H}_n)\right\} \,,
	\end{align}
	where $C(M)$ is the complexity of the POVM $M$. When $G$ is taken as a polynomial $p(n): \NN \rightarrow\NN$, then we write $\mathbf{M}^{p(n)}_{n,K}$.
\end{definition}

We fix a gate set $\mathcal{G}$ and a polynomial gate budget $p(n)$.
\begin{definition}[The set of efficient binary effects \cite{YHK25}]\label{def:eff-binary-effects}
	For each dimension $n$, let $\Eeff_n \subseteq \BB(\HH_n)$ denote the set of POVM effects
	$E$ such that $0\le E\le \id$ and the binary POVM $\{E,\id-E\}$ is implementable by a quantum
	circuit of size at most $p(n)$ over $\mathcal{G}$. We suppress the dependence on
	$\mathcal{G}$ and $p(n)$ in the notation.
\end{definition}

In \cite{YHK25}, the set and cone of efficiently generated quantum effect operators were introduced to formalize the measurements accessible to a computationally bounded observer. We now generalize this framework from measurements to quantum channels. For this, we use the Choi--Jamio\l kowski representation, but in a form tailored to our perspective: the relevant Choi operators are those of Heisenberg-picture maps, i.e., of completely positive unital (CPU) maps, rather than of Schr\"odinger-picture CPTP maps as in the standard convention (see \Cref{def:Choi}). Since our notion of efficient implementability is defined from the viewpoint of a bounded observer, the Heisenberg picture is the natural one to use.

\begin{definition}[Choi-Jamiołkowski representation of CPU maps]
Let $\HH^A_n, \HH^{A'}_n$ be finite-dimensional Hilbert spaces with $\dim \HH^A_n = d_A$, and let $\HH^A_n \cong \HH^{A'}_n$. Let $T^*: \BB(\HH^{A^\prime}_n) \to\BB(\HH^{B}_n)$ be a completely positive unital (CPU) map. For such a $T^*$, its Choi-Jamiołkowski operator representation on $AB$ is given by
	\begin{align*}
		J(T^*) \coloneq  d_A (\Id_A \otimes T^*) (\ketbra{\Omega_{AA'}}{\Omega_{AA'}})\,,
	\end{align*}
	where $\ket{\Omega_{AA'}} = \frac{1}{\sqrt{d_A}}\sum_{x \in [d_A]} \ket{x}_A \otimes \ket{x}_{A'}$.
\end{definition}

Since the Choi-Jamiołkowski representation is a linear bijection, the action of the CPU map $T^*$ can be recovered from its Choi-operator $J(T^*)$ in the following way,
\begin{align}\label{eq:Choi.action.operator}
	T^*(X) = \tr_A[J(T^*)(X^T \otimes \id_B)]\,,
\end{align}
where the transpose is taken with respect to the computational basis used to define the state $\ket{\Omega_{AA'}}$.

Similar to \Cref{Thm: Naimark}, one can represent a completely positive unital (CPU) map as a unitary (or, more generally, isometric) evolution on a larger space, followed by discarding an environment. In the Heisenberg picture, this is captured by the following Stinespring representation.

\begin{theorem}[Stinespring representation]\label{thm:Stinespring}
Let $T:\BB(\HH_{B})\rightarrow \BB(\HH_{A^\prime})$ be a completely positive trace preserving (CPTP) linear map, and let $T^*:\BB(\HH_{A^\prime})\to \BB(\HH_{B})$ denote its dual map. Then there exists an integer $r\le d_{A^\prime} d_B$ and an isometry $V:\HH_{B}\rightarrow \HH_{A^\prime}\otimes \CC^r$ such that
\begin{align*}
    T^*(X) &= V^*(X \otimes \id_r)V \\
    T(\rho) &= \tr_r[V\rho V^*]
\end{align*}
for all $X \in \BB(\HH_{A^\prime})$ and $\rho \in \BB(\HH_{B})$.
\end{theorem}

This representation allows us to model the circuit complexity of the CPTP and CPU maps.

\begin{definition}[Gate complexity of quantum maps]\label{def:CPUcomplexity} 
Let $T^*: \BB(\HH^{A^\prime}_n) \to \BB(\HH^{B}_n)$ be a completely positive and unital (CPU) linear map. Given a finite gate set $\mathcal G$, the \emph{gate complexity} $C(T^*)$ is the minimum number of $\mathcal G$-gates such that there exist ancillary spaces $\HH_E,\HH_K$ such that $\dim(\HH^{B}_n) \dim( \HH_E) = \dim(\HH^{A^\prime}_n) \dim( \HH_K) $ and a unitary circuit $U:\HH^{B}_n\otimes \HH_E \to \HH^{A^\prime}_n\otimes \HH_K$ generated by $C(T^*)$ gates from $\mathcal{G}$, such that
	\begin{align}
		T^*(X) = \bra{0}_E\, U^*\,(X\otimes \id_K)\,U\,\ket{0}_E
		\qquad \forall\,X\in \BB(\HH^{A^\prime}_n)\,.
	\end{align}
Completely analogously we say that the \emph{gate complexity} $C(T)$ of a completely positive trace-preserving (CPTP) map $T: \BB(\HH^{B}_n)\to \BB(\HH^{A^\prime}_n)$ is the minimal number of $\mathcal{G}$-gates required to implement a Stinespring dilation of $T$, i.e., such that there exist ancillary spaces $\HH_E,\HH_K$ such that $\dim(\HH^{B}_n) \dim( \HH_E) = \dim(\HH^{A^\prime}_n) \dim( \HH_K) $ and a unitary circuit $U:\HH^{B}_n\otimes \HH_E \to \HH^{A^\prime}_n\otimes \HH_K$ generated by $C(T)$ gates from $\mathcal{G}$, such that
\begin{align}
T(\rho)=\tr_K[U(\rho\otimes|0\rangle\langle 0|_E)U^*]\qquad \forall\,\rho\in \BB(\HH^{B}_n)\,.
\end{align}
\end{definition}
\begin{remark}
This notion of complexity of linear maps, allowing for arbitrary ancillary systems initiated in $\ketbra{0}{0}$ which in the end are traced-out and only counting the gate complexity of the intermitted unitary, is the same as in \cite{MRRLJE25, YHK25} and closely related to the one of \cite[Definition 18]{AHRA25}; for the latter see also~\cite{avidan2025quantum,CCL+17}. There, the authors consider circuit complexity of adaptive quantum circuits allowing for arbitrary ancillas, polynomial unitary gates, mid-circuit computational measurements, and arbitrary final computational measurements. Up to at most a polynomial overhead, the notion of complexity  we use is equivalent to~\cite[Definition 18]{AHRA25}. This is because the additional operations of mid-circuit measurement can be emulated using, e.g., extra ancillas and polynomially many CNOT gates.
\end{remark}

In particular, in this model, the gate complexity of a CPU map $T^*$ coincides with that of its dual CPTP map $T$.

\begin{lem}\label{lem:complexity.dual.equivalence}
Let $T : \BB(\HH^B_n) \to \BB(\HH^{A^\prime}_n)$ be a CPTP map then $C(T)=C(T^*)$.
\end{lem}

\begin{proof}

Let the CPTP map $T:\BB(\HH^B_n) \to \BB(\HH^{A^\prime}_n)$ be efficiently implementable with polynomial complexity $p(n)$. 
I.e. there exist ancillary spaces $\HH_E,\HH_K$, an ancilla state $\ket{0}_E\in \HH_E$, and a $p(n)$-depth unitary
	\begin{align*}
		U : \HH^B_n \otimes \HH_E \longrightarrow \HH^{A^\prime}_n \otimes \HH_K
	\end{align*} 
such that
\begin{align}
T(\rho)=\tr_K[U(\rho\otimes|0\rangle\langle 0|_E)U^*] = \tr_k[V\rho V^*]
\end{align}
where the Stinespring isometry $V : \HH^B_n \to \HH^{A^\prime}_n \otimes \HH_K$ is given by
	\[
		V \ket{\psi}_B \coloneq  U\left(\ket{\psi}_B \otimes \ket{0}_E\right)\,.	
	\]
	Therefore the dual map $T^* : \BB(\HH^{A^\prime}_n) \to \BB(\HH^B_n)$ is given by
	\[
		T^*(X) = V^* (X \otimes \id_K) V = \langle 0|_E U^*(X\otimes\id_K)U |0\rangle_E\,.
	\]
	Moreover, $T^*$ is specified by the same isometry $V$ (equivalently, by the same unitary $U$ and ancilla initialization) and thus satisfies $C(T^*)\leq p(n)$. 
    The converse follows by minimality of $C(T)$, in the sense that if there was some Stinespring dilation of $T^*$ with a complexity $C<p(n)$, then this would, by the above argument in reverse, yield a Stinespring dilation of $T$ with the a complexity $C<p(n)=C(T)$, which is not possible.
\end{proof}

Moreover, after a suitable rescaling, Choi--Jamio\l{}kowski operators of CPU maps can be viewed as effect operators on the space $\BB(\HH^A_n \otimes \HH^B_n)$.

\begin{lem}[Choi-Jamiołkowski representation of CPU maps and effect operators]\label{prop:CHoi.effect}
Let $T^*:\BB(\HH^{A^\prime}_n)\to \BB(\HH^{B}_n)$ be a completely positive and unital (CPU) map, and let $J(T^*)$ denote its Choi operator. Define the normalized Choi operator
\begin{align*}
    E_{AB} \coloneq \frac{J(T^*)}{d_A} \in \BB(\HH^A_n \otimes \HH^B_n)\,.
\end{align*}
Then $E_{AB}$ is an effect operator which satisfies
\begin{align}
    & i)\;0\le E_{AB}\le \id_{AB}\,, \\
    & ii)\;\tr_A[E_{AB}]=\frac{1}{d_A}\id_B\, .
\end{align}
\end{lem}

\begin{proof}
Let $\ket{\Omega_{AA'}}\coloneq\frac{1}{\sqrt{d_A}}\sum_{i=1}^{d_A}\ket{i}_A\otimes \ket{i}_{A'}$, so that
\begin{align*}
    J(T^*) \coloneq d_A(\id_A\otimes T^*)\!\left(\ketbra{\Omega_{AA'}}{\Omega_{AA'}}\right).
\end{align*}
Since $\ketbra{\Omega_{AA'}}{\Omega_{AA'}}\le \id_{AA'}$, complete positivity implies
\begin{align*}
    0\le J(T^*) \le d_A(\id_A\otimes T^*)(\id_{AA'}) = d_A\,\id_A\otimes T^*(\id_{A'}) = d_A\,\id_{AB},
\end{align*}
where we use the fact that $T^*$ is unital in the last step. Dividing by $d_A$ yields $0\le E_{AB}\le \id_{AB}$, proving (i).

For (ii), it follows from Eq.~\eqref{eq:Choi.action.operator} that
\begin{align*}
    \tr_A\!\left[J(T^*)\right]=T^*(\id_{A'})=\id_B\,.
\end{align*}
Therefore, we have that $\tr_A[E_{AB}]=\id_B/d_A$.
\end{proof}

Next, we relate the gate complexity of a CPU map to that of the associated binary test on $AB$ obtained from its (normalized) Choi operator.

\begin{lem}[Circuit complexity of Choi effect operators]\label{lem:choi.effect.complexity}
Let $T^*: \BB(\HH^{A^\prime}_n)\to \BB(\HH^{B}_n)$ be a CPU map with gate complexity $C(T^*)$ as in \Cref{def:CPUcomplexity}, and define its normalized Choi effect operator
\begin{align*}
    E_{AB}\coloneq\frac{1}{d_A}\,J(T^*)=(\Id_A\otimes T^*)\!\left(\ketbra{\Omega_{AA'}}{\Omega_{AA'}}\right)\in \BB(\HH^A_n\otimes \HH^B_n)\,,
\end{align*}
where $\ket{\Omega_{AA'}}=\frac{1}{\sqrt{d_A}}\sum_{x\in[d_A]}\ket{x}_A\otimes \ket{x}_{A'}$. Then $M_{AB}\coloneq\{E_{AB},\,\id_{AB}-E_{AB}\}$ is a valid binary POVM and is implementable with at most
\begin{align*}
    C(M_{AB}) \le C(T^*) + C(U_\Omega)
\end{align*}
$\mathcal{G}$-gates, where $U_\Omega$ denotes a circuit implementing the Bell test $\{\Omega_{AA'},\id_{AA'}-\Omega_{AA'}\}$. Moreover, for every state $\rho_{AB}$, $M_{AB}$ outputs $0$, i.e., the measurement outcome associated with the POVM $E_{AB}$, with probability $\Tr[\rho_{AB}E_{AB}]$.
\end{lem}

\begin{proof}
By \Cref{def:CPUcomplexity} there exist ancillas $\HH_E,\HH_K$, a fixed $\ket{0}_E\in\HH_E$, and a unitary
\begin{align*}
    U:\HH^{B}_n\otimes \HH_E \longrightarrow \HH^{A'}_n\otimes \HH_K
\end{align*}
implemented with $C(T^*)$ gates. Define the isometry $V:\HH^B_n\to \HH^{A'}_n\otimes \HH_K$ by
\begin{align*}
    V\ket{\psi}_B \coloneq U\left(\ket{\psi}_B\otimes \ket{0}_E\right)\qquad \forall\,\ket{\psi}_B\in \HH^B_n\,,
\end{align*}
so that
\begin{align*}
    T^*(X)=V^*(X\otimes \id_K)V\qquad \forall\,X\in\BB(\HH^{A'}_n)\,.
\end{align*}

Consider the following measurement procedure on input $\rho_{AB}$:
\begin{enumerate}
    \item Apply $V$ to register $B$, producing registers $A'K$.
    \item Perform the two-outcome projective measurement $\{\Omega_{AA'},\,\id_{AA'}-\Omega_{AA'}\}$ on $AA'$ (discard $K$).
\end{enumerate}
Step (1) can be implemented with $C(T^*)$ gates and step (2), with $C(U_\Omega)$ gates. The acceptance probability is
\begin{align*}
    \Tr\!\left[(\Omega_{AA'}\otimes \id_K)\,(\id_A\otimes V)\rho_{AB}(\id_A\otimes V^*)\right]
    &=\Tr\!\left[\rho_{AB}\,(\id_A\otimes V^*)(\Omega_{AA'}\otimes \id_K)(\id_A\otimes V)\right] \\
    &=\Tr\!\left[\rho_{AB}\,(\Id_A\otimes T^*)(\Omega_{AA'})\right]
    =\Tr[\rho_{AB}E_{AB}]\,,
\end{align*}
which proves the claim.
\end{proof}

Therefore, the complexity of implementing the CPU map $T^*$ upper bounds—up to the additive overhead $C(U_\Omega)$ for implementing the Bell measurement—the complexity of implementing the binary measurement $\{E_{AB},\id_{AB}-E_{AB}\}$ on $\BB(\HH^A_n\otimes \HH^B_n)$, where $\tr_A[E_{AB}]=\id_B/d_A$.

\subsection{The cone of efficiently implementable Choi operators}\label{sec:choi-cone}

Let us now define the set and cone of complexity constrained Choi operators.

\begin{definition}[Set of complexity-constrained Choi operators]\label{def:choi.set.eff}
Given a finite, discrete gate set $\mathcal{G}$, the \emph{set of $G$-complexity Choi operators} $\mathbf{J}^{G}_{n,A|B}$ is defined as
\begin{align}
	\mathbf{J}^{G}_{n,A|B}
	\coloneq\left\{\, J(T^*) \ \Big|\ J(T^*) \coloneq d_A (\Id_A \otimes T^*) \left(\ketbra{\Omega_{AA'}}{\Omega_{AA'}}\right),\ \ C(T^*) \leq G \right\},
\end{align}
where $C(T^*)$ denotes the circuit complexity of implementing the CPU map $T^*:\BB(\HH^{A^\prime}_n)\to \BB(\HH^{B}_n)$ given the gate-set $\mathcal{G}$.
When $G$ is a polynomial $p(n):\NN\to\NN$, we write $\mathbf{J}^{p(n)}_{n,A|B}$.
\end{definition}

\begin{lem}[Counting exact bounded-size channel circuits]\label{lem:counting-choi}
Fix a finite, discrete gate set $\mathcal{G}$. Since $\mathcal{G}$ is finite, there exists
$\ell\in\NN$ such that every gate in $\mathcal{G}$ acts nontrivially on at most $\ell$ registers.
Then, for every $n,G\in\NN$, the set $\mathbf{J}^{G}_{n,A|B}$ is finite. More precisely,
\begin{align}
    \left|\mathbf{J}^{G}_{n,A|B}\right|
    \leq
    \sum_{m=0}^{G}
    \left(|\mathcal{G}|(n+\ell G)^{\ell}\right)^m
    \leq
    2^{\,c\,G\log(n+G)}
\end{align}
for some constant $c>0$ depending only on $\mathcal{G}$ and $\ell$.
In particular, if $G=p(n)$ is polynomial, then there exists a polynomial $r(n)$ such that
\begin{align}
    \left|\mathbf{J}^{p(n)}_{n,A|B}\right| \leq 2^{r(n)} .
\end{align}
\end{lem}

\begin{proof}
By \Cref{def:CPUcomplexity}, every element of $\mathbf{J}^{G}_{n,A|B}$ is the Choi operator of a CPU map
$T^*:\BB(\HH^{A^\prime}_n)\to\BB(\HH^B_n)$ implemented by a unitary circuit $U$ of size at most $G$,
possibly using ancillary registers.

Although ancillary registers are unrestricted in principle, only those on which at least one gate acts
are relevant. Hence every such circuit is equivalent to one in which all ancillas on which no gate acts
have been removed. Moreover, the remaining ancillas may be relabelled according to the first gate in
which they appear. Since each gate acts on at most $\ell$ registers, a circuit with at most $G$ gates
uses at most $\ell G$ active ancilla registers. Therefore every such circuit admits a canonical
description on at most
\begin{align}
    N\coloneq n+\ell G
\end{align}
registers.

For each of the $m\leq G$ gate positions, one chooses:
\begin{enumerate}
    \item the gate in $\mathcal{G}$, giving at most $|\mathcal{G}|$ possibilities;
    \item the ordered tuple of registers on which it acts, giving at most $N^\ell$ possibilities.
\end{enumerate}
Hence the number of exact circuit descriptions of length $m$ is at most
\begin{align}
    \left(|\mathcal{G}|\,N^\ell\right)^m
    =
    \left(|\mathcal{G}|(n+\ell G)^\ell\right)^m .
\end{align}
Summing over all $m\leq G$ gives
\begin{align}
    \left|\mathbf{J}^{G}_{n,A|B}\right|
    \leq
    \sum_{m=0}^{G}
    \left(|\mathcal{G}|(n+\ell G)^\ell\right)^m .
\end{align}
Since distinct Choi operators are fewer than distinct circuit descriptions, this proves finiteness.

Finally, because $\ell$ and $|\mathcal{G}|$ are constants,
\begin{align}
    \sum_{m=0}^{G}
    \left(|\mathcal{G}|(n+\ell G)^\ell\right)^m
    \leq
    (G+1)\left(|\mathcal{G}|(n+\ell G)^\ell\right)^G
    \leq
    2^{\,c\,G\log(n+G)}
\end{align}
for some constant $c>0$. Taking $G=p(n)$ yields the claim.
\end{proof}

Therefore, for every polynomial $p(n)$, the family
$\{\mathbf{J}^{p(n)}_{n,A|B}\}_{n\in\NN}$ consists of at most $2^{r(n)}$ Choi operators for some polynomial $r(n)$. In particular, for each fixed $n$, the set $\mathbf{J}^{p(n)}_{n,A|B}$ is finite. There is a direct connection between this family and the family of efficiently generated effect operators $\{\Eeff_n\}_{n \in \NN}$ defined in \cite{YHK25}.

\begin{lem}\label{lem:CPU.Effect.relation}
Given a gate set $\mathcal{G}$, let $\Eeff_n \subseteq \BB(\HH^B_n)$ denote the set of efficiently generated binary effect operators from~\cite[Definition~3.8]{YHK25}, i.e.\ those $F_B$ for which the binary POVM $\{F_B,\id_B-F_B\}$ is implementable by a circuit of polynomial size $q(n)$.
Let $\HH^{A^\prime}_n \cong \CC^2$ with computational basis $\{\ket{0},\ket{1}\}$. Then, for every $F_B\in \Eeff_n$ there exists a CPU map $T^*:\BB(\HH^{A^\prime}_n)\to \BB(\HH^B_n)$ with $C(T^*)\le q(n)$ and Choi operator $J_{AB}\in \mathbf{J}^{q(n)}_{n,A|B}$ such that
\begin{align*}
    F_B
    = T^*(\ketbra{0}{0}_{A^\prime})
    = \tr_A\!\left[J_{AB}(\ketbra{0}{0}_A\otimes \id_B)\right]
    = \bra{0}_A J_{AB}\ket{0}_A \, .
\end{align*}
Similarly, $\id_B-F_B = T^*(\ketbra{1}{1}_{A^\prime})$. In particular,
\begin{align*}
    \Eeff_n \subseteq \left\{\, \bra{i}_A J_{AB}\ket{i}_A \ \Big|\ J_{AB}\in \mathbf{J}^{q(n)}_{n,A|B},\ i\in\{0,1\}\right\}.
\end{align*}
The same inclusion holds more generally for any choice of $\dim(\HH^{A^\prime}_n)\ge 2$.
\end{lem}

\begin{proof}
Fix $F_B\in \Eeff_n$. By definition, there exists a circuit of size $q(n)$ implementing the binary POVM $\{F_B,\id_B-F_B\}$. By \Cref{Thm: Naimark}, we may assume this circuit can be described in the following way
\begin{align*}
    F_B=\bra{0}_{A'}\,U^*(\id_B\otimes \ketbra{0}{0}_{A'})U\,\ket{0}_{A'}\, .
\end{align*}
Define $T^*:\BB(\HH^{A'}_n)\to \BB(\HH^B_n)$ by
\begin{align*}
    T^*(X)\coloneq\bra{0}_{A'}\,U^*(\id_B\otimes X)U\,\ket{0}_{A'} \qquad \forall\,X\in \BB(\HH^{A'}_n)\,.
\end{align*}
Then $T^*$ is completely positive and unital, and it satisfies $T^*(\ketbra{0}{0}_{A'})=F_B$ and $T^*(\ketbra{1}{1}_{A'})=\id_B-F_B$. Moreover, by \Cref{lem:complexity.dual.equivalence}, $C(T^*)\le q(n)$ since it is implemented by the same circuit. The rest of the proof follows from \Cref{eq:Choi.action.operator}.
\end{proof}

Then, for every $n \in \NN$, the set of efficient measurements $\Eeff_n$ implementable with $p(n)$ gates is included (as a representation of its corresponding quantum channel) in the set of complexity constrained Choi operators $\mathbf{J}^{p(n)}_{n,A|B}$.

Similarly to \cite{YHK25}, one can define the family of cones generated by the family of sets \(\{\mathbf{J}^{p(n)}_{n,A|B}\}_{n\in\NN}\). In contrast to the measurement setting considered there, these cones are not meant to capture efficient measurements on \(\BB(\HH^A_n)\). Rather, they capture the operator representations of efficiently implementable quantum channels $T:\BB(\HH^B_n)\to \BB(\HH^{A'}_n)$.

\begin{definition}[Family of the cones of efficiently generated Choi operators ]\label{def:fully.quantum.cone}
Given a Hilbert space $\HH^A_n \otimes \HH^B_n$, the cone of efficiently generated Choi operators is given by
	\begin{align}\label{equ:qq-cone}
		\cC^{\mathbf{E}_\mathrm{eff}}_{n, A|B}=\cone(\mathbf{J}^{p(n)}_{n,A|B}) = \bigcup_{\lambda \geq 0} \lambda \cdot \conv\big(\mathbf{J}^{p(n)}_{n,A|B}\big)\,.
	\end{align}
where $\conv\big(\mathbf{J}^{p(n)}_{n,A|B}\big)$ is the convex hull of the set $\mathbf{J}^{p(n)}_{n,A|B}$.
\end{definition}

\begin{lem}[Normalization of Choi operators]\label{lem:choi-normalization}
	For any CPU map $T^*$, the Choi operator satisfies $J(T^*)\ge 0$ and $\tr_A[J(T^*)]=\id_B$.
	Consequently every $F_{AB}\in\cC^{\mathbf{E}_\mathrm{eff}}_{n,A|B}$ obeys
	$\tr_A[F_{AB}]=\lambda\,\id_B$ for some $\lambda\ge 0$.
\end{lem}

\begin{proof}
	Complete positivity gives $J(T^*)\ge 0$, and
	$\tr_A[J(T^*)]=d_A\,T^*(\tr_A[\ketbra{\Omega_{AA^\prime}}{\Omega_{AA^\prime}}])=T^*(\id_{A^\prime})=\id_B$.
	The trace condition extends by linearity and the conic closure.
\end{proof}

Throughout this paper we will work with two normalized bases of the same cone $\cC^{\mathbf{E}_\mathrm{eff}}_{n,A|B}$. The first one is the \emph{channel-normalized} slice,
\begin{align*}
\conv\big(\mathbf{J}^{p(n)}_{n,A|B}\big)\coloneq\left\{F_{AB}\in \cC^{\mathbf{E}_\mathrm{eff}}_{n,A|B}\,:\,\tr_A[F_{AB}]=\id_B\right\},
\end{align*}
which corresponds to the convex hull of efficiently generated Choi operators. The second one is the \emph{effect-normalized} slice,
\begin{align*}
\Eeff_{n, A|B}\coloneq\left\{F_{AB}\in \cC^{\mathbf{E}_\mathrm{eff}}_{n,A|B}\,:\,\tr_A[F_{AB}]=\id_B/d_A\right\},
\end{align*}
which arises from the normalized Choi effect operators. Nevertheless, as discussed in \Cref{lem:choi-normalization}, these two normalizations differ only by a positive scalar rescaling, and hence generate the same proper cone. Which normalization is preferable depends on the application.

\begin{lem}\label{lem:proper.cone}
	Let
    \begin{align*}
        	\VV_{n,A|B}\;\coloneq\;\left\{X\in \Herm(\HH^A_n\otimes \HH^B_n)\;:\;\tr_A[X]\in \RR\cdot \id_{B} \right\}
    \end{align*}
be the subspace spanned by Hermitian operators whose partial trace is proportional to the identity. If
    \begin{align*}
        		\Span_{\RR}\{\,\mathbf{J}^{p(n)}_{n,A|B}\,\}\;=\;\VV_{n,A|B}
    \end{align*}
	for every $n \in \NN$, then the family of cones $\{ \cC^{\mathbf{E}_\mathrm{eff}}_{n, A|B}\}_{n \in\NN}$ is a family of proper cones with respect to the subspaces $\VV_{n,A|B}$, i.e., each $\cC^{\mathbf{E}_\mathrm{eff}}_{n, A|B}$ is closed, convex, pointed, and solid in $\VV_{n,A|B}$.
\end{lem}

\begin{proof}
	Convexity is clear from the definition. By \Cref{lem:counting-choi}, the set $\mathbf{J}^{p(n)}_{n,A|B}$ is finite for every $n$. Hence $\cC^{\mathbf{E}_\mathrm{eff}}_{n,A|B} =\cone(\mathbf{J}^{p(n)}_{n,A|B})$ is a finitely generated cone, and therefore it is closed. Pointedness follows since every generator $J(T^*)$ is positive semidefinite, since $T^*$ are all completely positive, and $\cC^{\mathbf{E}_\mathrm{eff}}_{n,A|B}\subseteq \Pos(\HH^A_n\otimes\HH^B_n)$, hence $\cC^{\mathbf{E}_\mathrm{eff}}_{n,A|B}\cap (-\cC^{\mathbf{E}_\mathrm{eff}}_{n,A|B})=\{0\}$. Moreover, for CPU maps one has $\tr_{A_n}J(T^*)=\id_{B_n}$, so $\cC^{\mathbf{E}_\mathrm{eff}}_{n,A|B}\subseteq \VV_{n,A|B}$.

	Lastly, solidness holds in $\VV_{n,A|B}$ as the assumption implies
	\begin{align*}
		\Span_{\RR}\left(\cC^{\mathbf{E}_\mathrm{eff}}_{n,A|B}\right)
		=\Span_{\RR}\{\mathbf{J}^{p(n)}_{n,A|B}\}
		=\VV_{n,A|B}\,,
	\end{align*}
	(using $\Span(\cone(\conv(S)))=\Span(S)$ for nonempty $S$). Hence, by the standard full-dimensionality  argument (cf.\ \cite[Section 2.1.3]{BV04}), $\mathrm{ri}_{\VV_{n,A|B}}(\cC^{\mathbf{E}_\mathrm{eff}}_{n,A|B})\neq\emptyset$, where $\mathrm{ri}$ denotes the relative interior.
\end{proof}

The solidness condition
\begin{align*}
	\Span_{\RR}\left(\cC^{\mathbf{E}_\mathrm{eff}}_{n,A|B}\right)
	=\Span_{\RR}\{\mathbf{J}^{p(n)}_{n,A|B}\}
	=\VV_{n,A|B}\,,
\end{align*}
can be seen as an informational completeness condition at the level of Choi representations: every Choi operator associated with a CPU map (efficient or not) lies in $\VV_{n,A|B}$ and can be spanned by Choi operators of efficiently implementable CPU maps, i.e.,
\begin{align*}
	\forall X \in \VV_{n,A|B}\; \exists\, c_i \in \RR,\, J^i_{AB} \in \mathbf{J}^{p(n)}_{n,A|B}\ \text{s.t.}\
	X = \sum_i c_i\, J^i_{AB}\,.
\end{align*}

\begin{example}[An explicit example of a set of informationally complete quantum channels]\label{rem:example.info.complet}
In order to show the feasibility of the informational completeness assumption under complexity constraints, let us provide an example. Fix $n\in\NN$ and set of symbols $\mathcal{S}_n\coloneq\{0,1,+,+i\}^n$, where its corresponding quantum states are defined as $|+\rangle\coloneq 1/\sqrt{2}(|0\rangle+|1\rangle)$ and $|+i\rangle\coloneq 1/\sqrt{2}(|0\rangle+i|1\rangle)$. Then, let us define the product states on $\HH^{A\prime}_n\simeq(\CC^2)^{\otimes n}$ by
\begin{align*}
    \sigma_s\coloneq\bigotimes_{k=1}^n |s_k\rangle\!\langle s_k|\,,
\end{align*}
and set $\sigma_{0^n}\coloneq|0\rangle\langle 0|^{\otimes n}$. Let $\mathsf{X},\mathsf{Y},\mathsf{Z}$ denote the standard Pauli gates.
Then, for each Pauli string $Q \in \{\pm\id, \pm \mathsf{X}, \pm \mathsf{Y}, \pm \mathsf{Z} \}^{\otimes n}$, let us now define the binary POVM $\{E_Q,\id-E_Q\}$ with 
\begin{align*}
      E_Q\coloneq \frac{\id+Q}{2}\,.
\end{align*}
Assuming that the POVMs $\{\{E_Q,\id-E_Q\}\}_Q$ and the states $\{\sigma_s\}_{s \in \mathcal{S}_n}$ are efficiently implementable, then the measurement map $T_{s,Q}: \BB(\HH_n^{B}) \to \BB(\HH_n^{A^\prime})$ defined by
\begin{align*}
    T_{s,Q}(\rho) \coloneq \Tr[E_Q \rho]\sigma_s + \Tr[(\id-E_Q )\rho]\sigma_{0^n}
\end{align*}
is also efficiently implementable for every $s\in\mathcal{S}_n$ and $Q$.

Note that this linear map is nothing but a measure-and-prepare channel, where the set of measurements is the tomographically complete set of Pauli measurements, and the prepared states are products of Pauli eigenstates. One can easily verify that it is a CPTP map and, thus, its dual, a CPU map. Its corresponding Choi operators can be written as
\begin{align*}
J(T_{s,Q}^*)&=\sigma_s^{\,\mathsf T}\otimes E_Q+\sigma_{0^n}^{\,\mathsf T}\otimes(\id-E_Q) \\
&= \sigma_{0^n} \otimes \id + (\sigma_s - \sigma_{0^n})^{\,\mathsf T} \otimes E_Q\,,
\end{align*}
where ${\mathsf T}$ denotes transpose in the computational basis. Therefore,
\begin{align*}
    \Span_{\RR}\{J(T_{s,Q}^*): s,Q\} = \RR\cdot(\id_A \otimes \id_B) \oplus \Herm_0(\HH_{A_n}) \otimes \Herm(\HH_{B_n}) =\VV_{n,A|B}\,.
\end{align*}
\end{example}

Let us now define the dual subspace of the channel subspace $\VV_{n,A|B}$. Following \cite{Jen13,Jen14}, the algebraic dual $\VV^*_{n,A|B}$ can be identified with a quotient of $\Herm(\HH^A_n \otimes \HH^B_n )$ via the Hilbert-Schmidt product. Let us define the orthogonal complement of $\VV_{n,A|B}$:
\begin{align*}
	\VV^{\perp}_{n,A|B}\coloneq  \left\{Y \in \Herm(\HH^A_n \otimes \HH^B_n ) \, |\, \Tr[YX] = 0 \; \forall X \in \VV_{n,A|B} \right\}\,.
\end{align*}
Let $\pi: \Herm(\HH^A_n \otimes \HH^B_n ) \rightarrow\Herm(\HH^A_n \otimes \HH^B_n )/\VV^{\perp}_{n,A|B}$, be the quotient map $\pi(Y) = [Y]_{\VV^{\perp}_{n,A|B}}$ assigning to each $Y$ the quotient class it falls in; see \cite{Jen13, Jen14} for more details. Then $\pi$ induces an isomorphism $\VV^*_{n,A|B} \cong \Herm(\HH^A_n \otimes \HH^B_n )/\VV^{\perp}_{n,A|B}$, where the duality between $\VV_{n,A|B}$ and $\VV^*_{n,A|B}$ is given by
\begin{align*}
	\langle \pi(Y), X \rangle = \Tr[YX], \qquad Y \in \Herm(\HH^A_n \otimes \HH^B_n ),\, X \in \VV_{n,A|B}\,.
\end{align*}
In particular, two operators $X$ and $X' \in \Herm(\HH^A_n \otimes \HH^B_n )$ define the same functional on the effect subspace $\VV_{n,A|B}$ if
\begin{align*}
	X-X' = \id_A \otimes Z, \qquad \Tr[Z] = 0\,
\end{align*}
because for any $Y\in \VV_{n,A|B}$ then it follows that
\begin{align*}
    \Tr[Y(X-X^\prime)] = \Tr[Y(\id_A\otimes Z)] = \Tr[\tr_A[Y]Z]\propto \Tr[Z]=0.
\end{align*}

\begin{definition}[Dual cone of quantum states]
	Let $\cC^{\mathbf{E}_\mathrm{eff} }_{n, A|B}\subseteq \VV_{n,A|B}$ be the cone of efficiently generated Choi operators. Its dual cone is defined as
	\begin{align}
		\cC^{\mathcal{S}_\mathrm{eff}}_{n, A|B}\coloneq  \cC^{\mathbf{E}_\mathrm{eff} \, *}_{n, A|B} = \left\{ \pi(Y) \, | \, Y \in  \Herm(\HH^A_n \otimes \HH^B_n ), \, \Tr[YX]\geq0 \; \forall X \in  \cC^{\mathbf{E}_\mathrm{eff}}_{n, A|B}  \, \right\}\,,
	\end{align}
	where $\pi(Y)= Y + \VV^{\perp}_{n,A|B}$ is the quotient representation of the operator $Y \in  \Herm(\HH^A_n \otimes \HH^B_n )$.
\end{definition}

Then, the dual cone contains (the image of) every quantum state: if $\cC^{\mathbf{E}_\mathrm{eff}}_{n,A|B}\subseteq \Pos(\HH^A_n\otimes \HH^B_n)$, then $\pi(\rho_{AB})\in \cC^{\mathcal{S}_\mathrm{eff}}_{n,A|B}$ for every $\rho_{AB}\ge 0$, since $\Tr[\rho_{AB}F_{AB}]\ge 0$ for all $\rho_{AB}\ge 0$ and all $F_{AB}\ge 0$. Thus, whenever $\cC^{\mathbf{E}_\mathrm{eff}}_{n,A|B}$ is a proper cone with respect to $\VV_{n,A|B}$, then $\cC^{\mathcal{S}_\mathrm{eff}}_{n,A|B}$ is also proper with respect to $\VV^*_{n,A|B}$.

Under the informational-completeness condition, $\cC^{\mathcal{S}_\mathrm{eff}}_{n, A|B}$ is a proper cone, and hence we can define a partial order on $\VV^*_{n,A|B}$, via
\begin{align}\label{eq:cone.partial.order}
	\pi(\rho_{AB}) \leq_{\cC^{\mathcal{S}_\mathrm{eff}}_{n, A|B}}\pi(\sigma_{AB}) \, \iff \Tr[X \rho_{AB}] \leq \Tr[X \sigma_{AB}] \; \forall X \in \cC^{\mathbf{E}_\mathrm{eff} }_{n, A|B}\,.
\end{align}
In a slight abuse of notation, we will often omit the quotient map $\pi(\cdot)$ and write  $\rho_{AB} \leq_{\cC^{\mathcal{S}_\mathrm{eff}}_{n, A|B}}\sigma_{AB}$ to mean $\pi(\rho_{AB}) \leq_{\cC^{\mathcal{S}_\mathrm{eff}}_{n, A|B}}\pi(\sigma_{AB})$.

\section{Computational Conditional Min-Entropy}\label{sec:comp-min-entropy}

In this section, we first introduce the \emph{computational max-divergence}, defined with respect to the cone of efficiently implementable Choi operators, and study its basic properties as well as its relation to other operationally meaningful quantities, such as a corresponding computational hypothesis-testing divergence. We then show, in \Cref{sec:op.mean.min.ent}, that the computational conditional min-entropy admits the same operational interpretation as in the informational setting: it characterizes the largest fidelity with a maximally entangled state achievable by a recovery channel acting on $B$. Next, in \Cref{sec:cq-case}, we consider the special case of bipartite classical-quantum states and prove that the corresponding operational quantity is the maximal guessing probability under efficiently implementable measurements. Finally, in \Cref{sec:non.cond.case}, we turn to the non-conditional case and derive a single-letter expression mirroring the informational one.

In \cite{YHK25}, the computational max-divergence is induced by the cone generated by efficiently implementable binary measurements. Here, we consider the analogous max-divergence induced by efficiently implementable completely positive unital (CPU) maps, represented as operators through the Choi-Jamiołkowski isomorphism. This viewpoint strictly generalizes the measurement-based approach of \cite{YHK25}: by \Cref{lem:CPU.Effect.relation}, every complexity-constrained effect operator from \cite{YHK25} arises from an efficiently implementable CPU map. Moreover, in \Cref{subsec:cq.cone.YHK} we show that the max-divergence of \cite{YHK25} is recovered as a special case of the construction introduced below.

\begin{definition}[Computational Max-Divergence]\label{def:comp.max.div}
Given a family of proper cones $\cC^{\mathbf{E}_\mathrm{eff}}_{n,A|B}$ and two positive semidefinite operators $\rho_{AB}, \sigma_{AB} \in \Pos(\HH^A_n \otimes \HH^B_n)$, we denote its corresponding computational max-divergence by:
	\begin{align}\label{equ:def:ConeDmax}
		\CompMaxDiv^{A|B}(\rho_{AB} \|\sigma_{AB}) \coloneqq \log \inf\{\lambda \in \RR \,|\, \rho_{AB} \leq_{ \cC^{\mathbf{S}_\mathrm{eff}}_{n, A|B}} \lambda \sigma_{AB}\}  =\log \sup_{F \in   \cC^{\mathbf{E}_\mathrm{eff}}_{n, A|B}} \frac{\Tr[ \rho_{AB} F_{AB} ]}{\Tr[ \sigma_{AB} F_{AB} ]} \; .
	\end{align}
If there does not exists any $\lambda \in \RR$ such that $\rho_{AB} \leq_{ \cC^{\mathbf{S}_\mathrm{eff}}_{n, A|B}} \lambda \sigma_{AB}$, then $\CompMaxDiv^{A|B}(\rho_{AB} \|\sigma_{AB}) = \infty$.
\end{definition}
We note that the second equality in Eq.~\eqref{equ:def:ConeDmax} directly follows from Eq.~\eqref{eq:cone.partial.order}. Compared to \cite{YHK25}, the order here is induced by the cone generated by efficiently implementable CPU maps (via their Choi operators), rather than directly by efficiently implementable effect operators on $\HH^A_n\otimes \HH^B_n$.

Similarly to \cite{YHK25}, let us assume that the polynomial defining the family of sets of complexity-constrained Choi operators $\{\mathbf{J}^{p(n)}_{n,A|B}\}_{n \in \NN}$ is strictly super additive, i.e., $p(n+m) \geq p(n)+ p(m)+1$. Then, the following lemma holds.

\begin{lem}[Super-additivity of the $\CompMaxDiv^{A|B}$]\label{lem:super.additive.Dmax}
    The computational max-divergence $\CompMaxDiv^{A|B}(\rho_{AB} \|\sigma_{AB})$ is super-additive, i.e., for any $\rho_{AB} = \rho_{AB}^1 \otimes \rho_{AB}^2 $ and $\sigma_{AB} = \sigma_{AB}^1 \otimes \sigma_{AB}^2 $ such that $\rho^{i}_{AB},\sigma^{i}_{AB} \in \Pos(\HH^A_n \otimes \HH^B_n)$,
    \begin{align*}
        \CompMaxDiv^{A|B, \,p(n_1+n_2)}(\rho_{AB} \|\sigma_{AB}) \geq \CompMaxDiv^{A|B, \,p(n_1)}(\rho^1_{AB} \|\sigma^1_{AB})+ \CompMaxDiv^{A|B, \,p(n_2)}(\rho^2_{AB} \|\sigma^2_{AB})\,.
    \end{align*}
\end{lem}
\begin{proof}
The proof follows from the same arguments as~\cite[Lemma 4.3]{YHK25}.
\end{proof}

Moreover, the structure of the cone of Choi operators allows us to prove the following lemma

\begin{lem}[$\sigma_B$ independence of the $\CompMaxDiv^{A|B}(\rho\|\id \otimes \sigma_B)$]\label{cor:indeo.max.div}
Given a quantum state $\rho_{AB}\in \HH^A_n \otimes \HH^B_n$, the computational max divergence satisfies
\begin{align*}
    \CompMaxDiv^{A|B}(\rho_{AB}\| \id_A \otimes \sigma_B) =  \CompMaxDiv^{A|B}(\rho_{AB}\| \id_A \otimes \tau_B)
\end{align*}
for any quantum states $\sigma_B, \tau_B \in \HH^B_n$ that satisfy $\rho_{AB}\ll \id_A\otimes\sigma_B,\ \rho_{AB}\ll \id_A\otimes\tau_B$.
\end{lem}
\begin{proof}
Let us recall that the computational max-divergence can be expressed as
\begin{align*}
  \CompMaxDiv^{A|B}(\rho_{AB}\|\id_A\otimes\sigma_B)
=\log\sup_{F_{AB}\in\cC^{\mathbf{E}_{\mathrm{eff}}}_{n,A|B}}
\frac{\Tr[\rho_{AB}F_{AB}]}{\Tr[(\id_A\otimes\sigma_B)F_{AB}]}\,.  
\end{align*}
Every non-zero $F_{AB}\in\cC^{\mathbf{E}_{\mathrm{eff}}}_{n,A|B}$ satisfies $\tr_A[F_{AB}]=\lambda(F_{AB})\,\id_B$ for some $\lambda(F_{AB})> 0$. Hence, for $\Tr[\sigma_B]=1$,
    \begin{align*}
        \Tr[(\id_A\otimes\sigma_B)F_{AB}] =\Tr\!\left[\sigma_B\,\tr_A[F_{AB}]\right] =\lambda(F_{AB})\,\Tr[\sigma_B] =\lambda(F_{AB})\,,
    \end{align*}
which is independent of $\sigma_B$. Therefore, each ratio in the supremum is independent of $\sigma_B$. Replacing $\sigma_B$ by $\tau_B$ yields the claim.
\end{proof}

Unlike the computational hypothesis-testing relative entropies introduced in \cite{MKN+25,YHK25,MRRLJE25}, the divergence in \Cref{def:cHT.div} is restricted to efficiently generated Choi effect operators. Such effects can be interpreted as the action of a CPTP map on subsystem $B$, followed by the projective measurement onto the maximally entangled state $|\Omega\rangle$. In particular, they do not span the set of all efficiently implementable measurements on $AB$.

\begin{lem}[Equivalence of hypothesis-testing and max-divergence]\label{lem:DHr.Dmax.choi}
Let $\rho_{AB} \in \Pos(\HH^A_n \otimes \HH^B_n)$ be a state, and let $\rho_B=\tr_A[\rho_{AB}]$. Then
\begin{align*}
  \overset{\smash{\text{\tiny\hspace{0.3em}\raisebox{-0.3ex}{c}}}}{D}_{H}^{p(n),\eta}(\rho_{AB} \Vert \id_A \otimes \rho_B) =\CompMaxDiv^{A|B,\,p(n)}(\rho_{AB} \Vert \id_A \otimes \rho_B)
\end{align*}
for every feasible $\eta$, i.e.,~for every $0<\eta\le \sup_{E\in\Eeff_{n,A|B}}\Tr[E\rho_{AB}]$ .
\end{lem}

\begin{proof}
	Let $E_{AB}\in \Eeff_{n,A|B}$. By definition of the normalized Choi slice,
	\begin{equation}
		\ptr{A}{E_{AB}}=\frac{1}{d_A}\id_B \, .
	\end{equation}
	Hence
	\begin{equation}
		\Tr[E_{AB}(\id_A\otimes \rho_B)]
		=
		\Tr\!\left[\rho_B\,\ptr{A}{E_{AB}}\right]
		=
		\frac{1}{d_A}\Tr[\rho_B]
		=
		\frac{1}{d_A}\, .
	\end{equation}
	Therefore, for every feasible $E_{AB}$,
	\begin{equation}
		\frac{\Tr[E_{AB}(\id_A\otimes \rho_B)]}{\Tr[E_{AB}\rho_{AB}]}
		=
		\frac{1}{d_A\,\Tr[E_{AB}\rho_{AB}]}\, .
	\end{equation}
	It follows that
	\begin{align}
		  \overset{\smash{\text{\tiny\hspace{0.3em}\raisebox{-0.3ex}{c}}}}{D}_{H}^{p(n),\eta}(\rho_{AB}\Vert \id_A\otimes \rho_B)
		&=
		-\log
		\inf_{\substack{E_{AB}\in \Eeff_{n,A|B}\\ \Tr[E_{AB}\rho_{AB}]\ge \eta}}
		\frac{1}{d_A\,\Tr[E_{AB}\rho_{AB}]} \\
		&=
		\log\!\left(
		d_A\,
		\sup_{\substack{E_{AB}\in \Eeff_{n,A|B}\\ \Tr[E_{AB}\rho_{AB}]\ge \eta}}
		\Tr[E_{AB}\rho_{AB}]
		\right) \, .
	\end{align}
	Since $\eta$ is feasible, the constraint $\Tr[E_{AB}\rho_{AB}]\ge \eta$ does not affect the
	maximum, and thus
	\begin{equation}
		  \overset{\smash{\text{\tiny\hspace{0.3em}\raisebox{-0.3ex}{c}}}}{D}_{H}^{p(n),\eta}(\rho_{AB}\Vert \id_A\otimes \rho_B)
		=
		\log\!\left(
		d_A\,
		\max_{E_{AB}\in \Eeff_{n,A|B}} \Tr[E_{AB}\rho_{AB}]
		\right) \, .
	\end{equation}
	By the Choi-cone dual formulation of the computational max-divergence, equivalently
	by \Cref{cor:indeo.max.div}, this is exactly
	\begin{equation}
		\CompMaxDiv^{A|B,\,p(n)}(\rho_{AB}\Vert \id_A\otimes \rho_B)\, ,
	\end{equation}
	which proves the claim.
\end{proof}

Similar to the informational setting, the (conditional) min-entropy is the entropy measure canonically associated with the max-divergence order: it quantifies the smallest scalar by which a reference operator of the form $\id_A\otimes\sigma_B$ must be scaled to dominate $\rho_{AB}$ in the relevant partial order. In our complexity-constrained framework, the same construction applies with the cone-induced order $\le_{\cC^{\mathbf{S}_\mathrm{eff}}_{n,A|B}}$. Consequently, and in direct parallel with the informational case \cite{KRS09}, we define the computational conditional min-entropy from the computational max-divergence.

\begin{definition}\label{def:comp.min.entrp}
	Given a family of proper cones $\{ \cC^{\mathbf{E}_\mathrm{eff}}_{n, A|B}\}_{n\in\NN}$, and any quantum state $\rho_{AB} \in \Pos(\HH^A_n\otimes\HH^B_n )$, we denote the corresponding \term{computational min-entropy}  of $A$ conditioned on $B$  by:
	\begin{align}
		\CompMin (A|B)_{\rho}\coloneq - \inf_{\substack{\sigma_B \geq 0 \\ \Tr[\sigma_B]=1}} \CompMaxDiv^{A|B}(\rho_{AB} \| \id_A \otimes \sigma_B) \,.
	\end{align}
\end{definition}

The expression of the computational min-entropy can be simplified as follows.

\begin{cor}\label{cor:comp.hmin.wo.inf}
Given a quantum state $\rho_{AB}\in \HH^A_n \otimes \HH^B_n$, its computational conditional min-entropy satisfies
\begin{align*}
   \CompMin(A|B)_{\rho}=- \CompMaxDiv^{A|B}(\rho_{AB}\| \id_A \otimes \rho_B)
\end{align*}
for $\rho_B=\tr_A[\rho_{AB}]$.
\end{cor}
\begin{proof}
    It directly follows from the fact that the computational max divergence is independent of the $\sigma_B$ state as proven in \Cref{cor:indeo.max.div}.
\end{proof}

In the information-theoretic setting, one distinguishes between the two conditional min-entropies
	\begin{align*}
	H_{\mathrm{min}}(A|B)_{\rho}
&\coloneq  H_{\infty}^{\uparrow}(A|B)_{\rho}=- \inf_{\substack{\sigma_B \geq 0 \\ \Tr[\sigma_B]=1}}
		D_{\mathrm{max}}(\rho_{AB}\|\id_A \otimes \sigma_B)\,,\\		H_{\infty}^{\downarrow}(A|B)_{\rho}
		&\coloneq  - D_{\mathrm{max}}(\rho_{AB}\|\id_A \otimes \rho_B)\,,
	\end{align*}
	where $\rho_B=\tr_A[\rho_{AB}]$. 

    In general, these two quantities do not coincide, i.e., there is a non-zero gap between them in the information-theoretic setting~\cite{T16}. 
	In our complexity-constrained framework, the computational conditional min-entropy defined above is the direct analogue of the up-arrow quantity, i.e.,
	\begin{align*}
		\CompMin(A|B)_{\rho}
		=
		- \inf_{\substack{\sigma_B \geq 0 \\ \Tr[\sigma_B]=1}}
		\CompMaxDiv(\rho_{AB}\|\id_A \otimes \sigma_B)\,.
	\end{align*}
	One may also attempt to introduce the corresponding down-arrow quantity
	\begin{align*}
		\CompDownArrowMin(A|B)_{\rho}
		\coloneq 
		- \CompMaxDiv^{A|B}(\rho_{AB}\|\id_A \otimes \rho_B)\,.
	\end{align*}
    However, by \Cref{cor:indeo.max.div}, the quantity $\CompMaxDiv(\rho_{AB}\|\id_A \otimes \sigma_B)$ is independent of the choice of $\sigma_B$,

\begin{remark}[Computational up- and down-arrow equivalence]\label{rem:up.down.comp.min}
    For any state $\rho_{AB}$ it holds that
	\begin{align*}
		\CompMin(A|B)_{\rho}
		=
		\CompDownArrowMin(A|B)_{\rho}
		=
		-\CompMaxDiv^{A|B}(\rho_{AB}\|\id_A \otimes \rho_B)\,.
	\end{align*}
    \end{remark}
Hence, in our computational framework, the up-arrow and down-arrow definitions coincide. This is a special consequence of the Choi cone structure, where we exclusively optimize over efficient effect operators $F_{AB}$ whose partial trace $\tr_A[F_{AB}]$ are proportional to the identity.
We highlight that a consequence of \Cref{thm:op.mean.fullyquantum} is that we recover the computational min-entropy of~\cite{AHRA25}. That computational entropy is equivalent to the information-theoretic min-entropy $H_{\min}(A|B)_{\rho}$
when the computational restrictions are removed~\cite[Lemma~32]{AHRA25}. As such, our conic structure naturally encodes the optimization over $\sigma_B$ and does not provide the tools to define a computational notion of $H_{\infty}^{\downarrow}(A|B)_{\rho}$.

\begin{lem}[Sub-additivity of the $\CompMin (A|B)_{\rho}$]\label{lem:sub-add.Hmin}
The computational min-entropy is sub-additive, i.e., for any $\rho = \rho_1 \otimes \rho_2$  such that $\rho_1, \rho_2 \in \Pos(\HH_n)$ it holds that,
\begin{align*}
    \CompMin (A_{1}A_{2}|B_{1}B_{2})_{\rho}\leq \CompMin (A_{1}|B_{1})_{\rho_1}+ \CompMin (A_{2}|B_{2})_{\rho_2}\,.
\end{align*}
\end{lem}
\begin{proof}
    It follows from \Cref{cor:comp.hmin.wo.inf} together with the fact of the computational max-divergence being super-additive (\Cref{lem:super.additive.Dmax}).
\end{proof}

Up to this point, we have introduced the computational max-divergence and the induced conditional min-entropy, and established the structural properties needed for its analysis. We now show that this entropy retains the standard operational interpretation, with the unrestricted recovery map replaced by an efficiently implementable channel.

\subsection{The Operational Meaning of the Computational Min-Entropy}\label{sec:op.mean.min.ent}
Let us now prove that, in the complexity-constrained setting, $\CompMin(A|B)_\rho$ admits an operational interpretation analogous to the standard (information-theoretic) conditional min-entropy \cite{KRS09}. In particular, we recover the corresponding notion considered in \cite{AHRA25}. Let us first prove the following lemma.

\begin{lem}\label{lem:dual.program.Hmin} Let $\HH^A_n$ and $\HH^B_n$ be finite-dimensional Hilbert spaces. Let $\{\mathbf{J}^{p(n)}_{n,A|B}\}_{n \in \NN}$ be the family of Choi operators of CPU maps generated by at most $p(n)$ given a gate-set $\mathcal{G}$ and  $ \cC^{\mathbf{S}_\mathrm{eff}}_{n, A|B}$ its corresponding dual cone. Then, for any $\rho_{AB} \in \Pos(\HH^A_n \otimes \HH^B_n)$ and $\sigma_B \in \Pos(\HH^B_n)$, it holds that
	\begin{align}\label{eq:dual.program.Hmin}
		\min_{\substack{\sigma_B \geq 0 \\\id_A \otimes \sigma_B \geq_{ \cC^{\mathbf{S}_\mathrm{eff}}_{n, A|B}} \rho_{AB}}} \Tr[\sigma_B] =  \max_{\substack{F_{AB} \geq_{ \cC^{\mathbf{E}_\mathrm{eff}}_{n, A|B}} \!\!\!\!\!0\\ \tr_A[F_{AB}] = \id_B}} \Tr[\rho_{AB} F_{AB}] \,.
	\end{align}
\end{lem}

\begin{proof}
The argument is inspired by the proof of \cite[Lemma 4]{KRS09}, but its implementation in our setting requires several nontrivial modifications in order to accommodate the complexity constraints.

Let us first exploit the special structure of the cone of efficiently generated Choi operators to reduce the problem to a dual program.

For each efficiently generated Choi operator $F_{AB}$ we have that $\tr_A[F_{AB}] =  \id_B$. Therefore, for the cone of efficient Choi operators $\cC^{\mathbf{E}_\mathrm{eff}}_{n, A|B}$ define in \Cref{def:fully.quantum.cone}, any element $F_{AB}\in \cC^{\mathbf{E}_\mathrm{eff}}_{n, A|B}$ can be written as
	\begin{align*}
		F_{AB} = \sum_j \alpha_j J^{(j)}_{AB},\qquad \alpha_j \ge 0\, ,
	\end{align*}
with $J^{(j)}_{AB}$ a (channel-normalized) Choi operator, i.e., $\tr_A [J^{(j)}_{AB}] = \id_B$. In consequence, it also satisfies
	\begin{align*}
		\tr_A[F_{AB}] = \left(\sum_j \alpha_j\right)\id_B =: t\,\id_B
	\end{align*}
	for some scalar $t\ge 0$. If $t=0$, then $F_{AB}=0$ and the objective is zero.
	For $0<t\le 1$, define the rescaled operator
	\begin{align*}
		\widehat{F}_{AB} \coloneq  \frac{1}{t} F_{AB}\, .
	\end{align*}
	Since $\cC^{\mathbf{E}_\mathrm{eff}}_{n, A|B}$ is a cone, then $\widehat{F}_{AB}\in  \cC^{\mathbf{E}_\mathrm{eff}}_{n, A|B}$, and $\tr_A[\widehat{F}_{AB}] = \id_B$. Moreover, for any state $\rho_{AB}\ge 0$,
	\begin{align*}
		\Tr[\rho_{AB} \widehat{F}_{AB}]  = \frac{1}{t}\,\Tr[\rho_{AB} F_{AB}] \ge \Tr[\rho_{AB} F_{AB}]\,.
	\end{align*}

	Thus, since the supremum in the r.h.s. of Eq.\eqref{eq:dual.program.Hmin} is attained on elements with $\tr_A[F_{AB}] = \id_B$, we can extend the maximization to all operators $F_{AB}$ such that $\tr_A[F_{AB}] \leq \id_B$. Therefore, we will prove the following claim:
	\begin{align}\label{eq:relaxed.dual.program.Hmin}
		\min_{\substack{\sigma_B \geq 0 \\\id_A \otimes \sigma_B \geq_{ \cC^{\mathbf{S}_\mathrm{eff}}_{n, A|B}}  \! \!\!\! \rho_{AB}}} \Tr[\sigma_B] = \max_{\substack{F_{AB} \geq_{ \cC^{\mathbf{E}_\mathrm{eff}}_{n, A|B}} \! \!\!\! 0\\ \tr_A[F_{AB}] \leq \id_B}} \Tr[\rho_{AB} F_{AB}]\,.
	\end{align}

We formulate Eq.~\eqref{eq:relaxed.dual.program.Hmin} as a conic program using the tools introduced in \cref{sec:cone.theory} on the spaces
\begin{align*}
  \VV_1\coloneq \Herm(\HH_B^n)\,, \qquad \VV_2\coloneq \VV^*_{n,A|B}\,,  
\end{align*}
with cones
\begin{align*}
    C_1\coloneq \Pos(\HH_B^n)\,, \qquad C_2\coloneq C^{S,\mathrm{eff}}_{n,A|B}\,,
\end{align*}
and identifying the elements of Eq.\eqref{eq:conic.primal} and Eq.\eqref{eq:conic.dual} as $c\coloneq \id_B$ and $b\coloneq \pi(\rho_{AB})$, and the linear map 
\begin{align*}
    \Phi:\VV_1\to \VV_2,\qquad \Phi(\sigma_B)\coloneq \pi(\id_A\otimes \sigma_B)\,.
\end{align*}

Then the primal problem is exactly
\begin{align}
\gamma_{\mathrm{primal}}=\inf_{\sigma_B\in \Herm(\HH_B^n)}
\left\{\Tr[\sigma_B]:\ \sigma_B\in \Pos(\HH_B^n),\Phi(\sigma_B)-\pi(\rho_{AB})\in C^{S,\mathrm{eff}}_{n,A|B}
\right\}.
\end{align}
Since the pairing between $\VV^*_{n,A|B}$ and $\VV_{n,A|B}$ is given by
\begin{align*}
 \langle \pi(Y),F_{AB}\rangle = \Tr[YF_{AB}],\qquad Y\in \Herm(\HH_A^n\otimes \HH_B^n),\ \ F_{AB}\in \VV_{n,A|B}\,,   
\end{align*}
the adjoint map $\Phi^*:\VV_{n,A|B}\to \Herm(\HH_B^n)$ is
\begin{align*}
    \Phi^*(F_{AB})=\tr_A[F_{AB}]\,,
\end{align*}
because
\begin{align*}
    \langle \Phi(\sigma_B),F_{AB}\rangle=\Tr[(\id_A\otimes \sigma_B)F_{AB}]=\Tr[\sigma_B\,\tr_A[F_{AB}]]\,.
\end{align*}
Hence, the dual problem is given by
\begin{align*}
\gamma_{\mathrm{dual}} = \sup_{F_{AB}\in \VV_{n,A|B}} \left\{ \Tr[\rho_{AB}F_{AB}] : \ F_{AB}\in C^{E,\mathrm{eff}}_{n,A|B},\ \id_B-\tr_A[F_{AB}] \in \Pos(\HH_B^n) \right\}\,,
\end{align*}
which can be simplified as,
\begin{align}
\gamma_{\mathrm{dual}}=\sup_{F_{AB}\in C^{E,\mathrm{eff}}_{n,A|B}}\left\{ \Tr[\rho_{AB}F_{AB}] : \tr_A[F_{AB}] \le \id_B \right\}\,.
\end{align}
Therefore, by weak duality for feasible primal-dual conic programs,
\begin{align}
\min_{\substack{\sigma_B\ge 0\\ \pi(\id_A\otimes \sigma_B)-\pi(\rho_{AB})\in C^{S,\mathrm{eff}}_{n,A|B}}} \Tr[\sigma_B] \ \ge\ \max_{\substack{F_{AB}\in C^{E,\mathrm{eff}}_{n,A|B}\\ \tr_A[F_{AB}]\le \id_B}} \Tr[\rho_{AB}F_{AB}] \,.
\end{align}

To prove equality, it remains to show that there is no duality gap. For this, we verify Slater's condition for the primal problem. Since $ C^{E,\mathrm{eff}}_{n,A|B}$ is a proper cone in $\VV_{n,A|B}$, its dual cone $C^{S,\mathrm{eff}}_{n,A|B}$ is a proper cone in $\VV^*_{n,A|B}$. Thus it suffices to find $\sigma_B\in \mathsf{ri}(\Pos(\HH_B^n))$ such that
\begin{align*}
\Phi(\sigma_B)-\pi(\rho_{AB})=\pi(\id_A\otimes \sigma_B-\rho_{AB})\in \mathsf{ri}(C^{S,\mathrm{eff}}_{n,A|B})\,.
\end{align*}

By taking $\sigma_B \coloneq  2\lambda_{\max}(\rho_{AB})\,\id_B$, then $\sigma_B\in \mathsf{ri}(\Pos(\HH_B^n))$. Moreover, for every
$F_{AB}\in C^{E,\mathrm{eff}}_{n,A|B}\setminus\{0\}$,
\begin{align*}
\left\langle \pi(\id_A\otimes \sigma_B-\rho_{AB}),F_{AB}\right\rangle
&=\Tr[(\id_A\otimes \sigma_B-\rho_{AB})F_{AB}] \\
&=2\lambda_{\max}(\rho_{AB})\Tr[F_{AB}]-\Tr[\rho_{AB}F_{AB}] \\
&\ge\lambda_{\max}(\rho_{AB})\Tr[F_{AB}]>0\,,
\end{align*}
where the second to last inequality is Hölder inequality and the last strict inequality follows since $\Tr[E]=\|E\|_1>0$ as by assumption $F_{AB}\neq0$ and $\rho_{AB}$ is a finite dimensional normalized quantum state, hence $\lambda_{\max}(\rho_{AB})>0$. Hence, by the standard characterization of the interior of the dual cone,
\begin{align*}
    \pi(\id_A\otimes \sigma_B-\rho_{AB})\in \mathsf{ri}(C^{S,\mathrm{eff}}_{n,A|B})\,.
\end{align*}
Therefore, Slater's condition holds, and the primal and dual optimal values coincide.

\end{proof}

Let us now prove the main theorem of this section.

\begin{theorem}\label{thm:op.mean.fullyquantum}
	Given a quantum state $\rho_{AB}\in \HH^A_n \otimes \HH^B_n$, its computational min-entropy can be expressed as
	\begin{align}\label{eq:q.H.min}
		\CompMin(A|B)_{\rho} = -\log \CompQCorr(A|B)_{\rho}\,,
	\end{align}
	where
	\begin{align}
		\CompQCorr(A|B)_\rho
		\coloneq  d_A\cdot \max_{T} F\left( (\Id_A \otimes T)(\rho_{AB}), \ketbra{\Omega_{AA'}}{\Omega_{AA'}} \right) 
	\end{align}
	is the maximal fidelity with a maximally entangled state achievable by channel  $T: \BB(\HH^B_n) \rightarrow\BB(\HH^{A'}_n)$, with $ \HH^{A'}_n\cong\HH^{A}_n$ and $\ket{\Omega_{AA'}} \coloneq  \frac{1}{\sqrt{d_A}} \sum_{x \in [d_A]}\ket{x}_A \otimes \ket{x}_{A'}$.
\end{theorem}
\begin{proof}
	By definition of the computational max-divergence,
	\begin{align*}
		\CompMin(A|B)_{\rho}
		 & = - \inf_{\substack{\sigma_B \ge 0                            \\ \Tr[\sigma_B] = 1}}
		\CompMaxDiv^{A|B}\left(\rho_{AB}\,\big\|\,\id_A \otimes \sigma_B\right) \\
		 & = - \log \inf_{\substack{\sigma_B \ge 0                       \\ \Tr[\sigma_B] = 1}}
		\inf\left\{\lambda \ge 0 :
		\rho_{AB} \le_{\cC^{\mathbf{S}_\mathrm{eff}}_{n, A|B}}
		\lambda(\id_A\otimes\sigma_B)\right\}.
	\end{align*}
	Introduce the unnormalised variable $\tilde{\sigma}_B \coloneq  \lambda \sigma_B \ge 0$.
	Then the order constraint can be written as
	\begin{align*}
		\rho_{AB} \le_{\cC^{\mathbf{S}_\mathrm{eff}}_{n, A|B}}
		\id_A\otimes \tilde{\sigma}_B
		\iff
		\id_A\otimes \tilde{\sigma}_B - \rho_{AB} \in \cC^{\mathbf{S}_\mathrm{eff}}_{n, A|B},
	\end{align*}

	and $\Tr[\tilde{\sigma}_B] = \lambda$. Hence
	\begin{align*}
		2^{-\CompMin(A|B)_{\rho}}
		 & = \inf_{\tilde{\sigma}_B \ge 0}
		\left\{\Tr[\tilde{\sigma}_B] :
		\id_A\otimes \tilde{\sigma}_B
		\ge_{\cC^{\mathbf{S}_\mathrm{eff}}_{n, A|B}} \rho_{AB}\right\}\;.
	\end{align*}
	This is exactly the primal conic program of ~\Cref{lem:dual.program.Hmin}, then it follows that
	\begin{align}
		2^{-\CompMin(A|B)_{\rho}}
		= \max_{\substack{F_{AB} \in \cC^{\mathbf{E}_\mathrm{eff}}_{n, A|B} \\
				\tr_A[F_{AB}] = \id_B}}
		\Tr[\rho_{AB} F_{AB}] \;.
		\label{eq:cq-dual}	
    \end{align}
	Let us now prove the following equivalence:
	\begin{align*}
		\log \max_{\substack{F_{AB} \in \cC^{\mathbf{E}_\mathrm{eff}}_{n, A|B} \\
				\tr_A[F_{AB}] = \id_B}}
		\Tr[\rho_{AB} F_{AB}]=\log\left( d_A \cdot \max_{T} F\left( (\Id_A \otimes T)(\rho_{AB}), \ketbra{\Omega_{AA'}}{\Omega_{AA'}} \right) \right)
	\end{align*}
	By construction, $F_{AB} \in \BB(\HH^A_n \otimes \HH^B_n)$ is a nonnegative operator such that $\tr_A[F_{AB}] = \id_B$. Then, due to the Choi-Jamiołkowski isomorphism, it corresponds to the Choi operator of a CPU map $T^*$. Therefore,
	\begin{align*}
		\Tr[\rho_{AB}F_{AB}]
		 & =d_A \Tr \left[ \rho_{AB}(\Id_A \otimes T^{*})\ketbra{\Omega_{AA'}}{\Omega_{AA'}}\right]    \\
		 & = d_A \Tr \left[ (\Id_A \otimes T)(\rho_{AB})(\ketbra{\Omega_{AA'}}{\Omega_{AA'}})\right]   \\
		 & = d_A \bra{\Omega_{AA'}} (\Id_A \otimes T)(\rho_{AB})\ket{\Omega_{AA'}}                     \\
		 & =  d_A F \left( (\Id_A \otimes T)(\rho_{AB}), \ketbra{\Omega_{AA'}}{\Omega_{AA'}}\right)\,,
	\end{align*}
	where the first equality follows from duality, with no complexity overhead as proven in \Cref{lem:complexity.dual.equivalence}, and the last one follows by the definition of the fidelity $F(\rho,\sigma) = \| \sqrt{\rho} \sqrt{\sigma}\|_1^2 = \left( \Tr\!\left[\sqrt{\sqrt{\rho} \sigma \sqrt{\rho}}\right]\right)^2$ taking into account that $\ket{\Omega_{AA'}}$ is a pure state.

	Then, since the maximization is done for a linear function w.r.t. a convex set, it suffices to take the extremal point of the convex set, i.e., $\mathbf{J}^{p(n)}_{n,A|B}$. Lastly, due to the Choi-Jamiołkowski isomorphism to \Cref{lem:complexity.dual.equivalence}, this corresponds to the maximization w.r.t. the set of efficiently generated CPTP maps.
\end{proof}

Therefore, the computational min-entropy admits the same operational interpretation as in the information-theoretic setting: it is the (negative logarithm of the) optimal fidelity with respect to a maximally entangled state achieved by acting on the conditioning system when the quantum channel is restricted to be efficiently implementable. This mirrors the characterization of \cite{KRS09} and, under our complexity constraints, recovers the computational conditional min-entropy considered in \cite{AHRA25}. It directly follows that dropping the computational restrictions, we recover the information theoretic min-entropy; see also~\cite[Lemma~32]{AHRA25}.

\subsection{The Classical-Quantum Case}\label{sec:cq-case}

Let us now study the case in which one of the registers is classical. This case can be seen as a reduction of the previous one, in which both subsystems $X\equiv A$ and $X^\prime\equiv A^\prime$ are classical, however, $B$ may not be. Then, as in the information-theoretic case, we recover the interpretation as proportional to the optimal guessing probability under computational constraints on the allowed measurements, analogous to~\cite[Lemma~32]{AHRA25}.
\begin{thm}\label{th:cq.oper.int}
	Let $\rho_{XB} = \sum_{i=1}^{|X|} p_i \ketbra{i}{i}_X \otimes \rho_B^i$ be a cq state.
	Then
	\begin{align}
		\CompMin(X|B)_{\rho} = -\log \CompPGuess(X|B)_{\rho}
	\end{align}
	where
	\begin{align}
		\CompPGuess(X|B)_{\rho}
		\coloneq  \max_{\{E_B^i\}_{i\in X}\in \mathbf{M}^{p(n)}_{n,|X|}}
		\sum_{i\in X} p_i\,\Tr\!\left[E_B^i \rho_B^i\right]
	\end{align}
	is the optimal guessing probability when restricted to efficient $|X|$-outcome POVMs on $B$.
\end{thm}

\begin{proof}
    Recall \Cref{lem:dual.program.Hmin}. Then, for a given cq state $\rho_{XB} = \sum_i p_i \ketbra{i}{i}_X\otimes\rho_B^i$, it follows that
	\begin{align}
		2^{-\CompMin(X|B)_{\rho}} & =\max_{\substack{E_{XB} \in \cC^{\mathbf{E}_\mathrm{eff}}_{n, X|B} \\
				\tr_X[E_{XB}] = \id_B}}
		\Tr[\rho_{XB} E_{XB}]   \\
		 & =\max_{J_{XB}\in \mathbf{J}_{n,X|B}}\Tr[\rho_{XB}J_{XB}]           \\ &=\max_{\underset{C(T^*)\leq p(n)}{T^*:X^\prime\to B}} \Tr\left[\left(\sum_{i=1}^{|X|}p_i\ketbra{i}{i}_X\otimes\rho_B^i\right)d_A(\Id_X\otimes T^*(\ketbra{\Omega}{\Omega}_{XX^\prime}))\right] \\ &=\max_{\underset{C(T^*)\leq p(n)}{T^*:X^\prime\to B}} \Tr\left[\sum_{i=1}^{|X|}p_i\rho^i_BT^*(\ketbra{i}{i}_{X^\prime})\right] \\ &= \max_{\{E_B^i\}_{i\in X}\in \mathbf{M}^{p(n)}_{n,|X|}}
		\sum_{i\in X} p_i\,\Tr\!\left[E_B^i \rho_B^i\right]\,.
	\end{align} 
  Here, the second equality is a consequence of the fact that, due to convexity, the optimal choice of $E_{XB}$ is an extremal point. By construction, this choice of $E_{XB}$ must therefore also be contained in the strictly smaller set $\mathbf{J}_{n,X|B} \subset \cC^{\mathbf{E}_\mathrm{eff}}_{n, X|B}$ (see Definition~\ref{def:fully.quantum.cone}). The third equality follows from the definition of $\mathbf{J}_{n,X|B}$ and the fourth from
	\begin{align}
		\Tr & \left[\left(\sum_{i=1}^{|X|}p_i\ketbra{i}{i}_X\otimes\rho_B^i\right)d_A(\Id_X\otimes T^*(\ketbra{\Omega}{\Omega}_{XX^\prime}))\right] \\ &= \sum_{i=1}^{|X|}p_i\Tr[(|i\rangle\langle i|_X\otimes\rho_B^i)\sum_{y,z}  \ketbra{y}{z}_X\otimes T^*(\ketbra{y}{z}_{X^\prime})] \\ &= \sum_{i=1}^{|X|}p_i\Tr[\rho_B^iT^*(\ketbra{i}{i}_{X^\prime})].
	\end{align}
	Finally the last equality is denoting the POVM elements $T^*(\ketbra{i}{i}_{X^\prime})=:E_B^i$, which by construction via the efficient map $T^*$ form an efficient $|X|-$outcome POVM $\{E_B^i\}_{i=1}^{|X|}\in\mathbf{M}^{p(n)}_{n,|X|}$ and an analogous argument as for the second equality.
	Taking the logarithm and a minus sign concludes the proof.
\end{proof}

Note that in the c-q setting we may also effectively replace the optimization over $\mathcal{C}^{\Eeff}_{n,A|B}$ by an optimization over efficient quantum-classical maps, which is also equivalent to pinched versions of the efficient quantum-quantum maps $T^*$, and the above optimization over efficient POVMs.

\subsection{Measurement Channels for CQ--States}\label{subsec:cq.cone.YHK}

For classical-quantum states, we show that restricting to efficient Choi operations preserves the classical-quantum structure. Let $\rho_{XB}$ be a classical-quantum state, i.e.\ $\rho_{XB}=\sum_{i=1}^{|X|} \ketbra{i}{i}_X\otimes \rho_B^i$ is block diagonal in the computational basis of $X$.
Then for every $F_{XB}\in\Herm(\HH_X\otimes \HH_B)$ one has
\begin{align}
    \Tr(\rho_{XB}F_{XB})=\Tr\left(\rho_{XB}\Delta_X(F_{XB})\right),
\end{align}
where $\Delta_X(F)=\sum_{i}(\ketbra{i}{i}_X\otimes \id_B)F(\ketbra{i}{i}_X\otimes \id_B)$ is the pinching map in the computational basis. Hence, in any cone program whose objective is of the form $F\mapsto \Tr(\rho_{XB}F)$, one may w.l.o.g.\ restrict the optimization to $X$-diagonal operators.

In particular, the optimization over the cone of efficient Choi-operators in the computational min-entropy can be restricted by an optimization over the cone of efficient \emph{flagged POVMs}
\begin{align}
	\tilde{\cC}^{\mathbf{E}_\mathrm{eff}}_{n, X|B}
	\coloneq (\Delta_X\otimes\Id_B)\!\left(\cC^{\mathbf{E}_\mathrm{eff}}_{n, X|B}\right)=\cone\left\{ \sum_{i=1}^{|X|}\ketbra{i}{i}_X\otimes E^i_B \ \bigg| \ \{E^i_B\}_{i=1}^{|X|}\in \mathbf{M}^{p(n)}_{n,|X|}\right\},
\end{align}

where $\mathbf M^{p(n)}_{n,|X|}$ denotes the family of $|X|$-outcome POVMs on $B$ implementable with complexity $p(n)$. The second equality is easily verified by explicit computation.

Equivalently, one might see the cone $\tilde{\cC}^{\mathbf{E}_\mathrm{eff}}_{n, X|B}$ as the cone of efficiently generated quantum-classical (q-c) measurement channels. Given a POVM $\{E_i\}_{i=1}^{|X|}\in \mathbf {M}^{p(n)}_{n,|X|}$, its corresponding \emph{quantum-classical measurement channel} $\cM:\BB(\HH^B_n)\to \BB(\HH^{X^\prime}_n)$ is given by
\begin{align}
	\cM(\rho_B)\coloneq \sum_{i=1}^{|X|}\tr(E^i_B\rho_B)\,\ketbra{i}{i}_{X^\prime}.
\end{align}
Its adjoint $\cM^*:\BB(\HH^{X^\prime}_n)\to \BB(\HH^B_n)$ is a CPU map and satisfies
\begin{align}
	\cM^*(\ketbra{i}{i}_{X^\prime})=E^i_B,\qquad \cM^*(\ketbra{i}{j}_{X^\prime})=0 \ \ (i\neq j),
\end{align}
so that its the Choi operator
\begin{align}
	J(\cM^*)=\sum_{i=1}^{|X|}\ketbra{i}{i}_X\otimes E^i_B,
\end{align}
is diagonal in $X$. Therefore, the optimization over $\tilde{\cC}^{\mathbf{E}_\mathrm{eff}}_{n, X|B}$ is equivalently an optimization over efficient q-c measurement channels $\cM$.
    
When $|X|=2$, every operator in the set that generates the cone $\tilde{\cC}^{\mathbf{E}_\mathrm{eff}}_{n, X|B}$ has the form
\begin{align}
	\ketbra{0}{0}_X\otimes E_B \;+\; \ketbra{1}{1}_X\otimes (\id_B-E_B),
\end{align}
where $\{E_B,\id_B-E_B\}\in \mathbf M^{p(n)}_{n,2}$ is an efficient binary POVM.
Thus, in the binary c-q case, optimization over flagged POVMs is equivalently parametrized by the choice of an efficient effect operator $E_B$ (the first POVM element), with the second outcome fixed by complement.
In particular, this is the same binary-measurement model as in \cite[Def.~3.4, Def.~3.6, Def.~3.13]{YHK25}: their polynomially generated set of efficient effect operators and its induced cone of efficient effects describe exactly the admissible binary tests underlying the above flagged form.

\begin{remark}[On efficient effect operators vs. efficient Choi operators and recovering the information-theoretic max-divergence]\label{rem:non.recovery}

In the previous section we considered a notion of efficient measurements, which in the classical-quantum setting lead us to consider the cone generated by 
\begin{align*}
    J^{p(n)}_{n,X|B} =\left\{\sum_i|i\rangle\langle i|_X\otimes E^x_B\,|\, E^{x}_B=T^*(|x\rangle\langle x|) \text{ for some efficient CPU } T^* \right\}.
\end{align*} This says that effect operators of efficient flagged POVMs are efficiently implementable by themselves.
Coming from information theory, it may be natural to wonder whether this set is the same as the set of flagged measurements 
\begin{align*}
 \sum_x \ketbra{x}{x}_X\otimes E_B^x
\end{align*} where each effect $E_B^x$ is part of \emph{some} efficient POVM, not necessarily stemming from the same map $T^*$. This is the setting of the efficient binary measurements of \cite{YHK25} and morally also the set of measurements considered without any computational restrictions \cite{KRS09}. However, when considering superpolynomial many measurement outcomes $|X|$ these two approaches are strictly different, with the former being the only operationally sensible one. This is a subtle but important difference to the unrestricted setting succinctly summarized by: \emph{not all flagged POVMs made of efficient effect operators are efficient}.

For example, let $\{\pi_n\}_{n\in\mathbb N}$ be a family of permutations $\pi_n:\{0,1\}^n\to\{0,1\}^n$ such that $\pi_n^{-1}$ is not implementable by polynomial-size circuits. Such a family must exist, since there are $n!$ many permutations on $\{0,1\}^n$ to choose from. Now consider
\begin{align*}
   F^{(n)}_{XB}:=\sum_{x\in\{0,1\}^n} \ketbra{x}{x}_X\otimes \ketbra{\pi_n(x)}{\pi_n(x)}_B  = \sum_{x\in\{0,1\}^n} \ketbra{x}{x}_X\otimes E^{x,n}_B.
\end{align*}
Information-theoretically, $F^{(n)}_{XB}$ is the Choi operator of the q-c measurement channel $T_{\pi_n}$ associated with the POVM $\{E_B^{x,n}\}_x$.  Its dual is given by  
\begin{align*}
   T_{\pi_n}^*(Z) = \sum_x \langle x|Z|x\rangle\,\ketbra{\pi_n(x)}{\pi_n(x)} .
\end{align*}
For each fixed $x\in\{0,1\}^n$, the projectors $\{E_B^{x,n},\, \id- E_B^{x,n}\}$ are efficiently implementable binary effects. However, an implementation of this POVM with classical output $x$ would compute $\pi_n^{-1}$ on computational-basis inputs. Thus the family $\{E_B^{x,n}\}_x$ is not an efficiently implementable POVM in the sense of this work \Cref{def:eff.POVM}, even though each effect $E_B^{x,n}$ is efficiently implementable as a binary test \Cref{def:eff-binary-effects}.

In the fully quantum setting there is an analogous distinction. The cone $\cC^{\mathbf{E}_\mathrm{eff}}_{n,A|B}$ we consider in this work is generated by Choi operators of efficiently implementable CPU maps $T^*:\BB(\HH^{A'}_n)\to\BB(\HH^B_n)$.  
On the other hand, inspired by \cite[Lemma 4]{KRS09},
one could try to define the set of efficient measurements on $AB$, generated by maps $T^*:\BB(\HH^{A}_n\otimes\HH^{B}_n)\to\BB(\HH_n^A\otimes\HH^B_n)$, say $\{\,\{E^x_{AB}\}_x\,\}$ and then impose the linear constraint $\tr_A[E^x_{AB}]=\id_B$ to get a set of normalized Choi operators. These are in general, however, not Choi operatros of efficient maps $T^*:\BB(\HH^{A'}_n)\to\BB(\HH^B_n)$, as seen also in the above classical quantum counter example. 
\end{remark}

We emphasize the following important consequence.
\begin{remark}
As a notable consequence, there exists a distinction between the computational max-divergence $\CompMaxDiv^{A|B}$ considered in this work, which leads to the operationally meaningful computational min-entropy and the max-divergence $\CompMaxDiv$ defined via introducing computational constraints on global measurements $\{E_{n,AB}\}$ which, without complexity constraints, converges to the information-theoretic max-divergence~\cite{YHK25}. We expect this distinction to generally hold true; any divergence that can lead to the computational min-entropy we consider in this work will not converge to the information-theoretic max-divergence.
\end{remark}

\subsection{Computational Max-Divergence and Measured Max-Divergence Relations}\label{subesc:compmax.mesured.divergenceYHK25equive}

This classical-quantum reduction lets us extend \cite[Thm.~4.12]{YHK25} beyond the binary setting. In particular, we show that the computational max-divergence induced by the cone of efficiently implementable q-c measurement channels (equivalently, efficiently generated flagged POVMs) coincides with the computational measured Rényi divergence of order $\infty$ from \cite[Def.~4.9]{YHK25}, for arbitrary outcome alphabets~$X$. We do this via an embedding 
\begin{align*}
    \rho_B \mapsto \hat{\rho}_{XB}&:=\ketbra{1}{1}_X\otimes\rho_B
\end{align*} storing the state state in the $B$ register while $A$ holds an uncorrelated ancillary system.

\begin{thm}[Computational max-divergence and measured max-divergence equivalence]\label{lem:max.div.eq.Renyi}
	Given a finite gate-set $\mathcal{G}$, let $\tilde{\cC}^{\mathbf{E}_\mathrm{eff}}_{n, X|B}$ be the corresponding cone of efficiently generated flagged POVMs induced by $\mathbf M^{p(n)}_{n,|X|}$ as in \Cref{subsec:cq.cone.YHK}.
	Then for any $\rho_B,\sigma_B\in\Pos(\HH^B_n)$,
	\begin{align*}
		\CompMaxDiv^{X|B}(\hat\rho_{XB}\|\hat{\sigma}_{XB}) \;=\; \CompDiv^{\mathbf M^{p(n)}_{n,|X|}}_{\infty}(\rho_B\|\sigma_B)\,,
	\end{align*}
	where $\CompMaxDiv^{X|B}(\hat\rho_{XB}\|\hat{\sigma}_{XB})\coloneq  \CompMaxDiv^{X|B}\left(|1\rangle\langle1|_X \otimes \rho_{B}\||1\rangle\langle1|_X \otimes\sigma_{B}\right)$ is the computational max divergence of \Cref{def:comp.max.div} with $X=A$ and
	\begin{align*}
		\CompDiv^{\mathbf M^{p(n)}_{n,|X|}}_\infty(\rho_B\|\sigma_B)
		&\coloneqq 
		\sup_{M\in\mathbf M^{p(n)}_{n,|X|}}
		\log\inf\left\{\lambda\in\RR:\ \cM_M(\rho_B)\le \lambda\,\cM_M(\sigma_B)\right\},
	\end{align*}
	is the computational measured max-divergence of \cite[Definition IV.8]{YHK25} where $\cM_M:\BB(\HH^B_n)\to \BB(\HH^{X'}_n)$ denotes the q-c measurement channel corresponding to $M$.
\end{thm}

\begin{proof}
	Let $M=\{E_B^i\}_{i=1}^{|X|}\in \mathbf M^{p(n)}_{n,|X|}$ and recall that its q-c measurement channel is
	\begin{align*}
		\cM_M(\tau)=\sum_{i=1}^{|X|}\Tr(E_B^i\tau)\,\ketbra{i}{i}_{X'}.
	\end{align*}
	Set $p_i\coloneq\Tr(E_B^i\rho_B)$ and $q_i \coloneq\Tr(E_B^i\sigma_B)$.
	Since $\cM_M(\rho_B)$ and $\cM_M(\sigma_B)$ are diagonal in the computational basis of $X'$, the operator inequality
	$\cM_M(\rho_B)\le \lambda\,\cM_M(\sigma_B)$ holds if and only if $p_i\le \lambda q_i$ for all $i$.
	Hence
	\begin{align*}
		\inf\{\lambda:\ \cM_M(\rho_B)\le \lambda\,\cM_M(\sigma_B)\}
		=\max_{i:\,q_i>0}\frac{p_i}{q_i},
	\end{align*}
	with the convention that if $q_i=0$ and $p_i>0$ for some $i$ then the infimum equals $+\infty$.
	Therefore,
	\begin{align}
		\CompDiv^{\mathbf M^{p(n)}_{n,|X|}}_\infty(\rho_B\|\sigma_B)
		=\sup_{M\in\mathbf M^{p(n)}_{n,|X|}}
		\log\max_{i\in[|X|]}\frac{\Tr(E_B^i\rho_B)}{\Tr(E_B^i\sigma_B)}\,.
		\label{eq:comp-meas-dinf-maxi}
	\end{align}

Let $M=(E_i)_{i=1}^{|X|}\in \mathbf M^{p(n)}_{n,|X|}$ be the POVM with the effect operator $E_{i^\star}$ that attains the maximum in Eq.~\eqref{eq:comp-meas-dinf-maxi}. Consider the polynomial-time classical post-processing that maps the outcome \(i^\star\) to \(1\) and all other outcomes to \(2\) in the following way
\begin{align*}
   \{E^i\}_{i=1}^{|X|}\mapsto \{F_j\}_{j=1}^{2}\,:
\quad
F_1=E_{i^\star}\,,\quad
F_2=\id-E_{i^\star}\,
\end{align*}
The resulting 2-outcome POVM is still efficiently implementable, up to polynomial classical post-processing. Hence
\begin{align*}
    \max_{i\in[|X|]} \frac{\Tr(E_i\rho_B)}{\Tr(E_i\sigma_B)} = \frac{\Tr(F_1\rho_B)}{\Tr(F_1\sigma_B)}\,.
\end{align*}
Taking the supremum over $M$ gives
	\begin{align}
		\CompDiv^{\mathbf M^{p(n)}_{n,|X|}}_\infty(\rho_B\|\sigma_B)
		=
		\sup_{M\in\mathbf M^{p(n)}_{n,|X|}} \log \frac{\Tr(E_B^1\rho_B)}{\Tr(E_B^1\sigma_B)}\,,
	\end{align}
Moreover, as proven by \cite[Lemma 2]{RSB24}, the measured max-divergence is jointly quasi-convex. Then, for the convex hull of $\mathbf M^{p(n)}_{n,|X|}$, which we denote as $\overline{\mathbf M}^{p(n)}_{n,|X|}$, we have that
\begin{align*}
		\CompDiv^{\mathbf M^{p(n)}_{n,|X|}}_\infty(\rho_B\|\sigma_B) = \sup_{\{E_n^i\}_{i \in |X|}\in\mathbf M^{p(n)}_{n,|X|}} \log \frac{\Tr(E_B^1\rho_B)}{\Tr(E_B^1\sigma_B)}= \sup_{\{E_n^i\}_{i \in |X|}\in\overline{\mathbf M}^{p(n)}_{n,|X|}} \log \frac{\Tr(E_B^1\rho_B)}{\Tr(E_B^1\sigma_B)}\,,
	\end{align*}
as proven by \cite[Lemma 9]{YHK25}. On the other hand, by the definition of $\CompMaxDiv^{A|B}(\rho_{AB}\|\sigma_{AB})$ in the c-q case,
\begin{align}\label{eq:cq.com.max.program}
		\CompMaxDiv^{X|B}(\hat\rho_{XB}\|\hat{\sigma}_{XB}) =
		\log\sup_{F_{XB}\in \cC^{\mathbf E_\mathrm{eff}}_{n,X|B}}
		\frac{\Tr(\hat\rho_{XB}F_{XB})}{\Tr(\hat\sigma_{XB}F_{XB})}\,,
		\qquad
		\hat\rho_{XB}=\ketbra{1}{1}_X\otimes\rho_B,\ \hat\sigma_{XB}=\ketbra{1}{1}_X\otimes\sigma_B\,.
	\end{align}

Since $\hat\rho_{XB}$ and $\hat\sigma_{XB}$ are block-diagonal in the $X$ basis, \Cref{subsec:cq.cone.YHK} implies that we may w.l.o.g. restrict to $X$-diagonal optimizers. By construction, for every $F_{XB} \in \tilde{\cC}^{\mathbf E_\mathrm{eff}}_{n,X|B}$, there exists an $E_{XB}=\sum_{i=1}^{|X|}\ketbra{i}{i}_X\otimes E_B^i$ for some $\{E_B^i\}_{i=1}^{|X|}\in\overline{\mathbf M}^{p(n)}_{n,|X|}$ such that $F_{XB} = \lambda \cdot E_{XB}$. Then, since the RHS of Eq.~\eqref{eq:cq.com.max.program} is scaling invariant and due to the form states $\hat\rho_{XB}$ and $\hat\sigma_{XB}$, we might restrict to
\begin{align*}
    	\Tr(\hat\rho_{XB}F_{XB})=\Tr(E_B^1\rho_B), \qquad \Tr(\hat\sigma_{XB}F_{XB})=\Tr(E_B^1\sigma_B)\,.
\end{align*}
Then,
	\begin{align*}
		\CompMaxDiv^{X|B}(\hat\rho_{XB}\|\hat{\sigma}_{XB}) = \sup_{\overline{\mathbf M}^{p(n)}_{n,|X|}} \log \frac{\Tr(E_B^1\rho)}{\Tr(E_B^1\sigma)}.
	\end{align*}
\end{proof}

Mirroring the information-theoretic setting, computational indistinguishability and computational unpredictability are captured here by the computational max-divergence~\cite{YHK25} and the computational conditional min-entropy~\cite{AHRA25}, respectively. Both quantities arise from the same order induced by efficiently implementable channels, placing them within a single framework.

\subsection{Non-conditional Computational Min-Entropy}
\label{sec:non.cond.case}
We now consider the case in which there is no side information. Following \cite[Definition~26]{AHRA25}, we obtain a single-system quantity by adjoining a trivial conditioning register. This is the non-conditional specialization of $\CompMin(A|B)$ in our framework.

It is important to distinguish this construction from the flagged c-q reduction of \Cref{lem:max.div.eq.Renyi}. In the latter, one adjoins a classical flag on the first register and obtains a formulation in terms of efficient POVMs. In the present case, however, the conditioning register is one-dimensional, and the corresponding cone is generated by efficient state-preparation maps. Thus, the resulting quantity is not, in general, the same as the measured divergence of \Cref{lem:max.div.eq.Renyi}.

\begin{definition}[Non-conditional computational min-entropy]\label{def:non.cond.min}
Let $B$ be a one-dimensional Hilbert space spanned by $\ket{1}_B$. For any
state $ \rho_A \in \Ss(\HH_n^A)$, define
\begin{align*}
    \hat{\rho}_{AB} \coloneq  \rho_A \otimes \ketbra{1}{1}_B \,.
\end{align*}
The non-conditional computational min-entropy of $\rho_A$ is
\begin{align*}
    \CompMin(A)_{\rho} \coloneq \CompMin(A|B)_{\hat{\rho}}\, .
\end{align*}
\end{definition}

Since the conditioning register is trivial, the infimum over conditioning states
disappears. The next lemma rewrites the resulting quantity in the form naturally
associated with the $A|B$ cone. Concretely, the optimization is over the
convex hull of efficiently preparable states on $A$.

\begin{lem}\label{lem:non.cond.min}
Let $B$ be a one-dimensional Hilbert space spanned by $\ket{1}_B$, and let
\begin{align*}
    \hat{\rho}_{AB} &\coloneq \rho_A \otimes \ketbra{1}{1}_B\,, \\
    \hat{\id}_{AB} &\coloneq \id_A \otimes \ketbra{1}{1}_B\, .
\end{align*}
Let $A'\cong A$, and define the set of efficiently preparable states on $A'$ by
\begin{align*}
    \Ss^{p(n)}_{n,A} \coloneq  \left\{ \tau_{A} \in \Ss(\HH_n^{A}) \,\Big|\, \tau_{A} = \Lambda(\ketbra{1}{1}_B) \text{ for some efficient CPTP map } \Lambda:\BB(\HH^B)\to\BB(\HH_n^{A}) \right\}\,.
\end{align*}
Then
\begin{align*}
    \CompMin(A)_{\rho} = -\log \sup_{\tau_{A} \in \Ss^{p(n)}_{n,A}} \Tr[\rho_A \tau_{A}^{\,\mathsf T}]\,,
\end{align*}
where the transpose is taken with respect to the computational basis used to define $\ket{\Omega_{AA'}}$.
\end{lem}

\begin{proof}
By definition of the non-conditional computational min-entropy,
\begin{align*}
    \CompMin(A)_{\rho} = \CompMin(A|B)_{\hat{\rho}} = - \inf_{\substack{\sigma_B \geq 0 \\ \Tr[\sigma_B]=1}} \CompMaxDiv^{A|B}(\hat{\rho}_{AB}\| \id_A \otimes \sigma_B)\,.
\end{align*}
Since $B$ is one-dimensional, the only state on $B$ is $\ketbra{1}{1}_B$. Hence
\begin{align*}
    \CompMin(A)_{\rho} = -\CompMaxDiv^{A|B}(\hat{\rho}_{AB}\|\hat{\id}_{AB})\,.
\end{align*}

Now apply \Cref{thm:op.mean.fullyquantum} to the bipartite state $\hat{\rho}_{AB}$. This gives
\begin{align*}
    \CompMin(A)_{\rho} = -\log\!\left(d_A \sup_{\Lambda}
        F\left((\Id_A \otimes \Lambda)(\hat{\rho}_{AB}), \ketbra{\Omega}{\Omega}_{AA'} \right)\right)\,,
\end{align*}
where the supremum is over efficient CPTP maps $\Lambda:\BB(\HH^B)\to \BB(\HH_n^{A'})$.
Since $B$ is one-dimensional, every such $\Lambda$ is a state-preparation map which can be associated with its unique image state $\tau_{A'}=\Lambda(\ketbra{1}{1}_B) \in \Ss(\HH_n^{A'})$ and so
\begin{align*}
    (\Id_A \otimes \Lambda)(\hat{\rho}_{AB}) = (\Id_A \otimes \Lambda)(\rho_A \otimes \ketbra{1}{1}_B) = \rho_A \otimes \tau_{A'}\, .
\end{align*}
Therefore,
\begin{align*}
    \CompMin(A)_{\rho} = -\log\!\left( d_A \sup_{\tau_{A'} \in \Ss^{p(n)}_{n,A'}} F\left( \rho_A \otimes \tau_{A'}, \ketbra{\Omega}{\Omega}_{AA'}\right)\right)\,.
\end{align*}
Since one argument of the fidelity is pure,
\begin{align*}
    F\left(\rho_A \otimes \tau_{A'},\ketbra{\Omega}{\Omega}_{AA'}\right)
    =\bra{\Omega_{AA'}}(\rho_A \otimes \tau_{A'})\ket{\Omega_{AA'}}\,.
\end{align*}
Using the identity
\begin{align*}
    d_A\bra{\Omega_{AA'}}(X_A \otimes Y_{A'})\ket{\Omega_{AA'}} =\Tr[X_A Y_{A}^{\,\mathsf T}]\,,
\end{align*}
valid for all compatible operators \(X_A\) and \(Y_{A'}\), we obtain
\begin{align*}
    d_A\bra{\Omega_{AA'}}(\rho_A \otimes \tau_{A'})\ket{\Omega_{AA'}}
    =\Tr[\rho_A \tau_{A}^{\,\mathsf T}]\,.
\end{align*}
Substituting this into the previous expression proves the claim.
\end{proof}

This expression mirrors the ordinary non-conditional min-entropy, for which it holds that, and which we recover when dropping all computational constraints
\begin{align} 
H_{\min}(A)_{\rho} = -\log \|\rho_A\|_\infty = -\log \sup_{\tau \in \mathrm{S}(\HH_n^A)} \Tr[\rho_A \tau]\,. 
\end{align} 
The only difference is that the unrestricted optimization over all states is replaced here by an optimization over efficiently preparable ones.

Moreover, \Cref{lem:non.cond.min} shows that the non-conditional min-entropy is determined by the maximal overlap with the transpose of an efficiently preparable state. For positive operators, it therefore defines a computational analogue of the largest eigenvalue. This motivates the associated computational operator norm on Hermitian operators, defined by taking the absolute value of the overlap. When the transposed efficiently preparable states span $\Herm(\HH_n^A)$, this norm coincides with the order-unit norm associated with the dual cone of $\cone(\Ss^{\,\mathsf T,p(n)}_{n,A})$; see \cite{RKW11,Jen14}.

\begin{cor}[Operational norm in the non-conditional setting]\label{cor:comp.oper.norm}
Let
\begin{align*}
    \Ss^{\,\mathsf T,p(n)}_{n,A} \coloneqq \left\{ \tau_A^{\,\mathsf T}\;:\;\tau_A \in \Ss^{p(n)}_{n,A}\right\}\,,\qquad\cC^{\Ss^{\mathsf T}}_{n,A}\coloneqq\cone\!\left(\Ss^{\,\mathsf T,p(n)}_{n,A}\right)\subseteq \Herm(\HH_n^A)\,,
\end{align*}
where the transpose is taken with respect to the computational basis. For $X_A \in \Herm(\HH_n^A)$, define
\begin{align*}
    \|X_A\|_{\infty,\mathrm{comp}}\coloneqq\sup_{\omega_A\in \Ss^{\,\mathsf T,p(n)}_{n,A}} \left|\Tr[X_A \omega_A]\right|\,.
\end{align*}
Then $\|\cdot\|_{\infty,\mathrm{comp}}$ is a seminorm on $\Herm(\HH_n^A)$. If further
\begin{align*}
    \Span_{\RR}\!\left(\Ss^{\,\mathsf T,p(n)}_{n,A}\right) =\Herm(\HH_n^A)\,,
\end{align*}
then $\|\cdot\|_{\infty,\mathrm{comp}}$ is a norm. In this case,
\begin{align*}
    \|X_A\|_{\infty,\mathrm{comp}}=\inf\Bigl\{t\ge 0\;:\;-t\,\id_A \le_{(\cC^{\Ss^{\mathsf T}}_{n,A})^*} X_A \le_{(\cC^{\Ss^{\mathsf T}}_{n,A})^*} t\,\id_A\Bigr\}\,.
\end{align*}
For $X_A \ge 0$ this is given by
\begin{align*}
    \|X_A\|_{\infty,\mathrm{comp}} = \sup_{\tau_A\in \Ss^{p(n)}_{n,A}} \Tr[X_A \tau_A^{\,\mathsf T}]\,.
\end{align*}
\end{cor}

\begin{proof}
Let us first prove that $ \|\cdot\|_{\infty,\mathrm{comp}}$ defines a seminorm. For every $\alpha \in \RR$ and $X_A,Y_A \in \Herm(\HH_n^A)$, one has
\begin{align*}
    \|\alpha X_A\|_{\infty,\mathrm{comp}} =\sup_{\omega_A \in \Ss^{\,\mathsf T,p(n)}_{n,A}}\left|\Tr[\alpha X_A \omega_A]\right|=|\alpha|\,\|X_A\|_{\infty,\mathrm{comp}}\,,
\end{align*}
and
\begin{align*}
    \|X_A+Y_A\|_{\infty,\mathrm{comp}}
    &=\sup_{\omega_A \in \Ss^{\,\mathsf T,p(n)}_{n,A}}\left|    \Tr[(X_A+Y_A)\omega_A]\right|\\
    &\le\sup_{\omega_A \in \Ss^{\,\mathsf T,p(n)}_{n,A}}\left(\left|\Tr[X_A\omega_A]\right|+\left|\Tr[Y_A\omega_A]\right|\right)\\
    &\le\|X_A\|_{\infty,\mathrm{comp}}+\|Y_A\|_{\infty,\mathrm{comp}}.
\end{align*}
Hence $\|\cdot\|_{\infty,\mathrm{comp}}$ is a seminorm on $\Herm(\HH_n^A)$.

Assume now that
\begin{align*}
    \Span_{\RR}\!\left(\Ss^{\,\mathsf T,p(n)}_{n,A}\right)=\Herm(\HH_n^A)\,.
\end{align*}
If $\|X_A\|_{\infty,\mathrm{comp}}=0$, then
\begin{align*}
    \Tr[X_A\omega_A]=0\qquad\forall\,\omega_A \in \Ss^{\,\mathsf T,p(n)}_{n,A}.
\end{align*}
By linearity, the same holds for every element of
$\Span_{\RR}(\Ss^{\,\mathsf T,p(n)}_{n,A})=\Herm(\HH_n^A)$. Since the Hilbert--Schmidt inner product is non-degenerate on $\Herm(\HH_n^A)$, it follows that $X_A=0$. Therefore $\|\cdot\|_{\infty,\mathrm{comp}}$ is a norm.

Under the same assumption, the cone $\cC^{\Ss^{\mathsf T}}_{n,A}$ is proper, and the corresponding dual order
is well defined. Let $t \ge 0$. Then
\begin{align*}
    -t\,\id_A \le_{(\cC^{\Ss^{\mathsf T}}_{n,A})^*} X_A \le_{(\cC^{\Ss^{\mathsf T}}_{n,A})^*} t\,\id_A \quad \Longleftrightarrow \quad  \Tr[(t\,\id_A \pm X_A)\eta_A] \ge 0 \qquad \forall\,\eta_A \in \cC^{\Ss^{\mathsf T}}_{n,A}\,.
\end{align*}

Since $\cC^{\Ss^{\mathsf T}}_{n,A}=\cone(\Ss^{\,\mathsf T,p(n)}_{n,A})$ and every $\omega_A \in \Ss^{\,\mathsf T,p(n)}_{n,A}$ has trace one, every nonzero $\eta_A \in \cC^{\Ss^{\mathsf T}}_{n,A}$ can be written as
\begin{align*}
    \eta_A = \lambda\,\bar{\eta}_A
\end{align*}
for some $\lambda=\Tr[\eta_A]>0$ and some
\begin{align*}
    \bar{\eta}_A
    \in
    \left\{
        \zeta_A \in \cC^{\Ss^{\mathsf T}}_{n,A}
        \;:\;
        \Tr[\zeta_A]=1
    \right\}
    =
    \conv\!\left(\Ss^{\,\mathsf T,p(n)}_{n,A}\right)\;.
\end{align*}
Hence, the preceding condition is equivalent to
\begin{align*}
    \Tr[(t\,\id_A \pm X_A)\bar{\eta}_A]
    \ge 0
    \qquad
    \forall\,\bar{\eta}_A \in \conv\!\left(\Ss^{\,\mathsf T,p(n)}_{n,A}\right)\;,
\end{align*}
or, equivalently,
\begin{align*}
    \left|
        \Tr[X_A\bar{\eta}_A]
    \right|
    \le t
    \qquad
    \forall\,\bar{\eta}_A \in \conv\!\left(\Ss^{\,\mathsf T,p(n)}_{n,A}\right)\;.
\end{align*}
This, in turn, is equivalent to
\begin{align*}
    \sup_{\omega_A \in \Ss^{\,\mathsf T,p(n)}_{n,A}}
    \left|
        \Tr[X_A\omega_A]
    \right|
    \le t\,.
\end{align*}
Indeed, one implication is immediate since
$\Ss^{\,\mathsf T,p(n)}_{n,A}\subseteq \conv(\Ss^{\,\mathsf T,p(n)}_{n,A})$. Conversely, if
$\bar{\eta}_A=\sum_i p_i \omega_A^{(i)}$ is a convex combination of elements of
$\Ss^{\,\mathsf T,p(n)}_{n,A}$, then
\begin{align*}
    \left|
        \Tr[X_A\bar{\eta}_A]
    \right|
    &=
    \left|
        \sum_i p_i \Tr[X_A\omega_A^{(i)}]
    \right|
    \le
    \sum_i p_i
    \left|
        \Tr[X_A\omega_A^{(i)}]
    \right|
    \le
    \sup_{\omega_A \in \Ss^{\,\mathsf T,p(n)}_{n,A}}
    \left|
        \Tr[X_A\omega_A]
    \right|.
\end{align*}
Taking the infimum over all $t \ge 0$ proves the order-unit formula.

If $X_A \ge 0$, then since every $\omega_A \in \Ss^{\,\mathsf T,p(n)}_{n,A}$ is positive semidefinite the absolute value is redundant since $\Tr[X_A\omega_A]\ge 0$.

\end{proof}

Therefore, \Cref{lem:non.cond.min} together with \Cref{cor:comp.oper.norm} implies that,
\begin{align}
    \CompMin(A)_\rho = -\log \|\rho_A\|_{\infty,\mathrm{comp}}\,,
\end{align}
mirroring the information-theoretic case.

Finally, the transpose in the preceding formulas is an artifact of the Choi representation used in the reduction to the non-conditional setting, rather than an intrinsic feature of the quantity. If the efficient state family is closed under computational-basis transpose up to polynomial overhead, then there exists a polynomial $q$ such that
\begin{align*}
    \Ss^{\,\mathsf T,p(n)}_{n,A}
    \subseteq
    \Ss^{q(n)}_{n,A}\;,
\end{align*}
so that, after replacing $p(n)$ by $q(n)$, the same optimizations may be written directly in terms of efficiently preparable states, without an explicit transpose.

Taken together, these observations show that the non-conditional computational min-entropy is governed by a computational analogue of the operator norm. For positive operators, it is given by the largest overlap with an efficiently preparable state, while under the spanning assumption, it extends to a genuine norm on $\Herm(\HH_n^A)$ through the associated order-unit structure. In this sense, the non-conditional reduction preserves the familiar interpretation of min-entropy as the logarithm of an operator norm, with the usual operator norm replaced by its complexity-constrained counterpart.

\section{Computationally Hidden Quantum Correlations: Separations Between Computational and Information-Theoretic Notions of Min-Entropy}
 Building on \Cref{sec:comp-min-entropy}, we use the operational characterization of $\CompMin(A|B)$ as the \emph{maximal achievable fidelity with the maximally entangled state under complexity-constrained operations} from \Cref{thm:op.mean.fullyquantum} to derive explicit, quantitative bounds on 
 $\CompMin(A|B)$ for specific families of states. In particular, we focus on fully quantum states and  analyze two complementary scenarios. 

In \Cref{sec:sep.pure.state}, we study the gap in the pure state case. Adapting the Haar-subsystem construction of~\cite{LREJ25}, we exhibit highly entangled pure states whose information-theoretic conditional min-entropy is nearly minimal, while their computational conditional min-entropy is bounded far above that minimum even when efficient operations act jointly on polynomially many copies.

Afterwards, in \Cref{sec:sep.mixed.state}, we show that mixed states can hide correlations even more effectively. Via the generalized Hilbert--Schmidt ensemble~\cite{HLW06,Col15,Ha98,ZPN11,BMB+25}, we construct mixed states with highly negative information-theoretic conditional min-entropy for which $\CompMin(A|B)$ is nearly maximal. This means that for any efficient observer holding $B$, the state  essentially appears maximally mixed.

\subsection{Separation for Pure States}\label{sec:sep.pure.state}
We now establish the pure-state separation by bounding how well efficient channels acting on $B$ can recover correlations from a Haar-random state. Following~\cite{LREJ25}, we compare the Haar-random and Haar-subsystem ensembles: they are statistically indistinguishable given polynomially many copies, although their information-theoretic conditional min-entropies differ substantially. We define them as follows:
\begin{align}\label{eq:haar.subs.ensensembles}
    &\EE_{\mathsf{Haar}} = \left\{|\psi_{AB}\rangle \, , \qquad |\psi_{AB}\rangle\sim \mathsf{Haar} (n) \right\}\,, \notag\\
    &\EE_m = \left\{ U_\pi |0 \rangle ^{n-m}|\phi_m \rangle \, : \, U_{\pi} \in P \, , |\phi_m \rangle  \sim \mathsf{Haar} (m) \right\}\,,
\end{align}
where $\EE_m$ is a Haar-subsystem state ensemble, with $P$ is the set of unitaries implementing a random permutation on $n$-bit strings $\pi : \{0,1\}^n \rightarrow \{0,1\}^n $, $U_\pi |x\rangle = |\pi(x)\rangle$. Then, given a function $m:\mathbb{N}\to\mathbb{N}$ such that $m =\omega(\log n_A)$ and $\HH_A \simeq \HH_B$, where $n =n_A + n_B$, both ensembles are statistically indistinguishable given access to $k = \poly(n)$ many copies, as proven in \cite[Lemma 16]{LREJ25}.

Let us now calculate their corresponding values of the (information-theoretic) min-entropy.

\begin{lem}[Min-entropy of a Haar-random state]
\label{lem:min.entr.Haar.state}
Let $k \in \poly(n)$. For a state $|\psi_{AB}\rangle \in \HH^A_n \otimes \HH^B_n$ drawn from the Haar random ensemble $\EE_{\mathsf{Haar}}$, one has, almost surely,
\begin{equation}
    H_{\min}(A^k \mid B^k)_{|\psi\rangle^{\otimes k}}
    = -k\bigl(n_A - O(1)\bigr).
\end{equation}
\end{lem}

\begin{proof}
Since for pure states $H_\text{min}(A|B)_{|\psi\rangle} = -H_\text{max}(A)_\psi$, (c.f. \cite{KRS09}), the conditional min-entropy for a pure state is given by
\begin{align*}
    H_\text{min}(A|B)_{|\psi\rangle}  = -2 \log \left(\Tr \sqrt{\rho_A}\right)\,,
\end{align*}
where $\rho_A = \tr_B(|\psi_{AB}\rangle \langle \psi_{AB}|)$.
Let us first prove the Haar-random case. Given $|\psi_{AB}\rangle\sim \mathsf{Haar}(n)$, \cite[Section III.D]{ZPN11} prove that the reduced state $\rho_A=\tr_B(|\psi_{AB}\rangle\langle\psi_{AB}|)$ can be represented in the following way:
\begin{align}\label{eq:rho_as_genibre}
    \rho_A=\frac{GG^\dagger}{\Tr(GG^\dagger)}\,,
\end{align}
where $G\in \MM_{d_A\times d_B}(\CC)$ is an element of the Ginibre ensemble, i.e., a random matrix with independently and identically distributed complex Gaussian entries \cite{Gi65}.

Then, let us express the functional $\Tr\sqrt{\rho_A}$ in terms of the rescaled empirical eigenvalue distribution of $\rho_A$. Writing $\{\lambda_i\}_{i=1}^{d_A}$ for the eigenvalues of $\rho_A$ and setting
\begin{align*}
    x_i\coloneq d_A\,\lambda_i(\rho_A),
\qquad
\mu_M\coloneq \frac{1}{d_A}\sum_{i=1}^{d_A}\delta_{x_i}\,,
\end{align*}
we obtain
\begin{align*}
    \Tr\sqrt{\rho_A}
    &= \sum_{i=1}^{d_A}\sqrt{\lambda_i(\rho_A)} \\
    &= \sqrt{d_A}\,\frac{1}{d_A}\sum_{i=1}^{d_A}\sqrt{x_i} \\
    &= \sqrt{d_A}\int \sqrt{x}\,d\mu_M(x)\,.
\end{align*}

Then, as proven by \cite[Theorem 5]{Ne07}, the empirical spectral distribution given by $\mu_M$ converges weakly almost surely to the Marchenko-Pastur distribution $\mu_{MP}$
\begin{align}
    \mu_{MP} = \max\{1-c, 0 \}\delta_0+ \frac{\sqrt{(x-a)(b-x)}}{2\pi x}\id_{[a,b]}(x)dx\,,
\end{align}
where $a = (\sqrt{c}-1)^2$ and $b=(\sqrt{c}+1)^2$ and $c \in (0,\infty)$. More specifically, by taking $c=\lim_{n \to \infty} d_B/d_A = 1$, the empirical distribution given by 
\begin{align*}
    \tilde{\mu}_M\coloneq  \frac{1}{d_A}\sum_{i=1}^{d_A}\delta_{y_i}\,,
\end{align*}
with $y_i= c \cdot x_i$ weakly converges to 
\begin{align*}
    \mu_{MP} = \frac{1}{2\pi}\sqrt{\frac{4-x}{x}}\id_{[0,4]}(x)dx\,.
\end{align*}

Since weak convergence implies convergence of integrals only against bounded continuous test functions (Portmanteau theorem, e.g.~\cite[Theorem 2.1]{Billingsley1999}), it cannot be applied directly to $x\mapsto \sqrt{x}$, which is unbounded on $[0,\infty)$. We therefore first use \cite[Theorem 6]{Ne07} to control the largest eigenvalue, and hence to obtain an almost-sure bound on the support of $\mu_M$. In particular,
\begin{align*}
    \max_{i} x_i = d_A \lambda_{\max}(\rho_A) \xrightarrow{d_A \to \infty} 4\,,
\end{align*}
since $c=1$. Hence, for every $\varepsilon>0$, almost surely there exists $d_0\in \NN$ such that for all $d_A\geq d_0$,
\begin{align*}
    \supp(\mu_M)\subset [0,4+\varepsilon]\,.
\end{align*} 

Let $g_\varepsilon$ be a continuous, bounded real function such that
\begin{align*}
    g_\varepsilon(x)=\sqrt{x}\qquad \text{for all }x\in[0,4+\varepsilon]\,.
\end{align*}
Then, for all $d_A\geq d_0$,
\begin{align*}
    \int \sqrt{x}\,d\mu_M(x)=\int g_\varepsilon(x)\,d\mu_M(x)\,.
\end{align*}
By weak convergence of $\mu_M$ to $\mu_{MP}$, we obtain
\begin{align*}
    \int g_\varepsilon(x)\,d\mu_M(x)
    \xrightarrow[d_A\to\infty]{a.s.}
    \int g_\varepsilon(x)\,d\mu_{MP}(x)\,.
\end{align*}
Since $\supp(\mu_{MP})=[0,4]$, we also have
\begin{align*}
    \int g_\varepsilon(x)\,d\mu_{MP}(x)
    =
    \int \sqrt{x}\,d\mu_{MP}(x)\,.
\end{align*}

Therefore,
\begin{align*}
    \int \sqrt{x} d{\mu_M} (x) &\xrightarrow{d_A \to \infty} \int  \sqrt{x} d{\mu_{MP}} (x)  = \frac{1}{2\pi}\int_0^4 \sqrt{4-x}dx = \frac{8}{3\pi}\,.
\end{align*}
Therefore, almost surely it holds that,
\begin{align*}
    \Tr\sqrt{\rho_A}= \left( \frac{8}{3\pi} + o(1) \right) \sqrt{d_A}\,,
\end{align*}
and hence,
\begin{align*}
    H_\text{min}(A|B)_{|\psi\rangle}  &= -\log d_A-2\log \left(\frac{8}{3\pi} \right) + o(1)\\
    &= - n_A + O(1)\,,
\end{align*}
almost surely, since $d_A= 2^{n_A}$. 
Lastly, since in the informational case the conditional min-entropy is additive, one has that
\begin{align*}
    H_\text{min}(A^k|B^k)_{|\psi\rangle^{\otimes k}} = k \cdot H_\text{min}(A|B)_{|\psi\rangle}\,.
\end{align*}

\end{proof}

Let us now bound the case in which the state is given by the Haar-subsystem state ensemble.

\begin{lem}[Min-entropy of a Haar-subsystem random state]
\label{lem:min.entr.Haar.subsystem}
Let $k \in \poly(n)$. For a state $|\psi_{AB}\rangle \in \HH^A_n \otimes \HH^B_n$ drawn from the Haar-subsystem random ensemble $\EE_{m(n_A)}$ of Eq.~\eqref{eq:haar.subs.ensensembles}, one has, almost surely,
\begin{equation}
    H_{\min}(A^k | B^k)_{|\psi\rangle^{\otimes k}}
    \ge -k\, m(n_A)\,.
\end{equation}
\end{lem}
\begin{proof}
Similarly to the proof of \Cref{lem:min.entr.Haar.state}, let us express the conditional min-entropy of a pure state in the following way:
\begin{align*}
    H_\text{min}(A|B)_{|\psi\rangle}  = -2 \log \left(\Tr \sqrt{\rho_A}\right)\,,
\end{align*}
where $\rho_A = \tr_B(|\psi_{AB}\rangle \langle \psi_{AB}|)$. On the other hand, when $|\psi_{AB}\rangle\sim \EE_m$, we use the deterministic bound
\begin{align*}
    \Tr\sqrt{\psi_A}\le \sqrt{\operatorname{rank}(\psi_A)}\le 2^{m/2},
\end{align*}

and therefore, 
\begin{align}
H_\mathrm{min}(A|B)_{|\psi\rangle \sim \EE_{m}} \geq-2\log\!\left(2^{m/2}\right) =-m\,.
\end{align}

Lastly, since in the informational case the conditional min-entropy is additive, one has that
\begin{align*}
    H_\text{min}(A^k|B^k)_{|\psi\rangle^{\otimes k}} = k \cdot H_\text{min}(A|B)_{|\psi\rangle}\,.    
\end{align*}
\end{proof}

Let us now prove the main theorem of this section. We characterize the computational min-entropy for a Haar random state (see Eq.\eqref{eq:haar.subs.ensensembles}) and compare it with the information-theoretic one. This yields a large separation between the two entropies. The separation holds unconditionally when one is restricted to polynomially many copies and polynomially generated quantum channels. 

\begin{remark}
We highlight that the choice of functions for $m$ are slightly more general in \Cref{th:lower-bound.ineff}, \Cref{th:lower-bound.ineff.mix}, \Cref{aqc:th:lower-bound.eff}, and \Cref{aqc:th:lower-bound.eff.mix} than in the main text. Concretely, we show here that the results hold for all functions $m(n_A)$ such that $m(n_A)=\omega(\log n_A)$ and $m(n_A)=o( n_A)$. The choice $m(n_A)=\Theta(\log^2 (n_A))$ from the main text is just a more specific choice that satisfies both conditions.
\end{remark}

\begin{theorem}[Lower bound on the computational min-entropy for pure states]\label{th:lower-bound.ineff}
Fix an $A|B$ bipartition with $n_A= n_B=\frac{n}{2}$, a function $m:\mathbb{N}\to\mathbb{N}$ such that $m(n_A)=\omega(\log n_A)$ and $m(n_A)=o(n_A)$, and let $k\in\poly(n)$ be the number of copies of a given state. Then, there exists a family of bipartite pure states $\{| \psi_{AB,n}\rangle\}_n$ such that
\begin{align*}
H_{\min}(A^k|B^k)_{|\psi_{n}\rangle^{\otimes k}} &=-k\,(n_A-O(1))\,,\\
\CompMin(A^k|B^k)_{|\psi_{n}\rangle^{\otimes k}} &\ge -k\,m(n_A)\,.
\end{align*}
\end{theorem}
\begin{proof}
Let $A \cong B$, where $A_k$ denotes the space $A^{\otimes k}$ and $B_k$ denotes $ B^{\otimes k}$, such that $d_{A_k}=d^k_{A}=2^{k \cdot n_A}$ and $d= d_A \cdot d_B = d_A^{2}$. Then,  $T:\BB(\HH^{B_k}_n)\to\BB(\HH^{A'_k}_n)$ acts on $k$ copies. Given the alternative expression for the min-entropy derived in Theorem~\ref{thm:op.mean.fullyquantum},
\begin{align*}
		\CompMin(A^k|B^k)_{|\psi_{n}\rangle^{\otimes k}} &= -\log \left(\ d_A^k\cdot \max_{T}  \bra{\Omega_{A_kA'_k}} (\Id_{A_k} \otimes T)(|\psi\rangle \langle \psi|^{\otimes k})\ket{\Omega_{A_kA'_k}}\right)\\
        &= -\log \left(\ d_A^k\cdot \max_{T} F_T (|\psi\rangle^{\otimes k})\right)\,,
\end{align*}
where $T: \BB(\HH^{B_{k}}_n) \rightarrow\BB(\HH^{A'_{ k}}_n)$ is maximized w.r.t. the set of polynomially implementable quantum channels and $F_T(|\psi\rangle^{\otimes k})\coloneq \bra{\Omega_{A_kA'_k}} (\Id_{A_k} \otimes T)(|\psi\rangle \langle \psi|^{\otimes k})\ket{\Omega_{A_kA'_k}}$.

Note that $F_T(\cdot)$ corresponds to the following observable: given any efficient $T$, let us define a (non-uniform) quantum poly time ($\mathsf{QPT}$) distinguisher $\mathcal{D}_T$ that, on input $k$ copies of a pure state $|\psi\rangle$, applies $(\Id\otimes T)$ and measures the two-outcome POVM $\{\ket{\Omega}\!\bra{\Omega},\,\id-\ket{\Omega}\!\bra{\Omega}\}$ on $AA'$, outputting $1$ if and only if the outcome is $\ket{\Omega}$.
Then
\begin{align*}
    \Pr[\mathcal{D}_T((|\psi\rangle^{\otimes k})=1]=F_T((|\psi\rangle^{\otimes k})\,.
\end{align*}

Let us now calculate the difference in the expected value of the observable $F_T$ for the following ensembles
\begin{align*}
    &\EE_{\mathsf{Haar}} = \left\{|\psi_{AB}\rangle \, , \qquad |\psi_{AB}\rangle\sim \mathsf{Haar} (n) \right\}\,, \\
    &\EE_m = \left\{ U_\pi |0 \rangle ^{n-m}|\phi_m \rangle \, : \, U_{\pi} \in P \, , |\phi_m \rangle  \sim \mathsf{Haar} (m) \right\}\,,
\end{align*}
where $\EE_m$ is a Haar-subsystem state ensemble, where $P$ is the set of unitaries implementing a random permutation on $n$-bit strings $\pi : \{0,1\}^n \rightarrow \{0,1\}^n $, $U_\pi |x\rangle = |\pi(x)\rangle$.

Let us now bound the difference between the expected values,  
\begin{align*}
    &\left|\mathbb{E}_{|\psi\rangle \sim\EE_{\mathsf{Haar}} }F_T(|\psi\rangle^{\otimes k} ) - \mathbb{E}_{|\psi\rangle \sim\EE_{m} }F_T(|\psi\rangle^{\otimes k} ) \right| \\
    &\qquad = \left| \Tr\left[M_T\left( \mathbb{E}_{|\psi\rangle \sim\EE_{\mathsf{Haar}}}(|\psi\rangle\langle\psi|^{\otimes k}) - \mathbb{E}_{|\psi\rangle \sim\EE_{m} }\left(|\psi\rangle\langle\psi|^{\otimes k}\right)\right)  \right] \right|
    \end{align*}
with $M_T \coloneq   (\Id_{A_k} \otimes T^*)(|\Omega_{AA'}\rangle\langle\Omega_{AA'}|)$, where the Haar average of $|\psi\rangle^{\otimes k}$ can be described as
\begin{align*}
   \mathbb{E}_{|\psi\rangle \sim\EE_{\mathsf{Haar}} }\left[|\phi\rangle\langle\phi|^{\otimes k} \right] = \frac{\Pi^{(n,k)}_{\mathsf{Sym}}}{\tr\left( \Pi^{(n,k)}_{\mathsf{Sym}}\right)}= \frac{\Pi^{(n,k)}_{\mathsf{Sym}}}{\binom{d+k-1}{k}}\,, 
\end{align*}
where $\Pi^{(n,k)}_{\mathsf{Sym}}$ is the projector onto the symmetric subspace (see for example \cite{Me24}).  On the other hand, for the Haar-subsystem ensemble $\EE_m$, a sample can be described equivalently as follows: pick a subset $S\subseteq[d]$ of cardinality $|S|=d_S\coloneq 2^m$ uniformly at random, and then draw $|\phi\rangle$ Haar-random on the $d_S$-dimensional coordinate subspace $\HH_S\coloneq \Span\{\ket{x}:x\in S\}\subseteq(\CC^2)^{\otimes n}$. Then, its average over all possible permutations can be represented as
\begin{align*} \mathbb{E}_{|\phi\rangle \sim\EE_m}\left[|\phi\rangle\langle\phi|^{\otimes k} \right] = \mathbb{E}_{S:|S|=d_S}\left[\frac{\Pi_\mathsf{Sym}^{(S,k)}}{\Tr(\Pi_\mathsf{Sym}^{(S,k)})} \right] \,.
\end{align*}
Moreover, let us fix a subset of basis labels $S \subseteq[d]$, with $d=d_Ad_B$ with $|S|=d_S$  such that
\begin{align*}
    (a,b) \in S \iff |a\rangle_A \otimes|b\rangle_B \in \HH_S\,,
\end{align*}
where $\HH_S$ is described in the computational basis. Let us now define the set $S_B \subset [d_B]$ of $B$-indices appearing in $S$ and its corresponding projector
\begin{align*}
    S_B\coloneq \left\{ b \in [d_B]\,:\, \exists a \in [d_A]\; \text{s.t.}\; (a,b) \in S\right\},\quad \Pi_{S_B}\coloneq \sum_{b \in S_B} |b\rangle \langle b|\,.
\end{align*}
Then,
\begin{align*}
    \Pi_S = \sum_{(a,b) \in S} |ab\rangle \langle ab| \leq \sum_{b\in S_B}\sum_{a \in [d_A]}|ab\rangle \langle ab|  = \id_{A} \otimes \Pi_{S_B}\,.
\end{align*}
Moreover, because the symmetric subspace of $S$ is trivially contained in $S$, we have $\Pi_\mathsf{Sym}^{(S,k)} \leq \Pi^{\otimes k}_S$,
\begin{align*}
    \Pi_\mathsf{Sym}^{(S,k)}\leq  \id_{A_k} \otimes \Pi^{\otimes k}_{S_B}\,.
\end{align*}

In order to bound the expected value of $F_T$ with respect to the Haar-random ensemble, we first use the fact that it is computationally indistinguishable from the Haar-subsystem random ensemble to estimate a bound $\delta$. Later, we prove that this bound $\delta$ holds with overwhelming probability for computationally bounded observers via Markov's inequality.

Let us first evaluate the asymptotic behavior of the expectation values of the Haar-random ensemble introduced in Eq. \eqref{eq:haar.subs.ensensembles}. 

\begin{align*}
    \mathbb{E}_{|\psi\rangle \sim\EE_{\mathsf{Haar}} }[F_T(|\psi\rangle^{\otimes k} ) ]
    & = \frac{k!}{d^k}\left(1+O\left(\frac{k^2}{d}\right) \right) \Tr\left[M_T\Pi^{(n,k)}_{\mathsf{Sym}}\right] \\
      & \leq \frac{k!}{d^k}\left(1+O\left(\frac{k^2}{d}\right) \right)\\
      & = O\left( \frac{k!}{d^k}\right)\,.
\end{align*}

Similarly, the expected value of $F_T$ for the Haar-susbsytem ensemble is given by
\begin{align*}
 \mathbb{E}_{|\psi\rangle \sim\EE_{m} }[F_T(|\psi\rangle^{\otimes k} )]
    & =  \frac{k!}{d_S^k}\left(1+O\left(\frac{k^2}{d_S}\right) \right)\mathbb{E}_{S:|S|=d_S} \Tr\left[M_T\Pi^{(S,k)}_{\mathsf{Sym}}\right]\\
     & \leq \frac{k!}{d_S^k}\left(1+O\left(\frac{k^2}{d_S}\right) \right)\mathbb{E}_{S:|S|=d_S} \Tr\left[M_T\left(\id_{A_k} \otimes \Pi^{\otimes k}_{S_B}\right)\right]\\
      &=\frac{k!}{d_A^kd_S^k}\left(1+O\left(\frac{k^2}{d_S}\right) \right)\mathbb{E}_{S:|S|=d_S} \Tr\left[ \Pi^{\otimes k}_{S_B}\right]\\
      &\leq \frac{k!}{d_A^k}\left(1+O\left(\frac{k^2}{d_S}\right) \right)\\
      & = O\left( \frac{k!}{d_A^k}\right)\,,
\end{align*}
where in the third line we use the fact that $ \Pi_\mathsf{Sym}^{(S,k)}\leq  \id_{A_k} \otimes \Pi^{\otimes k}_{S_B}$. Then, the fourth line follows from the Choi constraint $\tr_{A_k}[M_T] = \frac{\id_{B_k}}{d_A^k}$ and the fifth since $\Tr\left[ \Pi^{\otimes k}_{S_B}\right]=\mathsf{rank}(\Pi^{\otimes k}_{S_B})\leq d_S^k$ and $\Tr\left[M_T\Pi^{(n,k)}_{\mathsf{Sym}}\right]\leq \Tr\left[M_T\right]=1$.

Using the triangle inequality and the above asymptomatic inequalities, we can see that:
\begin{align*}
    \left|\mathbb{E}_{|\psi\rangle \sim\EE_{\mathsf{Haar}} }F_T(|\psi\rangle^{\otimes k} ) - \mathbb{E}_{|\psi\rangle \sim\EE_{m} }F_T(|\psi\rangle^{\otimes k} ) \right| = O\left( \frac{k!}{d_A^k}\right)\,.
\end{align*}

Therefore, it follows that
\begin{align*}
   \mathbb{E}_{|\psi\rangle \sim\EE_{\mathsf{Haar}} }F_T(|\psi\rangle^{\otimes k} )=  \mathbb{E}_{|\psi\rangle \sim\EE_{m} }F_T(|\psi\rangle^{\otimes k} ) +  O\left( \frac{k!}{d_A^k}\right)\,,
\end{align*}
where
\begin{align*}
     \mathbb{E}_{|\psi\rangle \sim\EE_{m} }F_T(|\psi\rangle^{\otimes k} )\leq \frac{k!}{d_A^k}\left(1+O\left(\frac{k^2}{d_S}\right)\right)\,.
\end{align*}
Thus,
\begin{align*}
   \mathbb{E}_{|\psi\rangle \sim\EE_{\mathsf{Haar}} }F_T(|\psi\rangle^{\otimes k} ) &\leq  \frac{k!}{d_A^k}\left(1+O\left(\frac{k^2}{d_S}\right)\right)\\
   &\leq \frac{2^{k\cdot m(n)}}{d_A^k}\,,
\end{align*}
where the last inequality follows by taking $m(n)=\omega(\log n)$ and  $k=\poly(n)$,  since 
$$\log (k!)\leq k (\log (k)-1) +O(\log(k))\;.$$

Let us now prove that there exists a state in the Haar ensemble that satisfies this bound with high probability. To do so, fix $t=t(n)\in \NN$. Then,
\begin{align*}
    \left(F_T\!\left(|\psi\rangle^{\otimes k}\right)\right)^t
    =
    \tr\!\left[
        M_T^{\otimes t}
        |\psi\rangle\langle\psi|^{\otimes kt}
    \right]\,.
\end{align*}
One can calculate the Haar moment bound for the $kt$ moment 
\begin{align*}
    \mathbb{E}_{|\psi\rangle\sim \EE_{\mathsf{Haar}}}
    \left[
        \left(F_T\!\left(|\psi\rangle^{\otimes k}\right)\right)^t
    \right]
    \le
    \frac{(kt)!}{d^{kt}}
    \left(
        1 + O\!\left(\frac{(kt)^2}{d}\right)
    \right) \;.
\end{align*}

Hence, by Markov's inequality, for
\begin{align*}
    \delta \coloneq  \frac{2^{k\,m(n_A)}}{d_A^k} \;,
\end{align*}

one has
\begin{align*}
    \Pr_{|\psi\rangle\sim \EE_{\mathsf{Haar}}}
    \left[F_T\!\left(|\psi\rangle^{\otimes k}\right)\ge \delta\right]
    &=\Pr_{|\psi\rangle\sim \EE_{\mathsf{Haar}}}\left[\left(F_T\!\left(|\psi\rangle^{\otimes k}\right)\right)^t\ge\delta^t\right] \\
    &\le\frac{\mathbb{E}_{|\psi\rangle\sim \EE_{\mathsf{Haar}}}\left[\left(F_T\!\left(|\psi\rangle^{\otimes k}\right)\right)^t\right]}{\delta^t} \\
    &\le\frac{(kt)!}{2^{kt(n_A+\,m(n_A))}}\left(1 + O\!\left(\frac{(kt)^2}{d}\right)\right).
\end{align*}

Let us now fix a polynomial circuit-size bound $p(n)$ and let $\mathsf{QPT}_{\le p}$ denote the set of channels implementable by circuits of size at most $p(n)$. By \Cref{lem:counting-choi}, there exists a polynomial $r(n)$ such that
\begin{align*}
    |\mathsf{QPT}_{\le p}| \le 2^{r(n)} .
\end{align*}
Therefore, by the union bound,
\begin{align*}
    \Pr_{|\psi\rangle\sim \EE_{\mathsf{Haar}}}\left[\exists\, T\in \mathsf{QPT}_{\le p} :F_T\!\left(|\psi\rangle^{\otimes k}\right)\ge \delta\right]\le 2^{r(n)}\frac{(kt)!}{2^{kt(n_A+\,m(n_A))}}\left(1 + O\!\left(\frac{(kt)^2}{d}\right)\right)\,.
\end{align*}
Since $k\in \poly(n)$, $r\in \poly(n)$, $d_S = 2^{m(n_A)}$, and  $m(n_A)=\omega(\log n_A)$, there exists a polynomial $t=t(n)$  
\begin{align*}
    r(n) + \log((kt)!)=r(n)+kt(n)[\log(kt(n))-1]+\mathcal{O}(\log n)  =o(kt(n_A+\,m(n_A)))
\end{align*}
for all sufficiently large $n$. 
In particular, if the degree of $r$ is $g_r$ and the degree of $k$ is $g_k$, then it suffices to take $t$ to be a polynomial of degree at least $\max\{0,g_r-g_k\}$. 
Hence
\begin{align*}
    \Pr_{|\psi\rangle\sim \EE_{\mathsf{Haar}}}
    \left[\forall\, T\in \mathsf{QPT}_{\le p} : F_T\!\left(|\psi\rangle^{\otimes k}\right) \le \frac{2^{k\,m(n_A)}}{d_A^k}\right]\ge 1-2^{-\Omega(n)}\;.
\end{align*}

By the previous estimate, the computational bound holds with probability at least $1-2^{-\Omega(n)}$ over $|\psi\rangle\sim \EE_{\mathsf{Haar}}$.
By Lemma~\ref{lem:min.entr.Haar.state}, the informational bound
$H_{\min}(A^k|B^k)_{|\psi\rangle^{\otimes k}}=-k(n_A-O(1))$
holds almost surely for Haar-random $|\psi\rangle$.
Hence, for all sufficiently large $n$, these two events have nonempty intersection. Then, there exists a pure state $|\psi_{AB,n}\rangle$ in this intersection such that
\begin{align*}
    \max_{T\in \mathsf{QPT}_{\le p}}
    F_T\!\left(|\psi_{AB,n}\rangle^{\otimes k}\right)
    \le
    \frac{2^{k\,m(n_A)}}{d_A^k}\;,
\end{align*}
and thus
\begin{align*}
   \CompMin(A^k|B^k)_{|\psi_{n}\rangle^{\otimes k}}
    &=
    -\log\!\left(
        d_A^k
        \max_{T\in \mathsf{QPT}_{\le p}}
        F_T\!\left(|\psi_{AB,n}\rangle^{\otimes k}\right)
    \right) \\
    &\ge
    -\log\!\left(
        d_A^k \cdot \frac{2^{k\,m(n_A)}}{d_A^k}
    \right) \\
    &=
    -k\,m(n_A)\,.
\end{align*}
\end{proof}

\subsection{Separation for Mixed States}\label{sec:sep.mixed.state}
In this section we turn to mixed states. We begin by extending our notion of a random ensemble beyond pure states, using the canonical construction obtained by partially tracing out $m$ qubits from a Haar-random $(n+m)$-qubit state. The resulting induced ensemble, also known as the GHSE ensemble \cite{HLW06,Col15,Ha98,ZPN11,BMB+25}, will be our main object of study.

\begin{definition}[Generalized Hilbert-Schmidt ensemble (GHSE) \cite{BMB+25}]\label{def:GHSE}
Let $\mathsf{Haar}(n+m)$ be the Haar measure on pure states of $n+m$ qubits. Then, the \emph{GHSE} ensemble is defined as
\begin{align*}
    \eta_{n,m} \coloneq  \{\tr_m(|\psi\rangle \langle \psi |)\}_{|\psi\rangle \sim \mathsf{Haar}(n+m)}\,.
\end{align*}
\end{definition}

Similarly to \Cref{th:lower-bound.ineff}, one can explicitly characterize the computational min-entropy of the GHSE ensemble. For mixed states, however, one can obtain a larger gap. We obtain that when restricted to a polynomial number of copies and to polynomially implementable quantum operations, the corresponding computational min-entropy unconditionally scales as the one of a maximally mixed state (up to negligible error) while its information-theoretic min-entropy is upper bounded by $-k\cdot\left(n_A-\omega(\log n)\right)$, where $k$ is the number of copies. Therefore, under complexity-constrained restrictions, one can get an almost maximal gap between the computational and the information-theoretic notion of conditional min-entropy.

\begin{theorem}[Lower bound on the computational min-entropy for mixed states]\label{th:lower-bound.ineff.mix}
Fix an $A|B$ bipartition with $n_A= n_B$, a function $m:\mathbb{N}\to\mathbb{N}$ such that $m(n_A)=\omega(\log n_A)$ and $m(n_A)=o(n_A)$, and let $k\in\poly(n)$ be the number of copies of a given state. Then, there exists a family of bipartite mixed states $\{ \rho_{n,m}\}_n$ such that
\begin{align*}
H_\mathrm{min}(A^k|B^k)_{\rho_{n,m}^{\otimes k}} &\leq -k\cdot\left(n_A-m(n_A)\right)+\negl(n)\,,\\
\CompMin(A^k|B^k)_{\rho_{n,m}^{\otimes k}} &\geq k\cdot n_A- \negl(n)\,.
\end{align*}
\end{theorem}
\begin{proof}
Similarly to the proof of \Cref{th:lower-bound.ineff}, $A \cong B$, where $A_k$ denotes the space $A^{\otimes k}$ and $B_k$ denotes $ B^{\otimes k}$, such that $d_{A_k}=d^k_{A}=2^{k \cdot n_A}$ and $d_n = d_{A} \cdot d_{B} = d^2_{A}$. Then, $T:\BB(\HH^{B_k}_n)\to\BB(\HH^{A'_k}_n)$ acts on $k$ copies. Given the definition of the computational min-entropy,
\begin{align}
		\CompMin(A^k|B^k)_{\rho^{\otimes k}} &= -\log \left(\ d_A^k\cdot \max_{T}  \bra{\Omega_{A_kA'_k}} (\Id_{A_k} \otimes T)(\rho^{\otimes k})\ket{\Omega_{A_kA'_k}}\right)\\
        &= -\log \left(\ d_A^k\cdot \max_{T} F_T (\rho^{\otimes k})\right) \; ,
\end{align}
where $T: \BB(\HH^{B_{k}}_n) \rightarrow\BB(\HH^{A'_{ k}}_n)$ is maximized over the set of polynomially implementable quantum channels and $F_T(\rho^{\otimes k})\coloneq\bra{\Omega_{A_kA'_k}} (\Id_{A_k} \otimes T)(\rho^{\otimes k})\ket{\Omega_{A_kA'_k}}$. 

Note that $F_T(\cdot)$ corresponds to the following observable: given any efficient $T$, let us define a (non-uniform) QPT distinguisher $\mathcal{D}_T$ that, on input $k$ copies of an unknown mixed state $\rho$, applies $(\Id\otimes T)$ and measures the two-outcome POVM $\{\ket{\Omega}\!\bra{\Omega},\,\id-\ket{\Omega}\!\bra{\Omega}\}$ on $AA'$, outputting $1$ if and only if the outcome is $\ket{\Omega}$.
Then
\begin{align*}
    \Pr[\mathcal{D}_T(\rho^{\otimes k})=1]=F_T(\rho^{\otimes k})\,.
\end{align*}
Let us now bound the difference between the expected values of $F_T$ for the Haar random and Haar subsystem ensembles introduced in Eq. \eqref{eq:haar.subs.ensensembles}. As stated in \cite[Theorem D.1]{BMB+25}, the $k$-th moment of the GHSE ensemble can be constructed as
\begin{align}\label{eq:GHSE.decom}
   \mathbb{E}_{\EE_{\rho \sim\eta_{n,m} }}(\rho^{\otimes k})= \mathbb{E}_{\psi \in \mathsf{Haar}(d_nd_m)}[\tr_m(|\psi\rangle \langle \psi|)^{\otimes k}]=\frac{(d_nd_m-1)!}{(d_nd_m+k-1)! }\sum_{\pi \in S_k}d_m^{\text{cycles}(\pi)}\hat{\pi}_n\,,
\end{align}
where $S_k$ is the symmetric group of degree $k$, and $\hat{\pi}_n$  denotes the permutation unitary acting on $\HH_n^{\otimes k}$ and $d_m$ is the dimension spanned by the $m$ qubits that will subsequently be traced out. The function $\text{cycles}(\pi)$ denotes the number of disjoint cycles in the cycle decomposition of $\pi\in S_k$, counting fixed points as cycles of length $1$. In particular,
\begin{align}\label{eq:cycles}
 \sum_{\substack{\pi \in S_k \\ \pi \neq \id}} d_m^{\text{cycles}(\pi)}
= \frac{(d_m+k-1)!}{(d_m-1)!}-d_m^k.   
\end{align}
 Using $M_T \coloneq  (\Id_{A_k} \otimes T^*)(|\Omega_{A_kA'_k}\rangle\langle\Omega_{A_kA'_k}|)$, we find that
\begin{align*}
    &\left|\mathbb{E}_{\rho \sim\eta_{n,m} }F_T(\rho^{\otimes k} )  - F_T\left(\frac{\id_n^{\otimes k}}{d_n^k} \right) \right| \\
     &\quad = \left| \Tr\left[M_T\left( \mathbb{E}_{\EE_{\rho \sim\eta_{n,m} }}(\rho^{\otimes k}) -\frac{\id_n^{\otimes k}}{d_n^k} \right)  \right] \right|\\
    &\quad = \left| \Tr\left[M_T\left( \frac{1}{d_n^k}\left( -\frac{k(k-1)}{2d_nd_m} + O\left( \frac{k^4}{d_n^2 d_m^2}\right)\right)\id_n^{\otimes k} \right. \right. \right.\\
    & \qquad\qquad \, \left. \left. + \frac{1}{(d_nd_m)^k}\left(1-\frac{k(k-1)}{2d_nd_m} +  O\left(\frac{k^4}{d_n^2d_m^2}\right)\right) \sum_{\substack{\pi \in S_k \\ \pi \neq \id}}d_m^{\text{cycles}(\pi)}\hat{\pi}_n \right] \right|\\   
     &\quad \leq   \frac{1}{d_n^k}\left( \frac{k(k-1)}{2d_nd_m} + O\left( \frac{k^4}{d_n^2 d_m^2}\right)\right)\Tr(M_T)   \\
     & \qquad\qquad \, + \frac{1}{(d_nd_m)^k}\left(1+\frac{k(k-1)}{2d_nd_m} +  O\left(\frac{k^4}{d_n^2d_m^2}\right)\right) \sum_{\substack{\pi \in S_k \\ \pi \neq \id}}d_m^{\text{cycles}(\pi)}|\Tr(M_T \hat{\pi}_n)| \\
    &\quad \leq   \frac{1}{d_n^k}\left( \frac{k(k-1)}{2d_nd_m} + O\left( \frac{k^4}{d_n^2 d_m^2}\right)\right) + \frac{1}{(d_nd_m)^k}\left(1+\frac{k(k-1)}{2d_nd_m} +  O\left(\frac{k^4}{d_n^2d_m^2}\right)\right) \sum_{\substack{\pi \in S_k \\ \pi \neq \id}}d_m^{\text{cycles}(\pi)} \\
       &\quad =   \frac{1}{d_n^k}\left( \frac{k(k-1)}{2d_nd_m} + O\left( \frac{k^4}{d_n^2 d_m^2}\right)\right) + \frac{1}{(d_nd_m)^k}\left(1+\frac{k(k-1)}{2d_nd_m} +  O\left(\frac{k^4}{d_n^2d_m^2}\right)\right) \left( \frac{(d_m +k-1)!}{(d_m-1)!}-d_m^k\right)\,.\\
    \end{align*}
The second equality follows from splitting the sum over $S_k$ into the identity permutation and the sum over non-identity permutations, as in \cite[Theorem D.1]{BMB+25}. The first inequality follows from applying the triangle inequality and the second one from the fact that $|\Tr(M_T \hat{\pi}_n)|\leq \Tr(M_T)=1$. The last equality follows by using Eq.~\eqref{eq:cycles}. Finally, using the Sterling approximation, we have
\begin{align*}
     \left( \frac{(d_m +k-1)!}{(d_m-1)!}-d_m^k\right) = d_m^k\left( \frac{k(k-1)}{2d_m}+O\left( \frac{k^4}{d_m^2}\right)\right)\,.
\end{align*}
It then follows that
    \begin{align*}  &\left|\mathbb{E}_{\rho \sim\eta_{n,m} }F_T(\rho^{\otimes k} )  - F_T\left(\frac{\id^{\otimes k}}{d^k} \right) \right| \\
          &\quad \leq   \frac{1}{d_n^k}\left( \frac{k(k-1)}{2d_nd_m} + O\left( \frac{k^4}{d_n^2 d_m^2}\right)\right) + \frac{1}{d_n^k}\left(1+\frac{k(k-1)}{2d_nd_m} +  O\left(\frac{k^4}{d_n^2d_m^2}\right)\right) \left( \frac{k(k-1)}{2d_m}+O\left( \frac{k^4}{d_m^2}\right)\right)\\   
        &\quad =   \frac{1}{d_n^k}\left( \frac{k(k-1)}{2d_nd_m} + \frac{k(k-1)}{2d_m}\right) + O\left(\frac{k^4}{d_n^{k}d_m^2}\right)\\
        &\quad = O \left( \frac{k^2}{d_n^kd_m}\right)\,.\\
    \end{align*}
Therefore we have
\begin{align*}
    \mathbb{E}_{\rho \sim\eta_{n,m} }F_T(\rho^{\otimes k} )\leq  F_T\left(\frac{\id^{\otimes k}}{d_n^k} \right)+  O\left(\frac{k^2}{d_n^kd_m}\right)\,.
\end{align*}
On the other hand, by using the definition of the (informational) min-entropy, it follows that for any efficient CPTP map $T$
\begin{equation}
\label{eq:alpha-upper-by-Hmin}
F_T(\rho^{\otimes k} )\ \le\ \frac{2^{-H_{\min}(A^k|B^k)_{\rho^{\otimes k}}}}{d_A^k}.
\end{equation}

Let us now evalaute the conditional min-entropy for $\rho= \frac{\id_{AB}}{d_Ad_B}$ with $d_n= d_A \cdot d_B$. Then, it directly follows that
\begin{align*}
    H_{\min}(A|B)_{\frac{\id_{AB}}{d_A d_B}}=\log d_A\,.
\end{align*}

By the additivity of the conditional min-entropy one has that
\begin{align*}
    \mathbb{E}_{\rho \sim\eta_{n,m} }F_T(\rho^{\otimes k} ) &\leq \frac{1}{d_n^{k}} +  O\left(\frac{k^2}{d_n^kd_m}\right)\\
    &\leq \frac{1+\varepsilon}{d_n^{k}}\,,
\end{align*}
for any $\varepsilon\in (0,1]$ satisfying $\varepsilon\geq Ck^2/d_m$ for some constant $C\geq 0$.

Let us now prove that there exists a state in the GHSE ensemble that satisfies this bound with high probability. To do so, we define the $t$-th moment of the distribution 
\begin{align*}
    (F_T(\rho^{\otimes k} ))^t = \Tr[M_T\rho^{\otimes k}]^t = \Tr\left( M_T^{\otimes t} \rho^{\otimes k t}\right) 
\end{align*}
By Markov's inequality, one has that for any moment $t \geq 1$,
\begin{align*}
    \Pr_{\rho\sim\eta_{n,m}}\left[F_T(\rho^{\otimes k} ) \geq \delta \right] =  \Pr_{\rho\sim\eta_{n,m}}\left[(F_T(\rho^{\otimes k} ))^t \geq \delta^t \right] \leq \frac{\mathbb{E}_{\rho \sim\eta_{n,m} } (F_T(\rho^{\otimes k} ))^t}{\delta^t}\,.
\end{align*}
Let us first calculate $\mathbb{E}_{\rho \sim\eta_{n,m} } (F_T(\rho^{\otimes k} ))^t$. Taking Eq.\eqref{eq:GHSE.decom}, one has
\begin{align*}
    \mathbb{E}_{\rho \sim\eta_{n,m} } (F_T(\rho^{\otimes k} ))^t &=  \frac{(d_nd_m-1)!}{(d_nd_m+kt-1)! }\sum_{\pi \in S_{kt}}d_m^{\text{cycles}(\pi)}\Tr(M_T^{\otimes t}\hat{\pi}_n)\\
    &\leq  \frac{(d_nd_m-1)!}{(d_nd_m+kt-1)! } \frac{(d_m + kt -1)!}{(d_m-1)! }\\
    &\leq \frac{1}{d_n^{kt}}\exp\left(\frac{kt(kt-1)}{2 d_m} \right)\,.
\end{align*}
Then, by taking $\delta= (1+\varepsilon)/d^{k}_n$, one has that for any moment $t\geq 1$,
\begin{align*}
    \Pr_{\rho\sim\eta_{n,m}}\left[F_T(\rho^{\otimes k} ) \geq \frac{(1+\varepsilon)}{d^{k}_n}\right] \leq \frac{1}{ (1+\varepsilon)^t}\cdot \exp\left(\frac{kt(kt-1)}{2d_m} \right) \leq \exp\left(\frac{k^2t^2}{2d_m} -t\log(1+\varepsilon)\right)\,.
\end{align*}
Fixing $\varepsilon\in (0,1)$ and taking
\begin{align*}
    t\coloneq \left\lfloor\frac{d_m}{k^2}\log(1+\varepsilon)\right\rfloor \geq1\,,
\end{align*}
one gets
\begin{align*}
    \Pr_{\rho\sim\eta_{n,m}}\left[F_T(\rho^{\otimes k} ) \geq \frac{(1+\varepsilon)}{d^{k}_n}\right] \leq  \exp\left(-\frac{d_m}{2k^2} \log(1+\varepsilon)^2\right)\,.
\end{align*}
Let us now fix a polynomial circuit-size bound $p(n)$ and let $\mathsf{QPT}_{\le p}$ denote the set of channels implementable by circuits of size at most $p(n)$. By \Cref{lem:counting-choi}, there exists a polynomial $r(n)$ such that $|\mathsf{QPT}_{\le p}|\le 2^{ r(n)}$.
Therefore, by taking the union bound
\begin{align*}
    \Pr_{ \rho \sim\eta_{n,m} } \left[\exists T \in \mathsf{QPT}_{\le p}\,:\;F_T(\rho^{\otimes k} ) \geq \frac{(1+\varepsilon)}{d^{k}_n}\right] \leq  2^{r(n)}\exp\left(-\frac{d_m}{2k^2} \log(1+\varepsilon)^2\right)\,.
\end{align*}
Let us now fix the variable $\varepsilon$ such that
\begin{align*}
    \varepsilon(n)\coloneq 2k(n)\sqrt{\frac{2(r(n)+n)\ln 2}{d_{m(n)}}}\,.
\end{align*}
Since  $m(n) = \omega(\log n)$ and $k(n), r(n) = \poly(n)$, $\varepsilon$ satisfies for a sufficiently large $n \in \NN$ that $\varepsilon \in (0,1]$ and $\varepsilon\geq Ck^2/d_m$. Moreover, using the standard bound $\log(1+\varepsilon)\geq \varepsilon/2$, it follows that
\begin{align*}
    \frac{d_m}{2k^2} \log(1+\varepsilon)^2\geq \frac{d_m\varepsilon^2}{8k^2}= (r(n)+n)\ln 2\,.
\end{align*}
Hence,
\begin{align*}
    \Pr_{ \rho \sim\eta_{n,m} } \left[\forall T \in \mathsf{QPT}_{\le p}\,:\;F_T(\rho^{\otimes k} ) \leq \frac{(1+\varepsilon)}{d^{k}_n}\right] \geq  
    1-2^{-n}\,.
\end{align*}

Therefore, for $k = \poly (n)$ there exists a state $\tilde{\rho} \in \eta_{n,m} $ such that for every polynomially implementable quantum channel $\mathsf{QPT}_{\leq p}\ni T: \BB(\HH^B_n) \rightarrow\BB(\HH^{A'}_n)$ with at most $p(n)$ gates,
\begin{align*}
    F_T(\tilde{\rho}^{\otimes k})  \leq \frac{1+\varepsilon}{d_n^{k}}\,.
\end{align*}
Finally, by the definition of the computational min-entropy, 
\begin{align*}
  \CompMin(A^k|B^k)_{\tilde{\rho}^{\otimes k}}  &\geq k\cdot n_A -\log(1+\varepsilon) \geq k\cdot n_A -\negl(n)\,.\\
\end{align*}
Let us now calculate the (informational) conditional min-entropy of a state sampled from the GHSE. Since $H_{\min}(A|B)_\rho \le H(A|B)_\rho$ and $H(A|B)_\rho= H(AB)_{\rho}-H(B)_{\rho}$, where $H(\cdot)$ is the von-Neumann entropy (see~\cite[Definition 5.1]{T16}), it follows that 
\begin{align*}
    H_{\min}(A|B)_\rho\leq H(M)_{\ket{\psi}}-H(B)_{\ket{\psi}}\,, 
\end{align*}
where we have used the fact that $H(M)_{\ket{\psi}} = H(AB)_{\ket{\psi}}$ for the bipartite state $|\psi\rangle \in \CC^{d_{A}d_{B}}\otimes \CC^{d_m}$ as defined in \Cref{def:GHSE}. Then, we used the fact that the von Neumann entropy is upper bounded by the rank, obtaining
\begin{align*}
    H(M)_{\ket{\psi}}\leq \log d_m = m
\end{align*}
almost surely. On the other hand, as proven in \cite[Lemma II.4]{HLW06}, since $d_m=\omega(\log n)$
\begin{align*}
   H(B)_{\ket{\psi}}\geq n_A-\negl(n) 
\end{align*}
almost surely as well.
Then, since the conditional min-entropy is additive and $k=\poly(n)$,
\begin{align*}
    H_\text{min}(A^k|B^k)_{\rho^{\otimes k}} \leq -k\cdot\left(n_A-m\right)+\negl(n)\,.
\end{align*}

\end{proof}

Then, given $k=\poly(n)$ copies, there exists a family of mixed states $\{\rho_{n,m}\}_n$ whose computational conditional min-entropy is asymptotically maximal:
\begin{align*}
    \CompMin(A^k|B^k)_{\rho_{n,m}^{\otimes k}} \ge k\cdot n_A-\negl(n),
\end{align*}
while its information-theoretic conditional min-entropy remains highly negative:
\begin{align*}
    H_\text{min}(A^k|B^k)_{\rho_{n,m}^{\otimes k}} \le -k\cdot\bigl(n_A-m(n_A)\bigr)+\negl(n).
\end{align*}
By subadditivity of $\CompMin(A|B)$ and additivity of the ordinary conditional min-entropy, this yields the corresponding per-copy bounds
\begin{align*}
    \CompMin(A|B)_{\rho_{n,m}} \ge n_A-\negl(n),
    \qquad
    H_\text{min}(A|B)_{\rho_{n,m}} \le -\bigl(n_A-m(n_A)\bigr)+\negl(n).
\end{align*}
Thus, in the mixed-state regime, the computational conditional min-entropy is nearly maximal even though the information-theoretic conditional min-entropy remains highly negative, yielding an almost maximal separation between operationally accessible and information-theoretic side information.

\section{Separations Between Computational and Information-Theoretic Notions of Min-Entropy: the Efficient Setting}\label{aqc:sec:hidden-efficient}
The gap between the information-theoretic and the computational conditional min-entropies can also be extended to the case where the states witnessing the separation are efficiently generated. 
The main idea is to replace the Haar random states used in the previous section with efficiently generated pseudorandom states (PRS)~\cite{JLS18,BS20}.
This will hold under the assumption that the Learning With Errors (LWE) problem is hard for quantum polynomial-time algorithms. More specifically, we will work with Asymptotically Random States (ARS) \cite{BS19}: an ARS generator assumes quantum oracle access to a random function and produces states that are statistically indistinguishable from Haar-random states. These generators allow us to construct \emph{scalable} pseudorandom states (PRS), for which the security parameter is independent of the size of the generated state \cite{BS20, BCJZ25}. Then, by replacing the random function by a quantum-secure pseudorandom function (PRF), scalable PRS can be efficiently generated. 

In \Cref{aqc:sec:prs.lwe}, we introduce the cryptographic background necessary for the construction of these pseudorandom states. Moreover, in \Cref{aqc:lem:lwe.length.doubling.prg}, we provide a concrete instantiation of this functionality based on a specific hardness problem: subexponential LWE. While the same gap could be obtained by assuming the coarse-grained hardness of the PRS functionality, we give this explicit instantiation for completeness.

Lastly, in \Cref{aqc:sec:sep.eff.state}, we prove the main results of this section: the separations for efficiently generated states in both the pure- and mixed-state settings. 

\subsection{Pseudorandom States and Learning with Errors}\label{aqc:sec:prs.lwe}

We start by defining the basic cryptographic primitive that allows us to provide the efficient gap: (scalable) pseudorandom quantum states.

\begin{definition}[Scalable Pseudorandom State (PRS) Generator \cite{BS20,BCJZ25}]\label{aqc:def:PRS}
Let $\lambda \in \NN$ be the security parameter. We say that a pair of polynomial-time quantum algorithms $(K,G)$ is a scalable pseudorandom-state (PRS) generator if, for any $n \in \NN$, the following holds:
\begin{itemize}
    \item \textbf{Key generation:} $K(1^n, 1^\lambda)$ is a polynomial algorithm which outputs a classical key $s$.
    \item \textbf{State generation:} Given $s$ in the support of $K(1^n, 1^\lambda)$, the algorithm $G(1^n, 1^\lambda,s)$ outputs an $n$-qubit quantum state.
    \item \textbf{Quantum pseudorandomness:} For any non-uniform QPT distinguisher $\mathcal{D}$ and any $k\in\poly(\lambda)$, there exists a negligible function $\nu(\lambda)$ such that
    \begin{align}
    \left| \Pr_{\underset{\psi \leftarrow G(1^n, 1^\lambda,s)}{s\leftarrow K(1^n,1^\lambda)}} \!\left[ \mathcal D(\psi^{\otimes k})=1 \right] -
\Pr_{\psi\sim\mathcal{E}_\mathrm{Haar}} \!\left[ \mathcal D(\psi^{\otimes k})=1 \right] \right| \le \nu(\lambda)\,.
\end{align}
\end{itemize}
\end{definition}

The classical analogue of pseudorandom states is given by pseudorandom functions (PRFs). Since our distinguishers are quantum, we will need the corresponding notion of quantum-secure pseudorandom functions.

\begin{definition}[Quantum-Secure Pseudorandom Function (QPRF)
\cite{BS20}]
Let $L:\NN\to\NN$ be polynomially bounded, let
$\mathcal{K}=\{\mathcal{K}_\lambda\}_{\lambda\in\NN}$ be an efficiently
samplable family of key distributions, and let
\begin{align*}
    \mathsf{PRF}_\lambda:
    \mathcal{K}_\lambda\times\{0,1\}^{\lambda}
    \longrightarrow
    \{0,1\}^{L(\lambda)}
\end{align*}
be efficiently computable. We say that $\mathsf{PRF}$ is a
quantum-secure pseudorandom function if, for every non-uniform QPT
distinguisher
$\mathcal D=\{\mathcal D_\lambda\}_{\lambda\in\NN}$ with quantum oracle
access, there exists a negligible function $\nu(\lambda)$ such that, for every $\lambda \in \NN$,
\begin{align}
    \left|
    \Pr_{s\leftarrow\mathcal K_\lambda}
    \left[
        \mathcal D_\lambda^{\mathsf{PRF}_\lambda(s,\cdot)}
        (1^\lambda)=1
    \right]
    -
    \Pr_{f\leftarrow
        (\{0,1\}^{L(\lambda)})^{\{0,1\}^{\lambda}}}
    \left[
        \mathcal D_\lambda^f(1^\lambda)=1
    \right]
    \right|
    \leq\nu(\lambda)\,.
    \label{aqc:eq:dist.qprf}
\end{align}
\end{definition}

We call a QPRF \emph{scalable} if its input and output lengths can be chosen as arbitrary polynomially bounded functions of the security parameter~\cite[Definition~3.27]{BCJZ25}. Such scalable QPRFs exist assuming quantum-secure one-way functions~\cite[Fact~3.29]{BCJZ25}.

Moreover, PRS can be instantiated from quantum-secure PRFs. In this chapter, we will focus on the construction based on Asymptotically Random State (ARS) generators.

\begin{definition}[Asymptotically Random State (ARS) Generator \cite{BS20}]\label{aqc:def:ARS}
A QPT algorithm $\mathsf{Gen}$ with input $(1^n, 1^\lambda)$ for $n,\lambda \in \NN$ and quantum oracle access to $U_f: |x,y\rangle \rightarrow |x,y \;\oplus\; f(x) \rangle$ for $f:\{0,1\}^n \rightarrow \{0,1\}^{\poly(n,\lambda)}$, is an ARS generator if there exists a negligible function $\nu(\lambda)$ such that for every  $k\in\poly(\lambda)$, for all natural numbers $n, \lambda$,
\begin{align*}
    \Delta(\mathbb{E}_{\overset{|\phi\rangle \leftarrow\mathsf{Gen}^{U_f}(1^n, 1^\lambda)}{f\leftarrow (\{0,1\}^{\poly(n,\lambda)})^{\{0,1\}^n}}} [|\phi\rangle \langle \phi |^{\otimes k}]\,, \,\mathbb{E}_{|\psi\rangle \sim\mathcal{E}_\mathrm{Haar}} [|\psi\rangle \langle \psi |^{\otimes k}] )\leq \nu(\lambda)\,.
\end{align*}
    
\end{definition}

We use the scalable ARS construction of \cite[Theorem~4.7]{BCJZ25}. More precisely, this result shows that, for any algorithm making at most $k$ oracle calls to the state-preparation isometry, the ARS output is within trace distance
\begin{align}
\varepsilon_{\mathrm{ARS}}(n,\lambda,k)
\le (1+\sqrt{2\pi})\,k\sqrt n\,2^{-\lambda/2} + 2\pi k\,2^{-\lambda} + 2^{-\lambda}
\end{align}
from $k$ copies of a Haar-random $n$-qubit state. Replacing the random function by a (scalable) quantum-secure PRF gives a scalable PRS generator by \cite[Theorem~5.1]{BCJZ25}. In the sequel we use this construction.

For the concrete LWE-based instantiation used below, the security bound satisfies the following estimate.

\begin{lem}[QPRF instantiation of the scalable ARS generator of \cite{BCJZ25}]
\label{aqc:lem:concrete.BS.hybrid}
Suppose the scalable ARS generator of \cite[Theorem~4.7]{BCJZ25} is instantiated, as in \cite[Theorem~5.1]{BCJZ25}, with a scalable quantum-secure PRF. Let $\varepsilon_{\mathrm{PRS}}(n,\lambda,k)$ denote the distinguishing advantage of a QPT distinguisher between $k$ copies of the resulting PRS output and $k$ copies of a Haar-random $n$-qubit state. Then
\begin{align}
\varepsilon_{\mathrm{PRS}}(n,\lambda,k)
\le \, & (1+\sqrt{2\pi})\,k\sqrt n\,2^{-\lambda/2} + 2\pi k\,2^{-\lambda} + 2^{-\lambda} + \varepsilon_{\mathsf{QPRF}}\!\left(
n,\lambda,k
\right)\,,
\label{aqc:eq:concrete.BS.prs.error}
\end{align}
where $\varepsilon_{\mathsf{QPRF}}$ is the distinguishing advantage for the quantum-secure pseudorandom function as in Eq.~\eqref{aqc:eq:dist.qprf}.
\end{lem}

\begin{proof}
Let us prove it by a hybrid argument. Let $\mathsf H_0$ be the experiment in which the distinguisher receives $k$ copies generated by the PRS, namely the ARS generator run with oracle access to $\mathsf{PRF}_s$. Let $\mathsf H_1$ be the same experiment with the PRF oracle replaced by a truly random function. Let $\mathsf H_2$ be
$k$ copies of a Haar-random $n$-qubit state.

By the triangle inequality,
\begin{align*}
\left|\Pr[\mathsf H_0=1]-\Pr[\mathsf H_2=1]\right| \le \left|\Pr[\mathsf H_0=1]-\Pr[\mathsf H_1=1]\right| + \left|\Pr[\mathsf H_1=1]-\Pr[\mathsf H_2=1]\right|.    
\end{align*}

The second term is bounded by \cite[Theorem~4.7]{BCJZ25}. For the first term, a qPRF adversary simulates the $k$ executions of the ARS generator using its oracle, then runs the final distinguisher; see Eq.~\eqref{aqc:eq:dist.qprf}.
\end{proof}

We now state the concrete post-quantum hardness assumption used in this section. We first recall the relevant LWE distributions.

\begin{definition}[Matrix distributions]
Let $q\ge 2$ be an integer. We write $U^{r\times s}_q$ for the uniform distribution over $\ZZ_q^{r\times s}$. For $\sigma>0$, let $D_{q,\sigma}$ denote the discretised Gaussian distribution over $\ZZ_q$ with width $\sigma$, meaning truncating the discretised Gaussian distribution to only its top $\sigma$ values.
Let $D^{r\times s}_{q,\sigma}$ denote the product distribution over $\ZZ_q^{r\times s}$ obtained by sampling all entries independently from $D_{q,\sigma}$.

For parameters $m,\ell,q,\sigma$, define the noisy rank-$\ell$ matrix distribution $L^{m\times m}_{q,\ell,\sigma}$ over $\ZZ_q^{m\times m}$ as follows:
\begin{enumerate}
    \item sample $B,C \leftarrow U^{\ell\times m}_q$ independently;
    \item sample $E \leftarrow D^{m\times m}_{q,\sigma}$;
    \item output $B^T C+E \bmod q$.
\end{enumerate}
\end{definition}

\begin{remark}[Gaussian sampling]
We use an idealized sampling model in which an exact sample from $D_{q,\sigma}^m$ is assigned a randomness cost of $m\log\sigma$ uniform bits. Implementing an approximate sampler requires accounting for both its randomness cost and its statistical error. For further discussion of approximating truncated discrete Gaussian distributions and controlling the resulting statistical error, see~\cite{dwarakanathSamplingdiscreteGaussianslatticebasedcryptographyconstraineddevice2014a}.
\end{remark}

\begin{definition}[Decision LWE]
For parameters $n,m,q,\sigma$, the problem $\mathsf{LWE}_{n,m,q,\sigma}$ is to distinguish the following two distributions over $\ZZ_q^{m\times n}\times \ZZ_q^m$:
\begin{align*}
    (A,u)\,, \qquad A\leftarrow U^{m\times n}_q\,,\quad u\leftarrow U^m_q\,,
\end{align*}
and
\begin{align*}
    (A,As+e)\,, \qquad A\leftarrow U^{m\times n}_q\,,\quad s\leftarrow U^n_q,\quad e\leftarrow D^m_{q,\sigma}\,.
\end{align*}
The advantage of an algorithm is the absolute difference between its acceptance probabilities on these two distributions.
\end{definition}

We use the following fine-grained, post-quantum LWE assumption. It is a subexponential-strength variant of the standard LWE hardness assumptions; see, e.g., \cite{BLP+13,Pei16} for background on LWE parameters and worst-case hardness evidence.

\begin{ass}[Subexponential-strength LWE]
\label{aqc:ass:subexp.lwe}

There exist constants $\alpha>0$, $\beta>1$ and $1>\delta>0$ such that every quantum algorithm of size $\poly(m)$ has advantage at most
\begin{align*}
    2^{-\Omega(\ell^{\alpha})}
\end{align*}
in solving $\mathsf{LWE}_{\ell,m,q,\sigma}$ for parameters satisfying
\begin{align*}
    m &= \ell^{\beta}\,,\\
    q &= \poly(m)\,,\\
    2\sqrt m &\le \sigma \le q^{1-\delta}\,.
\end{align*}
\end{ass}

Moreover, Zhandry proved that the GGM construction of PRFs~\cite{GGM86} is quantum-secure when instantiated with a length-doubling PRG secure against quantum distinguishers~\cite[Theorem~5.5]{Zha12}. Together with the construction of quantum-secure PRGs from quantum-secure one-way functions, this yields quantum-secure PRFs from quantum-secure one-way functions~\cite[Corollary~1.3]{Zha12}.

For the quantitative bounds needed here, however, we avoid this generic route and instead directly construct a length-doubling PRG from \Cref{aqc:ass:subexp.lwe}.

\begin{lem}[Length-doubling PRG from subexponential LWE]
\label{aqc:lem:lwe.length.doubling.prg}
Assume \Cref{aqc:ass:subexp.lwe}. Then there exists a family of pseudorandom generators (PRG), indexed by a public matrix,
\begin{align*}
    G_{\ell,A}:\{0,1\}^{d_\ell}\to \{0,1\}^{2d_\ell}, \qquad A\in\ZZ_q^{m\times \ell},
\end{align*}
where $d_\ell=\ell\log q+m\log\sigma$, such that, for
$A\leftarrow U_q^{m\times \ell}$, every QPT distinguisher
$\mathcal D=\{\mathcal D_\ell\}_{\ell\in\NN}$ has advantage at most
\begin{align*}
    \left| \Pr_{\substack{A\leftarrow U_q^{m\times \ell}\\ r\leftarrow U_{d_\ell}}}    \left[\mathcal {D}_\ell(A,G_{\ell,A}(r))=1\right] - \Pr_{\substack{A\leftarrow U_q^{m\times \ell}\\ u\leftarrow U_{2d_\ell}}} \left[\mathcal{D}_\ell(A,u)=1\right] \right|
    \le 2^{-\Omega(\ell^{\alpha})}\,.
\end{align*}
\end{lem}

\begin{proof}
Fix a public matrix $A\leftarrow U_q^{m\times \ell}$. We use the standard LWE-PRG construction bounding the number of uniform seed bits, or randomness cost, that is required to sample from a distribution that is hard to distinguish from uniform under the LWE hardness assumption.
Let
\begin{align*}
    G_{A,\ell}(s,e)=As+e \bmod q\,,
\end{align*}
where
\begin{align*}
    s\leftarrow U_q^\ell\,,
    \qquad
    e\leftarrow D^m_{q,\sigma}\,.
\end{align*}
We work in an idealized sampling model in which an exact sample $e\sim D_{q,\sigma}^m$ is assigned a randomness cost of $m\log\sigma$ bits. This is an explicit modelling assumption. Therefore, together with the $\ell\log q$ bits specifying $s$, the total input length is $\ell\log q+m\log\sigma$ bits.

Thus, the input size is given by
\begin{align*}
    d_\ell  =  \ell\log q+m\log\sigma\,,
\end{align*}
whereas the output of the function $G_{A,\ell}(s,e)$ has length
\begin{align*}
    L_\ell = m \log q\,.
\end{align*}

Let us now set the parameters of the  $\mathsf{LWE}_{\ell,m,q,\sigma}$ problem such that the corresponding PRG is length-doubling. Let $\delta$ be the constant from \Cref{aqc:ass:subexp.lwe}. Choose
\begin{align*}
    c>\max\left\{1,\frac{1}{2(1-\delta)}\right\},
\end{align*}
and instantiate
\begin{align*}
    m &= \ell^{\beta}\,,\\
    q &= 2^{\lceil c\log_2 m\rceil}\,,\\
    \sigma &= 2\sqrt m\,.
\end{align*}
Then $q=\poly(m)$ and $\sigma\ge 2\sqrt m$. Moreover, since
$c(1-\delta)>1/2$, for all sufficiently large $m$,
\begin{align*}
    \sigma=2\sqrt m\le q^{1-\delta}\,.
\end{align*}
Thus these parameters lie in the range of \Cref{aqc:ass:subexp.lwe}.

Now
\begin{align*}
    \log q &= c\log m+O(1)\,,\\
    \log\sigma &= \frac12\log m+O(1)\,.
\end{align*}
Hence
\begin{align*}
    2d_\ell
    &=
    2\ell\log q+2m\log\sigma\\
    &=
    2c\ell\log m+m\log m+O(m+\ell)\,.
\end{align*}
On the other hand,
\begin{align*}
    L_\ell   =   m\log q   =   cm\log m+O(m)\,.
\end{align*}
Therefore
\begin{align*}
    L_\ell-2d_\ell  =  (c-1)m\log m  -  2c\ell\log m  -  O(m+\ell)\,.
\end{align*}
Since $m=\ell^\beta$ with $\beta>1$, we have $\ell=o(m)$, and hence
\begin{align*}
    2c\ell\log m+O(m+\ell)=o(m\log m)\,.
\end{align*}
Because $c>1$, it follows that $L_\ell\ge 2d_\ell$ for all sufficiently large $\ell$.
Then, one can define the length-doubling generator by truncating:
\begin{align*}
    G^{\mathsf{dbl}}_{A,\ell}(s,e) =\operatorname{Pref}_{2d_\ell}(As+e\bmod q)\,.
\end{align*}

It remains to prove pseudorandomness. By \Cref{aqc:ass:subexp.lwe},
\begin{align*}
    (A,As+e)  \approx_c  (A,u)\,, \qquad u\leftarrow U_q^m\,,
\end{align*}
with distinguishing advantage at most
\begin{align*}
    2^{-\Omega(\ell^{\alpha})}
\end{align*}
against quantum algorithms of size $\poly(m)$. Since
\begin{align*}
    d_\ell= \ell\log q+m\log\sigma =\poly(m)\,,
\end{align*}
every QPT distinguisher of size $\poly(d_\ell)$ is also of size $\poly(m)$. Applying the deterministic map $\operatorname{Pref}_{2d_\ell}$ to the second
component cannot increase distinguishing advantage. Therefore
\begin{align*}
    \left| \Pr_{\substack{A\leftarrow U_q^{m\times \ell}\\ r\leftarrow U_{d_\ell}}} \left[D_\ell(A,G^{\mathsf{dbl}}_{\ell,A}(r))=1\right]
    - \Pr_{\substack{A\leftarrow U_q^{m\times \ell}\\ u\leftarrow U_{q}^m}} \left[D_\ell(A,\operatorname{Pref}_{2d_\ell}(u))=1\right] \right|
    \le 2^{-\Omega(\ell^{\alpha})}\,.
\end{align*}

Since $q$ is a power of two, the binary encoding of $u\leftarrow U_q^m$ is exactly uniform over $\{0,1\}^{L_\ell}$. Because $L_\ell\ge 2d_\ell$, its first $2d_\ell$ bits are distributed as $U_{2d_\ell}$. This proves the claim.
\end{proof}

Moreover, Zhandry's proof of the quantum security of the GGM construction \cite[Theorem~5.5]{Zha12} shows that any quantum-query adversary breaking the GGM PRF on $n$-bit inputs with advantage $\varepsilon$ yields, with only polynomial overhead and a factor $n$ loss in advantage, an oracle-security distinguisher for the underlying length-doubling PRG. By the security conversion for PRGs in \cite[Corollary~5.4]{Zha12}, which is based on \cite[Theorem~4.5]{Zha12}, the length-doubling PRG of \Cref{aqc:lem:lwe.length.doubling.prg} therefore gives a quantum-secure GGM PRF. Consequently, for $n$ and query complexity $\poly(d_\ell)$, the QPRF distinguishing advantage is bounded, up to the polynomial losses in Zhandry's reduction, by
\begin{align*}
    \poly(d_\ell,n)\cdot 2^{-\Omega(\ell^{\alpha})}
    =
    2^{-\Omega(\ell^{\alpha})},
\end{align*}
since $d_\ell=\poly(\ell)$.

\subsection{Separation for Efficiently Generated States}\label{aqc:sec:sep.eff.state}

We now prove the main results of this section. Under the hardness assumption introduced in \Cref{aqc:ass:subexp.lwe}, the unconstrained separations of \Cref{th:lower-bound.ineff,th:lower-bound.ineff.mix} can be reproduced with efficiently generated keyed ensembles. For pure states, \Cref{aqc:th:lower-bound.eff} replaces Haar-random states by scalable pseudorandom states. For mixed states, \Cref{aqc:th:lower-bound.eff.mix} uses a scalable pseudorandom density matrix obtained by tracing out part of a scalable PRS, inspired by the construction of \cite{BMB+25}.

\begin{theorem}[Lower bound on the computational min-entropy for efficiently generated pure states]\label{aqc:th:lower-bound.eff}

Fix an $A|B$ bipartition with $n_A= n_B=\frac{n}{2}$, a function $m:\mathbb{N}\to\mathbb{N}$ such that $m(n_A)=\omega(\log n_A)$ and $m(n_A)=o(n_A)$, and let $k\in\poly(n)$ be the number of copies of a given state. Fix also a polynomial circuit-size bound $p(n)$. Then, assuming the hardness of subexponential strength LWE (\Cref{aqc:ass:subexp.lwe}) there exists an efficiently generated keyed ensemble of bipartite pure states $\{| \psi_{s,AB,n}\rangle\}_{s,n}$ such that, with probability $1-\negl(n)$ over the key $s$,
\begin{align*}
H_{\min}(A^k|B^k)_{|\psi_{s,n}\rangle^{\otimes k}} &=-k\,(n_A-O(1))\,,\\
\CompMin(A^k|B^k)_{|\psi_{s,n}\rangle^{\otimes k}} &\ge -k\,m(n_A)\,.
\end{align*}
\end{theorem}

\begin{proof}
Let $d_A=2^{n_A}$ and $d=d_A^2$. Let $(K,G)$ be the scalable PRS generator constructed above from the scalable ARS construction, instantiated using the QPRF obtained from \Cref{aqc:ass:subexp.lwe}. We choose the security parameters
below and write
\begin{align*}
    |\psi_s\rangle_{AB}=G(1^n,1^\lambda,s)\,,
    \qquad
    s\leftarrow K(1^n,1^\lambda)\,.
\end{align*}

Similarly to \Cref{th:lower-bound.ineff}, let us define, for any $T\in\mathsf{QPT}_{\le p}$, the following observable
\begin{align*}
F_T(\rho^{\otimes k})
\coloneq
\bra{\Omega_{A_kA'_k}}
(\Id_{A_k}\otimes T)(\rho^{\otimes k})
\ket{\Omega_{A_kA'_k}} \,.
\end{align*}
As previously discussed, $F_T(\cdot)$ corresponds to the following observable: given any efficient $T$, let us define a (non-uniform) quantum poly time ($\mathsf{QPT}$) distinguisher $\mathcal{D}_T$ that, on input $k$ copies of a pure state $|\psi\rangle$, applies $(\Id\otimes T)$ and measures the two-outcome POVM $\{\ket{\Omega}\!\bra{\Omega},\,\id-\ket{\Omega}\!\bra{\Omega}\}$ on $AA'$, outputting $1$ if and only if the outcome is $\ket{\Omega}$.
Then
\begin{align*}
    \Pr[\mathcal{D}_T(|\psi\rangle^{\otimes k})=1]=F_T(|\psi\rangle^{\otimes k})\,.
\end{align*}

Equivalently, if $T^*$ denotes the Heisenberg adjoint of $T$, let
\begin{align*}
    M_T \coloneq (\Id_{A_k}\otimes T^*)\!\left(
    |\Omega_{A_kA'_k}\rangle\langle\Omega_{A_kA'_k}|
    \right)\,,
\end{align*}
so that $F_T(\omega)=\Tr[M_T\omega]$ for every state $\omega$ on $A_kB_k$. Thus, for the $t$-th moment,
\begin{align*}
    \left(F_T\!\left(|\psi\rangle^{\otimes k}\right)\right)^t   =  \tr\!\left[       M_T^{\otimes t} |\psi\rangle\langle\psi|^{\otimes kt}  \right]\,,
\end{align*}
and, by the quantum pseudorandomness property of \Cref{aqc:def:PRS}, the expected value of the PRS state satisfies
\begin{align*}
\mathbb E_{|\psi_s\rangle \sim \EE_{PRS}} [F_T(|\psi_s\rangle^{\otimes k})^t] \le \mathbb E_{|\phi\rangle\sim \EE_{\mathsf{Haar}}} [F_T(\phi^{\otimes k})^t] + \varepsilon_{\mathrm{PRS}}(n,\lambda,kt)\,,
\end{align*}
where $\varepsilon_{\mathrm{PRS}}(n,\lambda,kt)$ denotes the distinguishing advantage of the scalable PRS against the above $kt$-copy QPT test.

As shown in the proof of \Cref{th:lower-bound.ineff},
\begin{align*}
\mathbb E_{|\phi\rangle\sim \EE_{\mathsf{Haar}}} [F_T(\phi^{\otimes k})^t] \le \frac{(kt)!}{d^{kt}} \left(1+O\!\left(\frac{(kt)^2}{d}\right)\right)\,.
\end{align*}

Set
\begin{align*}
    \delta\coloneq \frac{2^{k\,m(n_A)}}{d_A^k}\,.
\end{align*}
By Markov's inequality,
\begin{align*}
\Pr_{|\psi_s\rangle \sim \EE_{PRS}}\!\left[ F_T(\psi_s^{\otimes k})\ge \delta \right]
&\le \frac{(kt)!}{2^{kt(n_A+m(n_A))}} \left(1+O\!\left(\frac{(kt)^2}{d}\right)\right)\\
&\quad+ \varepsilon_{\mathrm{PRS}}(n,\lambda,kt) \cdot 2^{kt(n_A-m(n_A))}\,.
\end{align*}

By \Cref{lem:counting-choi}, there is a polynomial $r(n)$ such that $|\mathsf{QPT}_{\le p}|\le 2^{r(n)}$. Therefore, by taking the union bound,

\begin{align*}
    \Pr_{|\psi_s\rangle \sim \EE_{PRS}}\left[\exists\, T\in \mathsf{QPT}_{\le p} :F_T\!\left(|\psi_s\rangle^{\otimes k}\right)\ge \delta\right]&\le 2^{r(n)}\frac{(kt)!}{2^{kt(n_A+\,m(n_A))}}\left(1 + O\!\left(\frac{(kt)^2}{d}\right)\right)\\ & + \varepsilon_{\mathrm{PRS}}(n,\lambda,kt) \cdot 2^{r(n)+kt(n_A-m(n_A))}\,.
\end{align*}

Similarly to \Cref{th:lower-bound.ineff}, since $k\in \poly(n)$, $r(n)\in \poly(n)$, and  $m(n_A)=\omega(\log n_A)$, there exists a polynomial $t=t(n)$  
\begin{align*}
    r(n) + \log((kt)!)=r(n)+kt(n)[\log(kt(n))-1]+\mathcal{O}(\log n)  =o(kt(n_A+\,m(n_A)))
\end{align*}
for all sufficiently large $n$. 

Thus, for this choice of $t=t(n)$,
\begin{align*}
    2^{r(n)} \frac{(kt)!}{2^{kt(n_A+m(n_A))}} \left(1+O\!\left(\frac{(kt)^2}{d}\right)\right) =\negl(n)\,.
\end{align*}

It remains to choose the security parameters so that the second term in the union bound is negligible. Let
\begin{align*}
    B(n)\coloneq r(n)+kt(n_A-m(n_A))+\omega(\log n)\,.
\end{align*}
Since $k,t,r,n_A$ are polynomially bounded and $m(n_A)\le n_A$, we have
$B(n)\in\poly(n)$. It is enough to ensure that
\begin{align*}
    \varepsilon_{\mathrm{PRS}}(n,\lambda,kt)
    \le 2^{-B(n)}\,.
\end{align*}

By the scalable ARS construction of \cite[Theorem~4.7]{BCJZ25}, together with
\Cref{aqc:lem:concrete.BS.hybrid}, the scalable PRS instantiated with a
quantum-secure PRF satisfies
\begin{align}\label{aqc:eq.bound.QPRS.PRF}
\varepsilon_{\mathrm{PRS}}(n,\lambda,kt) \le& (1+\sqrt{2\pi})\,kt\sqrt n\,2^{-\lambda/2} + 2\pi kt\,2^{-\lambda} + 2^{-\lambda} \notag \\
& + \varepsilon_{\mathsf{QPRF}}(n,\lambda,kt,\ell)\,.
\end{align}
Here $\lambda$ is the scalable ARS security parameter and $\ell$ is the LWE security parameter used in the GGM instantiation. By the discussion following \Cref{aqc:lem:lwe.length.doubling.prg}, Zhandry's quantum-security proof of GGM \cite[Theorem~5.5]{Zha12} gives
\begin{align*}
\varepsilon_{\mathsf{QPRF}}(n,\lambda,kt,\ell) \le 2^{-\Omega(\ell^\alpha)}\,.
\end{align*}

We now choose $\lambda=\lambda(n)$ and $\ell=\ell(n)$. First choose
\begin{align*}
    \lambda(n)  \ge  2B(n)+2\log_2(kt)+\log_2 n+O(1)\,.
\end{align*}
Then
\begin{align*}
(1+\sqrt{2\pi})\,kt\sqrt n\,2^{-\lambda/2} + 2\pi kt\,2^{-\lambda} + 2^{-\lambda} \le 2^{-B(n)-1}
\end{align*}
for all sufficiently large $n$. Since $B(n)\in\poly(n)$, we have $\lambda(n)\in\poly(n)$.

Next, since $B(n)$ is polynomially bounded, choose $\ell=\ell(n)\in\poly(n)$ large enough so that
\begin{align*}
    \ell(n)^\alpha \ge C B(n)
\end{align*}
for a sufficiently large constant $C>0$. Consequently,
\begin{align*}
    \varepsilon_{\mathrm{PRS}}(n,\lambda(n),kt)  \le  2^{-B(n)}\,.
\end{align*}
Therefore,
\begin{align*}
    \varepsilon_{\mathrm{PRS}}(n,\lambda(n),kt)
    \cdot
    2^{r(n)+kt(n_A-m(n_A))}
    \le
    2^{-\omega(\log n)}
    =
    \negl(n)\,.
\end{align*}

Combining this with the previous union bound gives
\begin{align*}
\Pr_{|\psi_s\rangle \sim \EE_{\mathrm{PRS}}}
\left[
\forall\,T\in \mathsf{QPT}_{\le p}:
F_T\!\left(|\psi_s\rangle^{\otimes k}\right)
\le
\frac{2^{k\,m(n_A)}}{d_A^k}
\right]
\ge
1-\negl(n)\,.
\end{align*}

Hence, with probability $1-\negl(n)$ over the key $s$,
\begin{align*}
\max_{T\in \mathsf{QPT}_{\le p}}
F_T\!\left(|\psi_s\rangle^{\otimes k}\right)
\le
\frac{2^{k\,m(n_A)}}{d_A^k}\,.
\end{align*}
Using \Cref{thm:op.mean.fullyquantum}, we obtain
\begin{align*}
\CompMin(A^k|B^k)_{|\psi_s\rangle^{\otimes k}}
&=
-\log\!\left(
d_A^k
\max_{T\in \mathsf{QPT}_{\le p}}
F_T\!\left(|\psi_s\rangle^{\otimes k}\right)
\right)\\
&\ge
-k\,m(n_A)\,.
\end{align*}

It remains to prove the information-theoretic bound. Let
\begin{align*}
    \rho_{s,A}\coloneq \Tr_B\ketbra{\psi_s}{\psi_s}\,.
\end{align*}
Let $\kappa=\lceil a n_A\rceil$ be the moment order, where $a>0$ is fixed below, and let $\pi_\kappa$ be the cyclic permutation on $A^{\otimes \kappa}$. Then
\begin{align*}
    \Tr\!\left[   \pi_\kappa\rho_{s,A}^{\otimes \kappa} \right]
    =  \Tr(\rho_{s,A}^\kappa)\,.
\end{align*}
The Hadamard test for the controlled-$\pi_\kappa$ estimates the real part of this quantity, which is $\Tr(\rho_{s,A}^\kappa)$, and therefore has acceptance
probability
\begin{align*}
    \frac{1+\Tr(\rho_{s,A}^\kappa)}{2}\,.
\end{align*}
This test is efficient for $\kappa=O(n_A)$, since $\pi_\kappa$ is implemented by $\kappa-1$ swaps \cite{EAO+02}. Since the PRS indistinguishability bound applies to the acceptance probabilities of this test, and the trace moment appears with the prefactor $1/2$, we get
\begin{align}\label{aqc:eq:PRS.vs.Haar.expected}
\mathbb E_{|\psi_s\rangle\sim\EE_{\mathrm{PRS}}} \Tr(\rho_{s,A}^\kappa) \le \mathbb E_{|\phi\rangle\sim\EE_{\mathsf{Haar}}} \Tr(\phi_A^\kappa) + 2\varepsilon_{\mathrm{PRS}}(n,\lambda,\kappa)\,.
\end{align}

Let us now bound the Haar-random term. For a Haar-random state $|\phi\rangle_{AB}\in\HH_A\otimes\HH_B$, write
\begin{align*}
    \phi_A\coloneq \tr_B[\ketbra{\phi}{\phi}_{AB}]\,.
\end{align*}
By the moment recurrence of induced random density matrices \cite[Sec.~3.3, Eqs.~(37)--(39)]{Ne07}, in the balanced case $d_A=d_B$, the
moments
\begin{align*}
    M_\kappa(d_A) \coloneq \mathbb E_{|\phi\rangle\sim\EE_{\mathsf{Haar}}} \Tr(\phi_A^\kappa)
\end{align*}
satisfy, for $\kappa\ge 3$,
\begin{align}
   M_\kappa(d_A)
   &= \frac{2(2\kappa-1)d_A}{(d_A^2+\kappa-1)(\kappa+1)}M_{\kappa-1}(d_A)\notag\\
   &\quad+\frac{(\kappa-1)^2(\kappa-2)}{(d_A^2+\kappa-1)(d_A^2+\kappa-2)(\kappa+1)} M_{\kappa-2}(d_A)\,.
   \label{aqc:eq:moment.recursive}
\end{align}
Moreover,
\begin{align*}
    M_1(d_A)=1\,,  \qquad  M_2(d_A)=\frac{2d_A}{d_A^2+1}\le \frac{2}{d_A}\,.
\end{align*}

We now bound both coefficients. For $\kappa>1$,
\begin{align*}
 \frac{2(2\kappa-1)d_A}{(d_A^2+\kappa-1)(\kappa+1)}
 &\le \frac{2(2\kappa-1)d_A}{d_A^2(\kappa+1)}\\
 &= \frac{4}{d_A}-\frac{6}{d_A(\kappa+1)}
 \le \frac{4}{d_A}\,.
\end{align*}
For $\kappa>2$,
\begin{align*}
   \frac{(\kappa-1)^2(\kappa-2)}
   {(d_A^2+\kappa-1)(d_A^2+\kappa-2)(\kappa+1)}
   &\le \frac{(\kappa-1)^2(\kappa-2)}{d_A^4(\kappa+1)}\\
   &\le \frac{\kappa^2}{d_A^4}\,.
\end{align*}
Thus
\begin{align*}
    M_\kappa(d_A)
    \le \frac{4}{d_A}M_{\kappa-1}(d_A)
    +\frac{\kappa^2}{d_A^4}M_{\kappa-2}(d_A)\,.
\end{align*}

Let us now prove that there exists a universal constant $C_0\ge 8$ such that, for all moment orders $1\le \kappa \le \lceil a n_A\rceil=O(\log d_A)$  and all sufficiently large $d_A$,
\begin{align*}
    M_\kappa(d_A)\le C_0^\kappa d_A^{1-\kappa}.
\end{align*}
The claim holds for $\kappa=1$, since
\begin{align*}
    M_1(d_A)=1\le C_0 d_A^0\,,
\end{align*}
and for $\kappa=2$, since
\begin{align*}
    M_2(d_A)\le \frac{2}{d_A}\le C_0^2d_A^{-1}\,.
\end{align*}
Assume now that the claim holds for $\kappa-1$ and $\kappa-2$. Then
\begin{align*}
    M_\kappa(d_A)
    &\le \frac{4}{d_A}C_0^{\kappa-1}d_A^{2-\kappa} +  \frac{\kappa^2}{d_A^4}C_0^{\kappa-2}d_A^{3-\kappa}\\
    &=  4C_0^{\kappa-1}d_A^{1-\kappa}  +  \kappa^2C_0^{\kappa-2}d_A^{-1-\kappa}\\
    &= C_0^\kappa d_A^{1-\kappa} \left(    \frac{4}{C_0}    +    \frac{\kappa^2}{C_0^2d_A^2} \right)\,.
\end{align*}
Since $\kappa=O(\log d_A)$, one has $\kappa^2/d_A^2=o(1)$. Therefore, for $C_0\ge 8$ and all sufficiently large $d_A$,
\begin{align*}
   \frac{4}{C_0}  +  \frac{\kappa^2}{C_0^2d_A^2}  \le 1\,.
\end{align*}
Hence
\begin{align*}
    M_\kappa(d_A)\le C_0^\kappa d_A^{1-\kappa}\,.
\end{align*}
Therefore,
\begin{align*}
\mathbb E_{|\phi\rangle\sim\EE_{\mathsf{Haar}}} \Tr(\phi_A^\kappa) \le C_0^\kappa d_A^{1-\kappa}\,, \qquad \kappa=O(\log d_A)\,.
\end{align*}
Combining the previous estimate with Eq. \eqref{aqc:eq:PRS.vs.Haar.expected} gives
\begin{align*}
\mathbb E_{|\psi_s\rangle\sim\EE_{\mathrm{PRS}}} \Tr(\rho_{s,A}^\kappa) \le C_0^\kappa d_A^{1-\kappa} + 2\varepsilon_{\mathrm{PRS}}(n,\lambda,\kappa)\,.
\end{align*}
Since
\begin{align*}
    \|\rho_{s,A}\|_\infty^\kappa \le \Tr(\rho_{s,A}^\kappa)\,,
\end{align*}
Markov's inequality implies that, for any constant $L>C_0$,
\begin{align}\label{aqc:eq:prob.bound.norm}
\Pr_{|\psi_s\rangle\sim\EE_{\mathrm{PRS}}} \left[  \|\rho_{s,A}\|_\infty>\frac{L}{d_A} \right]
 &\le \frac{ C_0^\kappa d_A^{1-\kappa} + 2\varepsilon_{\mathrm{PRS}}(n,\lambda,\kappa) }{ (L/d_A)^\kappa }   \notag\\
 &= d_A\left(\frac{C_0}{L}\right)^\kappa + 2\varepsilon_{\mathrm{PRS}}(n,\lambda,\kappa) \left(\frac{d_A}{L}\right)^\kappa .
\end{align}
Choose $L>C_0$ and $a>0$ such that
\begin{align*}
    a\log_2(L/C_0)>2\,.
\end{align*}
Since $\kappa=\lceil a n_A\rceil$ and $d_A=2^{n_A}$, the first term is exponentially small in $n_A$. It remains to choose the security parameters so that the second term of Eq. \eqref{aqc:eq:prob.bound.norm} is negligible. Since $d_A=2^{n_A}$ and $L$ is constant,
it is enough to ensure
\begin{align*}
    \varepsilon_{\mathrm{PRS}}(n,\lambda,\kappa)
    \le
    2^{-\kappa n_A-\omega(\log n)}\,.
\end{align*}
Indeed,
\begin{align*}
2\varepsilon_{\mathrm{PRS}}(n,\lambda,\kappa) \left(\frac{d_A}{L}\right)^\kappa
&\le 2\cdot 2^{-\kappa n_A-\omega(\log n)}\cdot 2^{\kappa n_A}\\
&= \negl(n)\,.
\end{align*}

As in the computational part of the proof, by Eq.~\eqref{aqc:eq.bound.QPRS.PRF}, we can choose $\lambda=\lambda(n)\in\poly(n)$ such that also satisfies
\begin{align*}
    \lambda(n) \ge 2\kappa n_A+2\log_2(\kappa\sqrt n)+\omega(\log n)\,.
\end{align*}
Then the first three terms in Eq.~\eqref{aqc:eq.bound.QPRS.PRF} are bounded by $2^{-\kappa n_A-\omega(\log n)-1}$. Next choose $\ell=\ell(n)\in\poly(n)$ large enough so that it also satisfies
\begin{align*}
    2^{-\Omega(\ell(n)^\alpha)}\le  2^{-\kappa n_A-\omega(\log n)-1}\,.
\end{align*}
Therefore
\begin{align*}
    \varepsilon_{\mathrm{PRS}}(n,\lambda(n),\kappa) \le
    2^{-\kappa n_A-\omega(\log n)}\,.
\end{align*}

Thus, with probability $1-\negl(n)$ over $s$,
\begin{align*}
    \|\rho_{s,A}\|_\infty\le \frac{L}{d_A}\,.
\end{align*}
For a pure bipartite state,
\begin{align*}
H_{\min}(A|B)_{\psi_s} = -H_{\max}(A)_{\psi_s} = -2\log\Tr\sqrt{\rho_{s,A}}\,.
\end{align*}
The bound $\|\rho_{s,A}\|_\infty\le L/d_A$ implies
\begin{align*}
    \Tr\sqrt{\rho_{s,A}} \ge \sqrt{\frac{d_A}{L}}\,,
\end{align*}
while Cauchy--Schwarz gives
\begin{align*}
    \Tr\sqrt{\rho_{s,A}}\le \sqrt{d_A}\,.
\end{align*}
Therefore
\begin{align*}
    -n_A\le H_{\min}(A|B)_{\psi_s} \le -n_A+\log L  =  -n_A+O(1)\,.
\end{align*}
By additivity of the conditional min-entropy for tensor powers of pure states,
\begin{align*}
H_{\min}(A^k|B^k)_{\psi_s^{\otimes k}} = k\,H_{\min}(A|B)_{\psi_s} = -k(n_A-O(1))\,.
\end{align*}
A union bound over the informational and computational high-probability events completes the proof.

\end{proof}

The previous theorem provides an efficiently generated analogue of \Cref{th:lower-bound.ineff}. With overwhelming probability over the key, the generated state is nearly maximally entangled from the information-theoretic point of view, while its computational conditional min-entropy remains lower bounded by $-k\,m(n_A)$. Therefore, under \Cref{aqc:ass:subexp.lwe}, the separation obtained for Haar-random states persists for efficiently generated pseudorandom states. This remains true even when the recovery channel is chosen using the classical description $(G,s)$ of the state-preparation circuit, since such a choice is included among the efficient channels considered in the theorem. This is analogous to the state-aware results of \cite{LREJ25}, although here the obstruction concerns the one-sided recovery task defining $\CompMin(A|B)$ rather than LOCC distillation or dilution.

Similarly to the pure-state case, the separation obtained in \Cref{th:lower-bound.ineff.mix} for the unconstrained mixed-state construction also persists for efficiently generated states, under the hardness assumption of \Cref{aqc:ass:subexp.lwe}. In this case, the relevant cryptographic primitive is a pseudorandom density matrix (PRDM), which generalizes pseudorandom states to mixed states by requiring computational indistinguishability from the generalized Hilbert--Schmidt ensemble \cite{BMB+25}. We use the following scalable version of this cryptographic primitive.

\begin{definition}[Scalable pseudorandom density-matrix generator]
Let $\lambda\in\NN$ be the security parameter and let $m=m(n)$ be a partial trace parameter. We say that a pair of polynomial-time quantum algorithms $(K,G)$ is a scalable pseudorandom density-matrix (PRDM) generator if, for any $n,m\in\NN$, the following holds:
\begin{itemize}
    \item \textbf{Key generation:} $K(1^n,1^m,1^\lambda)$ outputs a classical key $s$.

    \item \textbf{State generation:} Given $s$ in the support of $K(1^n,1^m,1^\lambda)$, the algorithm \newline $G(1^n,1^m,1^\lambda,s)$ outputs an $n$-qubit density matrix $\rho_{s,n,m}$.

    \item \textbf{Quantum pseudorandomness:} For any non-uniform QPT distinguisher $\mathcal D$ and any $k\in\poly(\lambda)$, there exists a
    negligible function $\nu(\lambda)$ such that
    \begin{align}
    \left| \Pr_{\substack{s\leftarrow K(1^n,1^m,1^\lambda)}}  \!\left[  \mathcal D(\rho_{s,n,m}^{\otimes k})=1  \right]  -  \Pr_{\rho\sim\eta_{n,m}}  \!\left[  \mathcal D(\rho^{\otimes k})=1  \right]  \right| \le \nu(\lambda)\,,
    \end{align}
    where $\eta_{n,m}$ denotes the generalized Hilbert--Schmidt ensemble obtained by tracing out $m$ qubits from a Haar-random $(n+m)$-qubit pure state.
\end{itemize}
\end{definition}

Let us highlight that, unlike the fixed-parameter PRDMs introduced in \cite{BMB+25}, the security parameter $\lambda$ in the previous definition is independent of the dimension of the pseudorandom density matrix. This construction allows us to prove the following separation.

\begin{theorem}[Efficiently generated mixed-state separation]
\label{aqc:th:lower-bound.eff.mix}
Fix an $A|B$ bipartition with $n_A=n_B=n/2$, let $m(n_A)=\omega(\log n_A)$ with and $m(n_A)=o(n_A)$, and let $k\in\poly(n)$ be the number of copies. Fix also a polynomial circuit-size bound $p(n)$ and a constant $c>0$. Assuming \Cref{aqc:ass:subexp.lwe}, there exists an efficiently generated keyed ensemble of bipartite mixed states $\{\rho_{s,AB,n}\}_{s,n}$ such that, with probability $1-\negl(n)$ over the key $s$,
\begin{align*}
H_{\min}(A^k|B^k)_{\rho_{s,n}^{\otimes k}}
&\le -k\bigl(n_A-m(n_A)\bigr)+\negl(n)\,,\\ \CompMin(A^k|B^k)_{\rho_{s,n}^{\otimes k}}
&\ge k n_A-O(n^{-c})\,.
\end{align*}
\end{theorem}

\begin{proof}
Let $d_A=2^{n_A}$, $d_n=d_A^2$, and $d_m=2^{m(n_A)}$. Let $(K,G)$ be the scalable PRS generator constructed above, instantiated using the QPRF obtained from \Cref{aqc:ass:subexp.lwe}. We choose the security parameters below and define
\begin{align*}
    |\psi_s\rangle_{ABM}  = G(1^{n+m(n_A)},1^\lambda,s)\,,  \qquad  s\leftarrow K(1^{n+m(n_A)},1^\lambda)\,,
\end{align*}
and
\begin{align*}
    \rho_{s,AB} = \Tr_M\ketbra{\psi_s}{\psi_s}_{ABM}\,,
\end{align*}
which corresponds to the traced-PRS construction of pseudorandom density matrices of \cite{BMB+25}.

We first prove the computational lower bound. For $T\in\mathsf{QPT}_{\le p}$, define as in \Cref{th:lower-bound.ineff.mix}
\begin{align*}
F_T(\rho^{\otimes k}) \coloneq \bra{\Omega_{A_kA'_k}} (\Id_{A_k}\otimes T)(\rho^{\otimes k}) \ket{\Omega_{A_kA'_k}} \,.
\end{align*}

For a fixed $T$, the test corresponding to $F_T(\rho^{\otimes k})^t$ acts on $kt$ copies: it splits the input into $t$ blocks of $k$ copies, applies the $F_T$ test to each block, and accepts iff all tests accept. Since $k,t\in\poly(n)$ and $T$ has size at most $p(n)$, this is a QPT test. Hence, by PRS security,
\begin{align}
\mathbb E_s F_T(\rho_{s,AB}^{\otimes k})^t \le \mathbb E_{\rho\sim\eta_{n,m}} F_T(\rho^{\otimes k})^t + \varepsilon_{\mathrm{PRS}}(n+m(n_A),\lambda,kt)\,.
\label{aqc:eq:eff.mix.transfer.moment}
\end{align}
As shown in the proof of \Cref{th:lower-bound.ineff.mix}, using
\cite[Theorem~D.1]{BMB+25},
\begin{align}
\mathbb E_{\rho\sim\eta_{n,m}} F_T(\rho^{\otimes k})^t \le d_n^{-kt} \exp\!\left( \frac{k^2t^2}{2d_m} \right)\,.
\label{aqc:eq:eff.mix.GHSE.moment}
\end{align}

Set
\begin{align*}
    \delta_n\coloneq \frac{1+n^{-c}}{d_n^k}\,.
\end{align*}
By Markov's inequality, \eqref{aqc:eq:eff.mix.transfer.moment} and
\eqref{aqc:eq:eff.mix.GHSE.moment},
\begin{align}
\Pr_s\left[ F_T(\rho_{s,AB}^{\otimes k}) \ge \delta_n \right]
&\le \exp\!\left( \frac{k^2t^2}{2d_m} - t\ln(1+n^{-c}) \right) \notag\\
&\quad+ \varepsilon_{\mathrm{PRS}}(n+m(n_A),\lambda,kt)\, \delta_n^{-t}.
\label{aqc:eq:eff.mix.fixed.T.tail}
\end{align}

By \Cref{lem:counting-choi}, there exists a polynomial $r(n)$ such that
$|\mathsf{QPT}_{\le p}|\le 2^{r(n)}$. A union bound over
$T\in\mathsf{QPT}_{\le p}$ gives
\begin{align}
&\Pr_s\left[
\exists T\in\mathsf{QPT}_{\le p}:
F_T(\rho_{s,AB}^{\otimes k}) \ge \delta_n
\right]
\notag\\
&\le
2^{r(n)}
\exp\!\left(
\frac{k^2t^2}{2d_m}
-
t\ln(1+n^{-c})
\right)
+
2^{r(n)}
\varepsilon_{\mathrm{PRS}}(n+m(n_A),\lambda,kt)\,
\delta_n^{-t}\,.
\label{aqc:eq:eff.mix.union.prechoice}
\end{align}

Let $s(n)=\omega(\log n)$ be any polynomially bounded  function, for instance $s(n)=\log^2 n$, and set
\begin{align}
t(n) \coloneq \left\lceil \frac{(r(n)+2s(n))\ln 2}{\ln(1+n^{-c})} \right\rceil \,.
\label{aqc:eq:eff.mix.t.choice}
\end{align}
Since $\ln(1+n^{-c})=\Theta(n^{-c})$ and $r(n),s(n)\in\poly(n)$, this $t(n)$ is polynomially bounded. Moreover,
\begin{align}
t\ln(1+n^{-c}) \ge (r(n)+2s(n))\ln 2\,.
\end{align}
Hence the first term in \eqref{aqc:eq:eff.mix.union.prechoice} is bounded by
\begin{align}
2^{r(n)} \exp\!\left( \frac{k^2t^2}{2d_m} - t\ln(1+n^{-c}) \right)
&\le \exp\!\left(\frac{k^2t^2}{2d_m}\right) 2^{-2s(n)} \,.
\end{align}
Since $k,t\in\poly(n)$ and $d_m=2^{m(n_A)}$, we have
\begin{align*}
    \frac{k^2t^2}{2d_m}=\negl(n)\,.
\end{align*}

Therefore the first term is at most $\negl(n)$ for all sufficiently large $n$. It remains to choose the security parameters so that the second term is
negligible. Since
\begin{align*}
    \delta_n^{-t}  =  \frac{d_n^{kt}}{(1+n^{-c})^t} \le
    d_n^{kt} = 2^{nkt}\,,
\end{align*}
it is enough to ensure
\begin{align}
\varepsilon_{\mathrm{PRS}}(n+m(n_A),\lambda,kt)
\le
2^{-r(n)-nkt-s(n)}.
\label{aqc:eq:eff.mix.security.choice}
\end{align}
Using Eq.~\eqref{aqc:eq.bound.QPRS.PRF}, this is achieved by choosing
$\lambda=\lambda(n)\in\poly(n)$ such that
\begin{align*}
    \lambda(n)  \ge  2\bigl(r(n)+nkt+s(n)\bigr)  +  2\log_2(kt\sqrt{n+m(n_A)})  +  O(1)\,,
\end{align*}
and then choosing $\ell=\ell(n)\in\poly(n)$ large enough so that
\begin{align*}
    2^{-\Omega(\ell(n)^\alpha)} \le 2^{-r(n)-nkt-s(n)}\,.
\end{align*}
This is possible because the functions $k,t,r,n$ and $m(n_A)$ are polynomially bounded. Hence \eqref{aqc:eq:eff.mix.security.choice} holds.

Combining \eqref{aqc:eq:eff.mix.union.prechoice}--\eqref{aqc:eq:eff.mix.security.choice}, we obtain, with probability $1-\negl(n)$ over $s$,
\begin{align}
\max_{T\in\mathsf{QPT}_{\le p}} F_T(\rho_{s,AB}^{\otimes k}) \le \frac{1+n^{-c}}{d_n^k}\,.
\label{aqc:eq:eff.mix.uniform.comp}
\end{align}
By \Cref{thm:op.mean.fullyquantum},
\begin{align*}
\CompMin(A^k|B^k)_{\rho_{s,AB}^{\otimes k}}
&= -\log\left( d_A^k \max_{T\in\mathsf{QPT}_{\le p}} F_T(\rho_{s,AB}^{\otimes k}) \right)\\
&\ge -\log\left( d_A^k \frac{1+n^{-c}}{d_n^k} \right)\\
&= k n_A-\log(1+n^{-c})\\
&= k n_A-O(n^{-c})\,.
\end{align*}

Let us now prove the information-theoretic upper bound. Let
\begin{align*}
    \rho_{s,B} = \Tr_{AM}\ketbra{\psi_s}{\psi_s}_{ABM}\,.
\end{align*}
We will show that $\rho_{s,B}$ has almost minimal purity. Recall that the purity $\Tr(\rho_{s,B}^2)$ is the expectation value of the swap operator on two copies of the $B$ register. Similarly, $\bigl(\Tr(\rho_{s,B}^2)\bigr)^2$ is the expectation value of the tensor product of two such swap operators on four copies. Therefore, applying the corresponding Hadamard tests and using PRS
indistinguishability from Haar-random states, we obtain
\begin{align}
\left| \mathbb E_s \Tr(\rho_{s,B}^2) - \mathbb E_{\phi\sim\mathsf{Haar}(ABM)} \Tr(\phi_B^2) \right|
&\le 2\varepsilon_{\mathrm{PRS}}(n+m(n_A),\lambda,2)\,, \label{aqc:eq:eff.mix.purity.transfer.1}\\
\left| \mathbb E_s \bigl(\Tr(\rho_{s,B}^2)\bigr)^2 - \mathbb E_{\phi\sim\mathsf{Haar}(ABM)} \bigl(\Tr(\phi_B^2)\bigr)^2 \right|
&\le 2\varepsilon_{\mathrm{PRS}}(n+m(n_A),\lambda,4)\,.
\label{aqc:eq:eff.mix.purity.transfer.2}
\end{align}

For a Haar-random pure state on $ABM$, the reduced state $\phi_B$ is an induced random density matrix with system dimension $d_A$ and environment dimension $d_A d_m$. Hence, by \cite[Eq.~(39)]{Ne07},
\begin{align}
\mathbb E_{\phi\sim\mathsf{Haar}(ABM)} \Tr(\phi_B^2) = \frac{d_A+d_A d_m}{d_A^2d_m+1} \le \frac{1+d_m^{-1}}{d_A}\,.
\label{aqc:eq:eff.mix.haar.purity.mean}
\end{align}
Moreover, by the exact second cumulant of the purity distribution
\cite[Eq.~(14)]{Gir07}, applied with dimensions $d_A$ and $d_A d_m$,
\begin{align}
\mathsf{Var}_{\phi\sim\mathsf{Haar}(ABM)} \!\left(\Tr(\phi_B^2)\right)
&= \mathbb E_{\phi\sim\mathsf{Haar}(ABM)} \left( \Tr(\phi_B^2) - \mathbb E_{\phi\sim\mathsf{Haar}(ABM)}\Tr(\phi_B^2) \right)^2 \notag\\
&= \frac{ 2(d_A^2-1)(d_A^2d_m^2-1) }{ (d_A^2d_m+1)^2(d_A^2d_m+2)(d_A^2d_m+3) } \notag\\
&\le \frac{2}{d_A^4d_m^2}\,.
\label{aqc:eq:eff.mix.haar.purity.var}
\end{align}

Equations \eqref{aqc:eq:eff.mix.purity.transfer.1} and \eqref{aqc:eq:eff.mix.purity.transfer.2} imply
\begin{align}
&\mathbb E_s \left( \Tr(\rho_{s,B}^2) - \mathbb E_{\phi\sim\mathsf{Haar}(ABM)}\Tr(\phi_B^2) \right)^2 \notag\\
&= \mathbb E_s \bigl(\Tr(\rho_{s,B}^2)\bigr)^2 - 2 \mathbb E_{\phi\sim\mathsf{Haar}(ABM)}\Tr(\phi_B^2) \, \mathbb E_s\Tr(\rho_{s,B}^2) + \left( \mathbb E_{\phi\sim\mathsf{Haar}(ABM)}\Tr(\phi_B^2) \right)^2 \notag\\
&\le \mathbb E_{\phi\sim\mathsf{Haar}(ABM)} \bigl(\Tr(\phi_B^2)\bigr)^2 + 2\varepsilon_{\mathrm{PRS}}(n+m(n_A),\lambda,4) \notag\\ 
&\quad - 2 \left( \mathbb E_{\phi\sim\mathsf{Haar}(ABM)}\Tr(\phi_B^2) \right)^2 + 4 \mathbb E_{\phi\sim\mathsf{Haar}(ABM)}\Tr(\phi_B^2) \, \varepsilon_{\mathrm{PRS}}(n+m(n_A),\lambda,2) \notag\\
&\quad + \left( \mathbb E_{\phi\sim\mathsf{Haar}(ABM)}\Tr(\phi_B^2) \right)^2 \notag\\
&= \mathsf{Var}_{\phi\sim\mathsf{Haar}(ABM)} \!\left(\Tr(\phi_B^2)\right) + 2\varepsilon_{\mathrm{PRS}}(n+m(n_A),\lambda,4) \notag\\
&\quad + 4 \mathbb E_{\phi\sim\mathsf{Haar}(ABM)}\Tr(\phi_B^2) \, \varepsilon_{\mathrm{PRS}}(n+m(n_A),\lambda,2)\,.
\end{align}
Since $\Tr(\phi_B^2)\le 1$, we have
\begin{align}
&\mathbb E_s \left( \Tr(\rho_{s,B}^2) - \mathbb E_{\phi\sim\mathsf{Haar}(ABM)}\Tr(\phi_B^2) \right)^2 \notag\\
&\le \mathsf{Var}_{\phi\sim\mathsf{Haar}(ABM)} \!\left(\Tr(\phi_B^2)\right) + 2\varepsilon_{\mathrm{PRS}}(n+m(n_A),\lambda,4) + 4\varepsilon_{\mathrm{PRS}}(n+m(n_A),\lambda,2) \notag\\
&\le \frac{2}{d_A^4d_m^2} + 2\varepsilon_{\mathrm{PRS}}(n+m(n_A),\lambda,4) + 4\varepsilon_{\mathrm{PRS}}(n+m(n_A),\lambda,2)\,.
\label{aqc:eq:eff.mix.centered.purity}
\end{align}

We now choose the security parameters so that the PRS errors in \eqref{aqc:eq:eff.mix.centered.purity} are of the same order as the Haar variance. By Eq.~\eqref{aqc:eq.bound.QPRS.PRF}, applied with $Q\in\{2,4\}$ copies, and by the QPRF estimate following \Cref{aqc:lem:lwe.length.doubling.prg}, we have
\begin{align*}
\varepsilon_{\mathrm{PRS}}(n+m(n_A),\lambda,Q) \le& (1+\sqrt{2\pi})Q\sqrt{n+m(n_A)}\,2^{-\lambda/2} + 2\pi Q\,2^{-\lambda} + 2^{-\lambda} + 2^{-\Omega(\ell^\alpha)}\,.
\end{align*}
For $Q\in\{2,4\}$, the sum of the first three terms is bounded by
\begin{align*}
&\sum_{Q\in\{2,4\}} \left[ (1+\sqrt{2\pi})Q\sqrt{n+m(n_A)}\,2^{-\lambda/2} + 2\pi Q\,2^{-\lambda} + 2^{-\lambda} \right]\\
&\le \left(6(1+\sqrt{2\pi})+12\pi+2\right) \sqrt{n+m(n_A)}\,2^{-\lambda/2}\,,
\end{align*}
where we used $2^{-\lambda}\le 2^{-\lambda/2}$ and $n+m(n_A)\ge1$. Therefore, choosing $\lambda$ such that it also satisfies
\begin{align*}
\lambda=\lambda(n) \ge 8n_A+4m(n_A) + 2\log_2\!\left( 2\left(6(1+\sqrt{2\pi})+12\pi+2\right) \sqrt{n+m(n_A)} \right)
\end{align*}
gives
\begin{align*}
&\sum_{Q\in\{2,4\}} \left[ (1+\sqrt{2\pi})Q\sqrt{n+m(n_A)}\,2^{-\lambda/2} + 2\pi Q\,2^{-\lambda} + 2^{-\lambda} \right] \le \frac{1}{2d_A^4d_m^2}\,.
\end{align*}
Since $m(n_A)\le n_A$, this choice satisfies $\lambda(n)\in\poly(n)$.

Next choose $\ell=\ell(n)\in\poly(n)$ such that it also satisfies
\begin{align*}
& 2^{-\Omega(\ell(n)^\alpha)} \le \frac{1}{4d_A^4d_m^2}\,.
\end{align*}
This is possible because $\lambda(n)$, $n_A$, and $m(n_A)$ are polynomially bounded. Consequently,
\begin{align}
\varepsilon_{\mathrm{PRS}}(n+m(n_A),\lambda(n),2) + \varepsilon_{\mathrm{PRS}}(n+m(n_A),\lambda(n),4) \le \frac{1}{d_A^4d_m^2}\,.
\label{aqc:eq:eff.mix.constant.copy.prs.choice}
\end{align}
Substituting \eqref{aqc:eq:eff.mix.constant.copy.prs.choice} into
\eqref{aqc:eq:eff.mix.centered.purity}, we obtain
\begin{align}
&\mathbb E_s \left( \Tr(\rho_{s,B}^2) - \mathbb E_{\phi\sim\mathsf{Haar}(ABM)}\Tr(\phi_B^2) \right)^2 \notag\\
&\le \frac{2}{d_A^4d_m^2} + 4\left( \varepsilon_{\mathrm{PRS}}(n+m(n_A),\lambda(n),2) + \varepsilon_{\mathrm{PRS}}(n+m(n_A),\lambda(n),4) \right) \notag\\
&\le \frac{6}{d_A^4d_m^2}\,.
\label{aqc:eq:eff.mix.centered.purity.final}
\end{align}

By Markov's inequality applied to the squared deviation from $\delta = \frac{1}{2d_A d_m^{1/4}}$ and \eqref{aqc:eq:eff.mix.centered.purity.final}, we have that
\begin{align}
&\Pr_s\left[ \left( \Tr(\rho_{s,B}^2) - \mathbb E_{\phi\sim\mathsf{Haar}(ABM)}\Tr(\phi_B^2) \right)^2 \ge \frac{1}{4d_A^2d_m^{1/2}} \right] \notag\\
&\le 4d_A^2d_m^{1/2} \, \mathbb E_s \left( \Tr(\rho_{s,B}^2) - \mathbb E_{\phi\sim\mathsf{Haar}(ABM)}\Tr(\phi_B^2) \right)^2 \notag\\
&\le \frac{24}{d_A^2d_m^{3/2}} = \negl(n)\,.
\label{aqc:eq:eff.mix.centered.purity.prob}
\end{align}
Equivalently, with probability $1-\negl(n)$ over $s$,
\begin{align}
\left| \Tr(\rho_{s,B}^2) - \mathbb E_{\phi\sim\mathsf{Haar}(ABM)}\Tr(\phi_B^2) \right| \le \frac{1}{2d_A d_m^{1/4}}\,.
\label{aqc:eq:eff.mix.centered.purity.good}
\end{align}
Combining \eqref{aqc:eq:eff.mix.centered.purity.good} with
\eqref{aqc:eq:eff.mix.haar.purity.mean}, we obtain
\begin{align}
\Tr(\rho_{s,B}^2)
&\le \mathbb E_{\phi\sim\mathsf{Haar}(ABM)}\Tr(\phi_B^2) + \frac{1}{2d_A d_m^{1/4}} \notag\\
&\le \frac{1+d_m^{-1}}{d_A} + \frac{1}{2d_A d_m^{1/4}} \notag\\
&\le \frac{1+d_m^{-1/4}}{d_A}\,,
\label{aqc:eq:eff.mix.good.purity}
\end{align}
for all sufficiently large $n$ with probability $1-\negl(n)$ over $s$.

On the event \eqref{aqc:eq:eff.mix.good.purity}, the collision entropy lower bounds the von Neumann entropy, and therefore
\begin{align*}
H(B)_{\rho_{s,B}} &\ge H_2(B)_{\rho_{s,B}}= -\log\Tr(\rho_{s,B}^2)\\
&\ge n_A-\log(1+d_m^{-1/4})\,.
\end{align*}
Since $\rho_{s,AB}$ is purified by $M$,
\begin{align*}
H(AB)_{\rho_{s,AB}} = H(M)_{\psi_s} \le \log d_m = m(n_A)\,.
\end{align*}
Consequently,
\begin{align*}
H_{\min}(A^k|B^k)_{\rho_{s,AB}^{\otimes k}}
&\le H(A^k|B^k)_{\rho_{s,AB}^{\otimes k}}\\
&= k\,H(A|B)_{\rho_{s,AB}}\\
&= k\bigl(H(AB)_{\rho_{s,AB}}-H(B)_{\rho_{s,B}}\bigr)\\
&\le -k\bigl(n_A-m(n_A)\bigr) + k\log(1+d_m^{-1/4})\,.
\end{align*}
Finally, since $d_m=2^{m(n_A)}$, $m(n_A)=\omega(\log n_A)$, and
$k\in\poly(n)$,
\begin{align*}
    k\log(1+d_m^{-1/4}) = \negl(n)\,.
\end{align*}
Hence
\begin{align*}
H_{\min}(A^k|B^k)_{\rho_{s,AB}^{\otimes k}} \le -k\bigl(n_A-m(n_A)\bigr)+\negl(n)\,.
\end{align*}
Taking the union bound with the computational event, both estimates hold simultaneously with probability $1-\negl(n)$ over the key $s$, completing the proof.

\end{proof}

Thus, in the efficiently generated setting, we obtain that, with overwhelming probability over the key, the generated states have information-theoretic conditional min-entropy of order $-k(n_A-m(n_A))$, whereas their computational conditional min-entropy is almost maximal. The computational bound is slightly weaker than in the unconstrained mixed-state construction, since the error is $O(n^{-c})$ rather than negligible. This loss comes from the way the high-moment tail bound is applied, where we compare the overlap with $(1+n^{-c})/d_n^k$. Since $c>0$ is arbitrary, the error can be made smaller than any prescribed inverse polynomial. Thus the efficiently generated mixed states still exhibit an almost maximal separation between information-theoretic and computationally accessible correlations.

Let us finally discuss how the separations established in \Cref{aqc:th:lower-bound.eff} and \Cref{aqc:th:lower-bound.eff.mix} extend to the setting in which the full classical description of the circuits that prepare the states are available. In both the information-theoretic and computational settings, the recovery channel acts only on $B^k$. Nevertheless, the separation arises because the former allows arbitrary channels, whereas the latter restricts their circuit size. Although $(G,s)$ provides an efficient preparation of the state, it does not necessarily imply an efficient recovery channel attaining the information-theoretic bound. Indeed, as we have proven, with overwhelming probability over $s$, the computational bound holds simultaneously for every recovery channel of size at most $p(n)$, including any such channel chosen using $(G,s)$. In this sense, these results provide an analogue, for the one-sided recovery task defining $\CompMin(A|B)$, of the state-aware setting studied in \cite{LREJ25}: efficient preparation and full knowledge of how the state is generated do not necessarily provide efficient access to its correlations.

\end{document}